\documentclass[11pt]{article}
\usepackage{authblk}
\usepackage{booktabs}

\usepackage{booktabs}   
\usepackage{subcaption} 

\usepackage{hyperref}
\usepackage{etex}
 \usepackage{amsmath}
 \usepackage{amssymb}
 \usepackage{amsthm}
\usepackage{amsfonts}
 \usepackage{mathtools}
  \usepackage{algorithm2e}
  \usepackage{caption}
 \usepackage{xcolor}
 \usepackage{framed}
 \usepackage{todonotes}
 \usepackage{enumitem}
\usepackage{array}
\usepackage{algorithm2e}
\usepackage{pifont}
\usepackage{mathtools}
\usepackage{tikz}
\usetikzlibrary{shapes.multipart}

\newcommand{\state}{z}
\newcommand{\exec}{\rho}
\newcommand{\states}{Z}

\DeclarePairedDelimiter{\abs}{\lvert}{\rvert}
\newcommand{\integers}{\mathbb{Z}}
\newcommand{\Rats}{\mathbb{Q}}
\newcommand{\Reals}{\mathbb{R}}
\newcommand{\Nats}{\mathbb{N}}
 
 \newcommand{\sdom}{\mathsf{DOM}}
 \newcommand{\ptransf}{\Delta}
\newcommand{\Lap}[1]{\mathsf{Lap}{(#1)}}
\newcommand{\DLap}[1]{\mathsf{DLap}{(#1)}}

\newcommand{\pexp}[1]{\mathsf{Exp}{(#1)}}
\newcommand{\euler}{e}
\newcommand{\eulerv}[1]{\ensuremath{\euler^{#1}}}
\newcommand{\cU}{\mathcal{U}}
\newcommand{\cV}{\mathcal{V}}
\newcommand{\cW}{\mathcal{W}}

\newcommand{\cM}{\mathcal{M}}

\newcommand{\Prob}{\mathsf{Prob}}

\newcommand{\st}  {\mathbin{|}}

\newcommand{\blue}[1]{{#1}}
\newcommand{\brown}[1]{{#1}}
\newcommand{\set}[1]{\{#1\}}
\newcommand{\true}{\mathsf{true}}
\newcommand{\false}{\mathsf{false}}
\newcommand{\cA}{\mathcal{A}}

\newcommand{\cE}{\mathcal{E}}
\newcommand{\cD}{\mathcal{D}}
\newcommand{\cL}{\mathcal{L}}
\newcommand{\cX}{\mathcal{X}}
\newcommand{\cB}{\mathcal{B}}
\newcommand{\cZ}{\mathcal{Z}}
\newcommand{\cR}{\mathcal{R}}
\newcommand{\sF}{\mathcal{F}}
\newcommand{\mx}{\mathsf{max}}
\newcommand{\bv}{\mathsf{b}}
\newcommand{\dv}{\mathsf{x}}
\newcommand{\iv}{\mathsf{z}}
\newcommand{\rv}{\mathsf{r}}
\newcommand{\mbar} {\mathbin{|}}
\newcommand{\Ifs}{\textsf{if}}
\newcommand{\Thens}{\textsf{then}}
\newcommand{\eds}{\textsf{end}}
\newcommand{\Elses}{\textsf{else}}
\newcommand{\Whiles}{\textsf{While}}
\newcommand{\Dos}{\textsf{do}}
\newcommand{\ifstatement}[3]{\Ifs \, #1\, \Thens\, #2 \, \Elses\, #3\, \eds}
\newcommand{\whilestatement}[2]{\Whiles \, #1\, \Dos\, #2\, \eds }

\newcommand{\ext}{\textsf{exit}}
\newcommand{\qs}{q_s}
\newcommand{\qf}{q_f}
\newcommand{\sLabels}{\mathsf{Labels}}

\theoremstyle{definition}
\newtheorem*{remark}{Remark}
\newtheorem{definition}{Definition}
\newtheorem{example}[definition]{Example}
\newtheorem{theorem}[definition]{Theorem}
\newtheorem{lemma}[definition]{Lemma}
\newtheorem{corollary}[definition]{Corollary}
\newtheorem{proposition}[definition]{Corollary}

\newcommand{\chose}{\textsf{choose}}
\newcommand{\exit}{\textsf{exit}}
\newcommand{\seq}{\mathsf{seq}}
\newcommand{\Disc}{\mathsf{Disc}}
\newcommand{\pos}{\s{pos}}

\newcommand{\problemone} {Fixed Parameter Differential Privacy}
\newcommand{\problemtwo} {Differential Privacy}
\newcommand{\lngreal} {\cL_{\mathsf{exp}}}
\newcommand{\threal} {\mathsf{Th}_{\mathsf{exp}}}
\newcommand{\tuple}[1] {\langle #1 \rangle}
\newcommand{\s}[1] {\mathsf{#1}}
\newcommand{\fb}[1]{\mathbf{#1}}
\newcommand{\ourlang} {{\sf DiPWhile}}
\newcommand{\finlang} {{\sf Finite DiPWhile}}
\newcommand{\tool} {{\sf DiPC}}
\newcommand{\reachprob}[2] {\Prob(#1,#2)}
\newcommand{\sem}[1] {[\![#1]\!]}
\newcommand{\nxt} {\s{next}}
\newcommand{\pto} {\hookrightarrow}
\newcommand{\rmv}[1] {}

\newcommand{\cmark}{\ding{51}}%
\newcommand{\xmark}{\ding{55}}%
\newcolumntype{L}{>{\centering\arraybackslash}m{2.0cm}}
\newcolumntype{M}{>{\centering\arraybackslash}m{1.8cm}}
\newcolumntype{N}{>{\centering\arraybackslash}m{1.2cm}}
\makeatletter
\newcommand{\removelatexerror}{\let\@latex@error\@gobble}
\makeatother
\newcommand{\parta}{T1} 
\newcommand{\partb}{T2}

\usepackage{float}
\usepackage{graphicx}
\begin{document}

\title{Deciding Differential Privacy for Programs with Finite Inputs and Outputs}   


\author[1]{Gilles Barthe}
\affil[1]{MPI Security and Privacy and IMDEA Software Institute}

\author[2]{Rohit Chadha}
\affil[2]{University of Missouri}

\author[3]{Vishal Jagannath}
\affil[3,5]{University of Illinois at Urbana-Campaign}

%
\author[4]{A. Prasad Sistla}
\affil[4]{University of Illinois at Chicago}
%
\author[5]{Mahesh Viswanathan}

\maketitle
\begin{abstract}
  Differential privacy is a \emph{de facto} standard for statistical
  computations over databases that contain private data. Its main and
  rather surprising strength is to guarantee individual privacy and
  yet allow for accurate statistical results.
  Thanks to its mathematical definition, differential privacy is also
  a natural target for formal analysis. A broad line of work develops
  and uses logical methods for proving privacy. A more recent and
  complementary line of work uses statistical methods for finding
  privacy violations. Although both lines of work are practically
  successful, they elide the fundamental question of decidability.
  
  This paper studies the decidability of differential privacy. We first
 establish that checking differential privacy is undecidable even if one
  restricts to programs having a single Boolean input and a single
  Boolean output. Then, we define a non-trivial class of programs and
  provide a decision procedure for checking the differential privacy of a
  program in this class. Our procedure takes as input a program $P$
  parametrized by a privacy budget $\epsilon$ and either establishes the
  differential privacy for all possible values of $\epsilon$ or
  generates a counter-example. In addition, our procedure works for both
  to $\epsilon$-differential privacy and
  $(\epsilon,\delta)$-differential privacy. Technically, the decision
  procedure is based on a novel and judicious encoding of the
  semantics of programs in our class into a decidable fragment of the
  first-order theory of the reals with exponentiation. We implement
  our procedure and use it for (dis)proving privacy bounds for many
  well-known examples, including randomized response, histogram,
  report noisy max and sparse vector.

\end{abstract}





\maketitle

%

\section{Introduction}
\label{sec:intro}
Differential privacy~\cite{DMNS06} is a gold standard for the privacy of
statistical computations. Differential privacy ensures that running
the algorithm on any two \lq\lq adjacent\rq\rq\ databases yields two
\lq\lq approximately\rq\rq\ equal distributions, where two databases
are adjacent if they differ in a single element, and two distributions
are approximately equivalent if their distance is small w.r.t.\, some
metric specified by privacy parameter $\epsilon$ and error parameter
$\delta.$ Thus, differential privacy delivers a very strong form of
individual privacy. Yet, and somewhat surprisingly, it is possible to
develop differentially private algorithms for many tasks.  Moreover,
the algorithms are useful, in the sense that their results have
reasonable accuracy. However, designing differentially private
algorithms is difficult, and the privacy analysis can be error-prone,
as witnessed by the example of the sparse vector technique.

This difficulty has motivated the development of formal approaches for
analyzing differentially private algorithms (see~\cite{BGHP16} for a
survey and the related work section of this paper). Broadly,
two successful lines of work have emerged. The first line of work
develops sound proof systems to establish differential privacy and
uses these proof systems to prove the privacy of well-known and intricate examples~\cite{RP10,GHHNP13,BKOZ13,BGGHS16,BFGGHS16,ZK17,AGHK18,AH18,WDWKZ19}.
The second line of work searches for counter-examples to demonstrate
the violation of differential privacy~\cite{DingWWZK18,BichselGDTV18}.
Unfortunately, both lines of work elide the question of decidability.
As previous experience in formal verification suggests, understanding
decidable fragments of a problem not only help advance our theoretical
knowledge, but can form the basis of practical tools when combined with ideas like abstraction and composition.

The goal of this paper is, therefore, to study the decision problem
for differential privacy, and to make a first attempt at delineating
the decidability/undecidability boundary. As a first contribution, we
show that, as expected, checking differential privacy is
computationally undecidable. Our undecidability result holds even if
one restricts to programs having a single Boolean input and a single
Boolean output. Given the undecidability result, we then consider the
task of identifying a rich class of programs, that encompasses many
known examples, for which checking differential privacy nonetheless is
decidable.  We impose two desiderata:
\begin{enumerate}
\item the class of programs must include programs with real-valued
  variables, and more generally, with variables over infinite domains.
  This requirement is critical for the method to cover a broad class
  of differential privacy algorithms;
  
\item \blue{the programs themselves are parametrized by the privacy parameter $\epsilon$ (throughout the paper,
  we assume that the error parameter $\delta$ is a function of $\epsilon$),} and the decision procedure should decide privacy for all possible
  instances of the privacy parameter $\epsilon$. This requirement is motivated by the fact, supported
  by practice, that \blue{differential privacy algorithms are typically parametrized by $\epsilon$, and} well-designed algorithms are private not only for
  a single value of $\epsilon$, but typically for all positive values of
  $\epsilon$.
\end{enumerate}

We focus our attention on programs whose input and output spaces are
finite. Note that such programs need not be finite-state, as per our
first requirement, they could use program variables ranging over
infinite (even uncountable) domains to carry out the computation. We
introduce a class of programs, called {\ourlang}, which are
probabilistic while programs, for which the problem of checking
differential privacy is decidable. We succeed in carefully balancing
decidability and expressivity, by judiciously delineating the use of
real-valued and integer-valued variables. Intuitively, the main
restriction we impose is that these infinite-valued variables be used
only to directly influence the program control-flow and not the
data-flow that leads to the computation of the final output. 
\blue{More precisely, in an execution, the program output value depends only on the 
input, values sampled from user-defined distributions and the exponential mechanism, 
and branch conditions on the control flow path taken. 
The sampled values of real/integer variables affect only the branch conditions. 
Thus, the output values depend only on the branch conditions satisfied by the sampled real/integer variable values, but not on their actual sampled values.}
This
restriction, though severe, turns out to capture many prominent
differential privacy algorithms, including Report Noisy Max and Sparse
Vector Technique (see Section~\ref{sec:experiments} on experiments).

Key observations that enable us to establish decidability of
{\ourlang} programs are as follows. The first result is that the
semantics of {\ourlang}-programs can be defined using
parametrized, \emph{finite-state} Markov chains~\footnote{A
parametrized Markov chain is a Markov chain whose transition
probabilities are a function of the privacy budget.}. The fact that
the semantics is definable using only finitely-many states is a
surprising observation because our programs have both integer and
real-valued variables, and hence a na\"{i}ve semantics yields
uncountably many possible states.  Our crucial insight here is that a
precise semantics for {\ourlang}-programs is possible without tracking
the explicit values of the real and integer-valued variables. Since
real and integer variables are intuitively used only in influencing
control-flow, the semantics only tracks the symbolic relationships
between the variables. Second, we show that the transition
probabilities of the Markov chain are ratios of polynomial functions
in $\epsilon$ and $\eulerv \epsilon$, where $\euler$ is the Euler's
constant; this was a difficult result to establish. These two
observations together, allow us to reduce the problem of checking the
differential privacy of {\ourlang}-programs to the decidable fragment
of the first-order theory of reals with exponentials, identified by
McCallum and Weispfenning~\cite{mccallum2012deciding}.

We leverage our decision procedure to build a stand-alone tool for
checking $\epsilon$- or $(\epsilon, \delta(\epsilon))$-differential
privacy of mechanisms specified by {\ourlang}-programs, for all values
of $\epsilon$. We have implemented our decision procedure in a tool
that we call {\tool} ({\bf Di}fferential {\bf P}rivacy {\bf
C}hecker). Given {\ourlang}-program, our tool constructs a sentence
within the McCallum-Weispfenning fragment of the theory of reals with
exponentials. It then calls Mathematica\textregistered\ to check if
the constructed sentence is true over the reals. Since our decision
procedure is the \emph{first} that can both prove differential privacy
and detect its violation, we tried the tool on examples that known to
be differentially private and those that are known to be not
differentially private including variants of Sparse Vector, Report
Noisy Max, and Histograms. {\tool} successfully checked differential
privacy for the former class of examples and produced counter-examples
for the latter class. Our counter-examples are exact and are more
compact than those discovered by prior tools.

As a contribution of independent interest, we also demonstrate how our method
yields a theoretical complete under-approximation method for checking
differential privacy of programs with infinite output sets. For such
programs, it is possible to discretize the output domain into a finite
domain, and to use the decision procedure to find privacy violations
for the discretized algorithm (by post-processing, privacy violations
for the discretized algorithms are also privacy violations for the
original algorithm). The discretization yields a method for generating counter-examples for algorithms with infinite output sets. 

We briefly contrast our results with prior work, and refer the reader
to Section~\ref{sec:related-work} for further details. Overall, we see our decidability
results as complementary to prior works in checking differential
privacy. In general, existing methods for proving or disproving
differential privacy, although inherently incomplete due to the undecidability of checking 
differential privacy, 
are likely to be more efficient because they can
trade-off efficiency for precision. However, the decision procedures
for a sub-class of programs, like the one presented here, maybe more
predictable --- if a decision procedure fails to prove privacy, then it
shall produce a counter-example that demonstrates that the algorithm is
not differentially private. Moreover, counter-example search methods
work for a fixed ($\epsilon$) privacy parameter. As the  counter-example methods are usually statistical, they may generate both
false positives and false negatives. In contrast, our decision procedures
work for all values for the privacy parameter and do not generate
false positives or false negatives.


\paragraph{Contributions.}
We summarize our key contributions.
\begin{itemize}
\item We prove the undecidability of the problem of checking differential 
  privacy of very simple programs, including those that have a single
  Boolean input and output. Though unsurprising, undecidability has
  not been previously established in any prior work.

\item We prove the decidability of differential privacy for an interesting
 class of programs. Our method is fully automatic that can  check both
  differential privacy \emph{and} detect its violation by
  generating \emph{counter-examples}. To the best of our knowledge, this
  is the first such result that encompasses sampling from integer and real-valued variables. 

\item We implement the decision procedure and evaluate our approach on
private and non-private examples from the literature.
\end{itemize}

\brown
{Due to lack of space, some proofs and other materials have been moved to an Appendix. The Appendix  has been uploaded as an anonymous supplementary submission.} 


\section{Primer on differential privacy}
\label{sec:diffp-primer}

Differential privacy~\cite{DMNS06} is a rigorous definition and
framework for private statistical data mining. In this model, a
trusted curator with access to the database returns answers to queries
made by possibly dishonest data analysts that do not have access to
the database. The task of the curator is to return probabilistically
noised answers, so that data analysts cannot distinguish between two
databases that are adjacent, i.e.\, only differ in the value of a single
individual. There are two common definitions: two databases are
adjacent if they are exactly the same except for the presence or
absence of one record, or for the difference in one record. We
abstract away from any particular definition of adjacency.


Henceforth, we denote the set of real numbers, rational numbers,
natural numbers and integers by $\Reals,\Rats,\Nats$, and $\integers$
respectively.  The Euler constant shall be denoted by $\euler$.
We assume given a set $\cU$ of inputs, and a set $\cV$ of outputs. A
randomized function $P$ from $\cU$ to $\cV$ is a function that takes
an input in $\cU$ and returns a distribution over $\cV$. For a
measurable set $S\subseteq \cV$, the probability that the output of
$P$ on $u$ is in the set $S$ shall be denoted by $\Prob(P(u)\in
S)$. In the case the output set is discrete, we use $\Prob(P(u)=v)$ as
shorthand for $\Prob(P(u)\in \{ v\})$.

We are now ready to define differential privacy. We assume that $\cU$
is equipped with a binary \emph{symmetric} relation
$\Phi \subseteq \cU\times \cU$, which we shall call
the \emph{adjacency relation}. We say that $u_1,u_2\in \cU$
are \emph{adjacent} if $(u_1,u_2) \in \Phi$.
\begin{definition} Let $\epsilon\geq 0$ and $0\leq\delta\leq 1$.
Let $\Phi\subseteq \cU\times \cU$ be an adjacency relation. Let $P$ be
a randomized function with inputs from $\cU$ and outputs in $\cV$. We
say that $P$ is $(\epsilon,\delta)$-differentially private with
respect to $\Phi$ if for all measurable subsets $S\subseteq \cV$ and
$u,u' \in \cU$ such that $(u,u')\in \Phi$,
$$ \Prob(P(u)\in S) \leq \eulerv  \epsilon   \, \Prob(P(u')\in S) +\delta$$
As usual, we say that $P$ is $\epsilon$-differentially private iff it
is $(\epsilon,0)$-differentially private. If the output domain is
discrete, it is equivalent to require that for all $v\in \cV$ and
$u,u' \in \cU$ such that $(u,u')\in \Phi$,
$$ \Prob(P(u)=v) \leq \eulerv \epsilon \, \Prob(P(u')=v)$$
\end{definition}
Differential privacy is preserved by post-processing. Concretely, if
$P$ is an $(\epsilon,\delta)$-differentially private computation from
$\cU$ to $\cV$, and $h:\cV \rightarrow \cW$ is a deterministic
function, then $h\circ P$ is an $(\epsilon,\delta)$-differentially
private computation from $\cU$ to $\cW$. In the remainder, we shall
exploit post-processing to connect differential privacy of randomized
computations with infinite output spaces to differential privacy of
their discretizations.


%
%
     
 \paragraph{Laplace Mechanism.} The Laplace
mechanism~\cite{DMNS06} achieves differential privacy for numerical
computations by adding random noise to outputs. Given $\epsilon>0$ and
mean $\mu,$ let $\Lap{\epsilon,\mu}$ be the continuous distribution
whose probability density function (p.d.f.) is given by
 $$ f_{\epsilon,\mu}(x) = \frac {\epsilon} 2 \ \euler^{-\epsilon \abs
 {x -\mu}} .$$ $\Lap{\epsilon,\mu}$ is said to be the \emph{Laplacian
 distribution} with mean $\mu$ and scale parameter $\frac 1 \epsilon.$
 Consider a real-valued function $q:\cU\rightarrow \Reals$. Assume
 that $q$ is $k$-sensitive w.r.t.\, an adjacency relation $\Phi$ on
 $\cU$, i.e.\, for every pair of adjacent values $u_1$ and $u_2$,
 $\abs{q(u_1)- q(u_2)}\leq k$. Then the computation that maps $u$ to
 $\Lap{\frac\epsilon k,q(u)}$ is $\epsilon$-differentially private.

It is sometimes convenient to consider the discrete version of the
Laplace distribution. Given $\epsilon>0$ and mean $\mu,$ let
$\DLap{\epsilon,\mu}$ be the discrete distribution on $\integers$, the
set of integers, whose probability mass function (p.m.f.)  is
    $$ f_{\epsilon,\mu} (i) = \frac {1-\euler^{-\epsilon}}
    {1+\euler^{-\epsilon}} \ \euler^{-\epsilon \abs {i -\mu}} .$$
    $\DLap{\epsilon,\mu}$ is said to be the \emph{discrete Laplacian
    distribution} with mean $\mu$ and scale parameter $\frac
    1 \epsilon$.  The discrete Laplace mechanism achieves the same
    privacy guarantees as the continuous Laplace mechanism.

\paragraph{Exponential mechanism.}
The Exponential mechanism~\cite{McSherryT07} is used for making non-numerical
computations private.  The mechanism takes as input a value $u$ from
some input domain and a scoring function $F:\cU\times \cV \rightarrow
\Reals$ and outputs a discrete distribution over $\cV$. Formally, given
$\epsilon>0$ and $u\in \cU$, the discrete distribution $\pexp
{\epsilon, F,u}$ on $\cV$ is given by the probability mass function:
     $$h_{\epsilon, F,u} (v)=\frac{\euler^{\epsilon
     F(u,v)}}{ \sum_{v\in \cV} \euler^{\epsilon F(u,v)}} .$$

Suppose that the scoring function is $k$-sensitive w.r.t.\, some
adjacency relation $\Phi$ on $\cU$, i.e., for all for each pair of
adjacent values $u_1$ and $u_2$ and $v \in\cV$, $\abs{F(u_1,r)-
F(u_2,r)}\leq k$. Then the exponential mechanism is $(2k\epsilon,
0)$-differentially private w.r.t.\, $\Phi$.


\section{Motivating Example}
\label{sec:examples}

Before presenting the mathematical details of our results, let us
informally introduce our method by showing how it would work on an
illustrative example.

\paragraph{Sparse Vector Technique.} 
Several differential privacy examples require that the randomized algorithms sampling from infinite support distributions (including continuous distributions). 
The Sparse Vector Technique (SVT)~\cite{DNRRV09,lyu2016understanding}
was designed to answer multiple $\Delta$-sensitive numerical queries
in a differentially private fashion. The relevant information we want
from queries is, which amongst them are above a threshold $T$. 
The Sparse Vector Technique as
given in Algorithm~\ref{fig:SVT} is designed to identify the first $c$
queries that are above the threshold $T$ in an
$\epsilon$-differentially private fashion.

 \RestyleAlgo{boxed} 
 \begin{algorithm}
  \DontPrintSemicolon
\SetAlgoLined

 \KwIn{$q[1:N]$}
 \KwOut{$out[1:N]$}
 \;
$\rv_T \gets \Lap{ \frac {\epsilon} {2 \Delta} , T}$\;
 $count \gets 0$\;
  \For{$i\gets 1$ \KwTo $N$}
  {
    $\rv\gets \Lap{\frac \epsilon {4c\Delta} , q[i]}$\;
    $\bv\gets \rv \geq \rv_T$\;
    \uIf{$b$}{
      $out[i] \gets \top$\;
      $count \gets count + 1$\;
       \If{$count\geq c$} {exit} 
      }
    \Else{
      $out[i] \gets \bot$}
  }

%

\caption{SVT  algorithm (SVT1)}
\label{fig:SVT}
\end{algorithm}


In the program, the integer $N$ represents the total number of queries,
and the array $q$ of length $N$ represents the answers to queries. The
array $out$ represents the output array, $\bot$ represents False and
$\top$ represents True. We assume that initially the constant $\bot$
is stored at each position in $out$. In the SVT technique, the $\top$
answers account for most of the privacy cost, and we can only answer
$c$ of them until we run out of the privacy
budget~\cite{DNRRV09,ZK17}. On the other hand, there is no restriction
on the number of $\bot$ answers. \blue{Please observe that the SVT algorithm is parametrized by the privacy budget $\epsilon$. 
Thus, the SVT algorithm can be considered as representing a class of programs, one for each $\epsilon>0$.}

Given $N$, the input set $\cU$ in this context is the set of $N$ length vectors
$q$, where the $k$th element $q[k]$ represents the answer to the $k$th
query on the original database. The adjacency relation $\Phi$ on
inputs is defined as follows: $q_1$ and $q_2$ are adjacent if and only
if $\abs{q_1[i]-q_2[i]}\leq 1$ for each $1\leq i \leq N$.

Let us consider an instance of the SVT algorithm when $T=0,$ $N=2$,
$\Delta=1$ and $c=1$. Let us assume that all array elements in $q$
come from the domain $\set{0,1}$. In this case, we have four possible
inputs $[0,0],[0,1],[1,1]$, and $[1,0]$, and three possible outputs
$[\bot,\bot], [\top,\bot]$, and $[\bot, \top]$.

For example, the
probability of outputting $[\bot,\top]$ on input $[0,1]$ can be computed as follows. 
Let $X_T$ be a random variable with
Laplacian distribution $\Lap{\frac \epsilon 2,0}$, $X_1$ be a random variable with
Laplacian distribution $\Lap{\frac \epsilon 4,0}$ and $X_2$ be the random variable with Laplacian distribution $\Lap{\frac \epsilon 4,1}.$
 The
probability of outputting $[\bot,\top]$ is the product of outputting
of outputting $\bot$ first, which  is  $Prob(X_1<X_0)$,  and the conditional probability of outputting $\top$ given that $\bot$ is output, which is 
 $Prob(X_2\geq X_0 | X_1<X_0) $. Note that we really require the second quantity to be conditional probability as the events $X_1<X_0$
 and $X_2\geq X_0$ are \underline{not} independent.
This probability can be computed to be 
\[
r_1 \blue{(\epsilon)} = \frac{24 \euler^{\frac{3\epsilon} 4} - 1 +
 8 \eulerv{\frac \epsilon 4}\ + 21 \eulerv{\frac \epsilon 2}
 }{48 \euler^{\frac{3\epsilon} 4}}.
\]
Similarly, when the input is $[1,1]$ and the output is $[\bot,\top]$,
the probability is given by
\[
r_2 \blue{(\epsilon)}= \frac{-22 + 32 \eulerv{\frac \epsilon 4} -3 \epsilon
}{48 \euler^{\frac{\epsilon} 2}}.
\]

\blue{Observe that  $r_1(\epsilon)$ and $r_2(\epsilon)$ are functions of $\epsilon$, and hence the probabilities of outputting $[\bot,\top]$ on inputs $[0,1]$ and $[1,1]$ 
vary with $\epsilon$.}
Our immediate challenge is to automatically compute expressions like $r_1\blue{(\epsilon)},r_2\blue{(\epsilon)}$  
from the  given program, the adjacent inputs, and outputs.  Note that this example involves sampling from continuous distributions and is a function of $\epsilon.$ Nevertheless, we shall establish that (see Section~\ref{sec:dipwhile} and Theorem~\ref{thm:semantics}) that  for several programs, the former can be accomplished by interpreting the program as a finite-state DTMC whose transition probabilities 
are functions parameterized by $\epsilon$
even when the randomized choices involve infinite-support random variables. The set of programs that
we identify (Section~\ref{sec:dipwhile}) is rich enough to
model the most known differential privacy mechanisms when restricted to
finite input and output sets.

 Having computed such expressions,
checking $\epsilon$-differential privacy requires one to determine if
\blue{$$\begin{array}{ll}
\mbox{for all }&   \epsilon >0 .\ (r_1(\epsilon)\leq \eulerv{\epsilon} r_2(\epsilon)) \\
\mbox{and } \mbox{for all }&  \epsilon>0.\ (r_2(\epsilon)\leq \eulerv{\epsilon} r_1(\epsilon)).\\
\end{array}$$}
Note that the particular condition for the SVT example under consideration above is encodable as a first-order
sentence with exponentials, and thus checking the formula for the example
reduces to determining if such a first-order sentence
is valid for reals, with the standard interpretation of multiplication,
addition, and exponentiation. Whether there is a decision procedure
that can determine the truth of first-order sentences involving
real arithmetic with exponentials, is a long-standing open problem. However,
a decidable fragment of such an extended first-order theory
has been identified by McCallum and Weispfenning~\cite{mccallum2012deciding}. The formula for the considered example lies in this fragment.  Indeed, we can show that all the formulas for the SVT example lie in this fragment. This observation presents a challenge, namely, what guarantees do we have that checking differential privacy is reducible to this decidable fragment. Indeed, we shall establish that the set of formulas that arise from the class of programs with finite-state DTMC semantics in Theorem~\ref{thm:semantics} also lead to formulas in the same decidable fragment.

\begin{remark}
Notice that if one can compute expressions for the probability
producing individual outputs on a given input, we could  
also check $(\epsilon,\delta)$-differential privacy, instead
of just $\epsilon$-differential privacy. The only change would be to
account for $\delta$ in our constraints, and to consider all possible
subsets of outputs, instead of just individual output values. Thus,
the methods proposed here go beyond the scope of most automated
approaches, which are restricted to vanilla $\epsilon$-differential
privacy.
\end{remark}

\section{Preliminaries}
\label{sec:prelim}

In this section, we formally define the problem of differential privacy verification that we consider in this paper and also introduce the decidable fragment of 
real arithmetic with exponentiation 
that plays a crucial role in our decision procedure. The set of reals/positive reals/rationals/positive rationals shall be denoted by 
$\Reals$/$\Reals^{>0}$/$\Rats$/$\Rats^{>0}$ respectively.  


\subsection{The Computational Problem}
\label{sec:problem-def}

As illustrated by the example in Section~\ref{sec:examples}, a
differential privacy mechanism is typically a randomized
program $P_\epsilon$ parametrized by a variable $\epsilon$. Having a
parameterized program $P_\epsilon$ captures the fact that the program's
behavior depends on the privacy budget $\epsilon$, intending to
guarantee that $P_\epsilon$ is
$(f(\epsilon),g(\epsilon))$-differentially private, where $f$ and $g$
are some functions of $\epsilon$. The parameter $\epsilon$ is assumed
to belong to some interval $I \subseteq \Reals^{>0}$ \blue{with rational end-points}; usually, we take
$\epsilon$ to just belong to the interval $(0,\infty)$. The program
$P_\epsilon$ shall be assumed to terminate with probability 1 for every
value of $\epsilon$ (in the appropriate interval).

The  randomized program  $P_\epsilon$
takes inputs from a set $\cU$ and produces output in a set
$\cV$. In this paper, we shall assume that both $\cU$ and $\cV$
are \emph{finite} sets that can be effectively enumerated. Despite our
restriction to finite input and output sets, the computational problem of checking
differential privacy is challenging (see Section~\ref{sec:undec}). At the same time, the decidable
subclass we identify (Section~\ref{sec:dipwhile}) is rich enough to
model most  differential privacy mechanisms when restricted to 
finite input and output sets. Extending our decidability results to
subclasses of programs that have infinite input and output sets, is a
non-trivial open problem at this time.

The computational problems we consider in this paper are as
follows. Since our programs take inputs from a finite set $\cU$, we
assume that the adjacency relation $\Phi \subseteq \cU \times \cU$ is
given as an explicit list of pairs. In general, when discussing
$(\epsilon,\delta)$-differential privacy of some mechanism, the
error parameter $\delta$ needs to be a function of $\epsilon$. To
define the computational problem of checking differential privacy, the
function $\delta: \Reals^{> 0} \to [0,1]$ must be given as input. We,
therefore, assume that this function $\delta$ has some finite
representation; if $\delta$ is the constant \blue{$\delta_0$} (which is often the
case), then we represent $\delta$ simply by the number \blue{$\delta_0$}. 
There are
two computational problems we consider in this paper.
\begin{description}
\item[\problemone] Given a program $P_\epsilon$ over inputs $\cU$ and 
  outputs $\cV$, adjacency relation $\Phi \subseteq \cU \times \cU$,
  and positive rational numbers $\epsilon_0,\delta_0,t \in \Rats^{>0}$, determine if
  $P_{\epsilon_0}$ is $(t\epsilon_0,\delta_0)$-differentially private
  with respect to $\Phi$.
\item[\problemtwo] Given a program $P_\epsilon$ over inputs $\cU$ and 
  outputs $\cV$, interval $I \subseteq \Reals^{>0}$ \blue{with rational end-points},
  $\delta:\Reals^{>0} \to [0,1]$, an adjacency relation
  $\Phi\subseteq \cU\times \cU$, and a rational number $t\in \Rats^{>0}$,
  determine if $P_\epsilon$ is
  $(t\epsilon,\delta(\epsilon))$-differentially private with respect
  to $\Phi$ \emph{for every} $\epsilon \in I$.
\end{description}

Observe that the {\problemone} problem can be trivially reduced to the
{\problemtwo} problem \blue{by considering the singleton interval $I=[\epsilon_0,\epsilon_0]$ and $\delta(\epsilon)=\delta_0$, where the goal is to check fixed parameter differential privacy for constant privacy budget $\epsilon_0$ and error parameter $\delta_0$. } Thus, an algorithm for checking {\problemtwo}
can be used to solve {\problemone}. Unfortunately, the {\problemone}
problem is extremely challenging even when restricted to finite input and output sets--- we show that it is
undecidable (Section~\ref{sec:undec}), and therefore, so is the
{\problemtwo} problem. We shall identify a class of programs
(Section~\ref{sec:dipwhile}) for which the {\problemtwo} problem (and
therefore the {\problemone} problem) is decidable.

When the differential privacy 
does not hold, we would like to output a counter-example.   

\begin{definition}
A counter-example of $(\epsilon,\delta)$ differential privacy for  $P_\epsilon$, with respect to an adjacency
relation $\Phi$, a function $\delta:\Reals^{>0} \to [0,1]$ and a value
$t\in\Rats^{>0}$, is a quadruple $(\fb{in},\fb{in'},O,\epsilon_0)$ such that
$(\fb{in},\fb{in'})\in \Phi$, $O\subseteq \cV$ and $\epsilon_0>0$ and 
$$ \Prob(P_{\epsilon_0}(\fb{in})\in O) > \eulerv {t \epsilon_0}   \, \Prob(P(\fb{in'})\in O) +\delta(\epsilon_0)$$
When $\delta$ is the constant function $0$, then $O$ is $\{\fb{out}\}$ for some $\fb{out}\in \cV.$
\end{definition}

\begin{remark}
For the rest of the paper, \blue{unless otherwise stated}, we shall assume that the interval  $I \subseteq \Reals^{>0}$ that contains the set of admissible $\epsilon$s is the interval $(0,\infty).$
\blue{In our paper, $\epsilon$ refers to the parameter in program $P_\epsilon$, and not the privacy budget. In our case, the privacy budget is $(t \epsilon)$. For example, some differential privacy algorithms $P_\epsilon$ are designed to satisfy $(\frac \epsilon 2,0)$-differential privacy, and so in this case $t$ would be $\frac 1 2$. In the standard differential privacy definition ``$\epsilon$'' refers to the privacy budget and so $t$ does not appear. However, many theorems for differential privacy algorithms use ``$\epsilon$'' as the program parameter, and then the privacy theorem is stated as the program being $(t \epsilon,\delta)$-differentially private. In most such cases, such a theorem is equivalent to saying that the program $P_{\frac \epsilon t}$ (obtained by replacing $\epsilon$ by $\frac \epsilon t$) is $(\epsilon, \delta(\frac \epsilon t))$-differentially private.}

\end{remark}

\subsection{Reals with exponentials}
\label{sec:decidable-reals}

As outlined in Section~\ref{sec:examples}, our
approach towards deciding differential privacy shall rely on reducing
the question to  the problem of checking the truth of a first-order sentence for the
reals. Because of the definition of differential privacy, the
constructed first-order sentence shall involve exponentials. It is a
long-standing open problem whether there is a decision procedure for
the first-order theory of reals with exponentials. However, some
fragments of this theory are known to be decidable. In particular,
there is a fragment identified by McCallum and
Weispfenning~\cite{mccallum2012deciding}, that we shall exploit in our
results.

We will consider first-order formulas over a restricted signature and
vocabulary. We will denote this collection of formulas as the language
$\lngreal$. Formulas in $\lngreal$ are built using variables
$\set{\epsilon} \cup \set{x_i\: |\: i \in \Nats}$, constant symbols
$0,1$, unary function symbol $\eulerv{(\cdot)}$ applied only to the
variable $\epsilon$, binary function symbols $+,-,\times$, and binary
relation symbols $=, <$. The terms in the language are integral
polynomials with rational coefficients over the variables
$\set{\epsilon} \cup \set{x_i\: |\:
i \in \Nats} \cup \set{\eulerv{\epsilon}}$. Atomic formulas in the
language are of the form $t = 0$ or $t < 0$ or $0 < t$, where $t$ is a
term. Quantifier free formulas are Boolean combinations of atomic
formulas. Sentences in $\lngreal$ are formulas of the form
\[
Q\epsilon Q_1x_1\cdots Q_nx_n \psi(\epsilon,x_1,\ldots, x_n)
\]
where $\psi$ is a quantifier free formula, and $Q$, $Q_i$s are
quantifiers. In other words, sentences are formulas in prenex form,
where all variables are quantified, and the outermost quantifier is
for the special variable $\epsilon$.

The theory $\threal$ is the collection of all sentences in $\lngreal$
that are valid in the structure $\langle \Reals, 0, 1, \eulerv{(\cdot)},$
$+, -, \times, =, <\rangle$, where the interpretation for $0,1,+,-,\times$ is
the standard one on reals, and $\euler$ is Euler's constant; notice
that this is an extension of the first-order theory of reals. The
crucial property about this theory is that it is decidable.
\begin{theorem}[McCallum-Weispfenning~\cite{mccallum2012deciding}]
$\threal$ is decidable. 
\end{theorem}

Finally, our tractable restrictions (and our proofs of decidability)
shall often utilize the notion of functions \emph{definable} in $\threal$; we,
therefore, conclude this section with its formal definition.

\begin{definition}
\label{def:definable-func}
A function $f : (0,\infty) \to \Reals$ is said to
be \emph{definable} in $\threal$, if there is a formula
$\varphi_f(\epsilon,x)$ 
in $\lngreal$ with two free variables
($\epsilon$ and $x$) such that
$$\begin{array}{l}
\mbox{\blue{for all }} \blue{a\in (0,\infty)}.\ f(a) = b \mbox{ iff } \\
 \ \ \tuple{\Reals,0, 1, \eulerv{(\cdot)}, +,
-, \times, =, <} \models \varphi_f(\epsilon,x) [\epsilon \mapsto a,
x \mapsto b]
\end{array}$$
\end{definition}


%

\section{Program syntax and semantics}
\label{sec:simple}
\newcommand{\disc}{\s{disc}}
We consider randomized algorithms written as simple
probabilistic while programs. We introduce the syntax of these 
programs, along with their ``natural'' semantics given using Markov
kernels~\cite{Cinlar-Book,Prakash-Markov}. We show that the problem of checking differential privacy is
undecidable for these programs. 
\begin{figure}
\begin{framed}
\raggedright

Expressions ($\bv\in \cB, \dv\in \cX, \iv\in \cZ, \rv\in \cR, d\in \sdom, i \in \integers, q\in \Rats, g\in \sF_{Bool}, f\in \sF_\sdom$):
$$\begin{array}{lcl}
B&::=& \true \mbar \false \mbar \bv \mbar not(B) \mbar B\ and\ B \mbar B \ or \ B \mbar g(\tilde{E})\\
E&::=&   d \mbar \dv  \mbar f (\tilde{E}) \\
Z&::=& \iv \mbar  i Z  \mbar E Z \mbar Z+Z \mbar Z+i \mbar Z+ E \\
R&::=&  \rv \mbar  q R \mbar E R \mbar R+R \mbar R+q \mbar R+ E\\
\end{array}$$

Basic Program Statements ($a\in \Rats^{>0}$, $\sim\in \set{<,>,=, \leq, \geq}$, $F$ is  a scoring function and $\chose$ is a user-defined distribution):
$$\begin{array}{lll}
 s&::= &      \dv\gets E \mbar     \iv\gets Z \mbar     \rv\gets R \mbar   \bv\gets B \mbar   \bv\gets Z_1\sim Z_2 \mbar  \\
          &&        \bv\gets Z \sim E \mbar       \bv\gets  R_1\sim R_2 \mbar      \bv\gets  R \sim E \mbar \\ 
         &&         \rv\gets\Lap{a\epsilon,E} \mbar     \iv\gets\DLap{a\epsilon,E} \mbar \\
          &&       \dv\gets\pexp {a\epsilon, F(\tilde{\dv}), E} \mbar     \dv\gets \chose(a\epsilon, \tilde{E}) \mbar \\
          &&      \ifstatement B P P \mbar   \whilestatement B P \mbar     \ext\\

\end{array}$$

Program Statements ($\ell \in \mathsf{Labels}$)
$$\begin{array}{lcl}
 P&::= &  \ell:\ s \mbar  \ell:\ s\, ;\,  P\\
\end{array} $$
%


\end{framed}
 
 \caption{\footnotesize{BNF grammar for $\s{Simple}$. $\sdom$ is a finite discrete domain. $\sF_{Bool}$, ($\sF_{\sdom}$ resp) are set of functions that output Boolean values ($\sdom$ respectively).   $\cB,\cX, \cZ, \cR$ are the sets of Boolean variables, $\sdom$ variables, integer random  variables and real random variables. 
 $\sLabels$ is a set of program labels.  For a syntactic class $S$, $\tilde{S}$ denotes a sequence of elements from $S$. {\ourlang} (see Section~\ref{sec:dipwhile}) is the subclass of programs in which the assignments to real and integer variables do not occur with the scope of a while statement.
 }}
\label{fig:BNFFull}
\end{figure}

\subsection{Syntax of $\s{Simple}$ programs}
\label{sec:simple-syntax}

We introduce a class of programs we call $\s{Simple}$. Programs in  $\s{Simple}$ are probabilistic while programs in which variables can be assigned values by drawing from distributions typically used in differential privacy algorithms. Programs in $\s{Simple}$ obey some syntactic restrictions; these syntactic restrictions are introduced to make it easier to describe the decidable fragment in
Section~\ref{sec:dipwhile}. Despite these restrictions, the problem of
checking differential privacy is undecidable for the language
introduced here.

The formal syntax of $\s{Simple}$ programs is shown in
Figure~\ref{fig:BNFFull}. Programs have four types of variables: $Bool
= \{\true$, $\false\}$; finite domain $\sdom$~\footnote{Though 
  not necessary to distinguish between Booleans and finite domains,
  having such a distinction makes our future technical development
  easier.} that we assume (without loss of generality) to be
$\set{-N_\mx, \ldots 0, 1, \ldots N_\mx}$, a finite subset of
integers~\footnote{Our decidability results also hold if $\sdom$ is
  taken to be a finite subset of the rationals.}; reals $\Reals$; and
integers $\integers$. The set of Boolean/$\sdom$/ integer/real program
variables are respectively denoted by $\cB$/$\cX$/$\cZ$/$\cR$. The set
of Boolean/$\sdom$/integer/real expressions is given by the
non-terminal $B/E/Z/R$ in Figure~\ref{fig:BNFFull}. We now explain the
rules for such expressions. Boolean expressions ($B$) can be built
using Boolean variables and constants, standard Boolean operations,
and by applying functions from $\sF_{Bool}$. $\sF_{Bool}$ is assumed
to be a collection of \emph{computable} functions returning a
$Bool$. We assume that  $\sF_{Bool}$ always contains a function $\s{EQ}(x_1,x_2)$ that returns $\true$ iff $x_1$ and $x_2$ are equal. 
$\sdom$ expressions ($E$) are similarly built from $\sdom$
variables, values in $\sdom$, and applying functions from set of
computable functions $\sF_\sdom$. Next, integer expressions ($Z$) are
built using multiplication and addition with integer constants and
$\sdom$ expressions, and additions with other integer
expressions. Finally, real expressions ($R$) are built using
multiplication and addition with rational constants and $\sdom$
expressions, and additions with other real-valued expressions. Notice
that integer-valued expressions cannot be added or multiplied, in
real-valued expressions; this syntactic restriction shall be useful
later.

A program in $\s{Simple}$ is a triple consisting of a set of (private)
input variables, a set of (public) output variables, and a finite
sequence of labeled statements (non-terminal $P$ in
Figure~\ref{fig:BNFFull}). The private input variables and public
output variables take values from the domain $\sdom$. Thus, the set of
possibles inputs/outputs ($\cU$/$\cV$), is identified with the set of
valuations for input/output variables; a valuation over a set of
variables $X'=\set{\dv_1,\dv_2,\ldots, \dv_m} \subseteq \cX$ is a
function from $X'$ to $\sdom$.  Note that if we represent the set $X'$
as a sequence $\dv_1,\dv_2,\ldots, \dv_m$ then a valuation $val$ over
$X'$ can be viewed as a sequence $ val(\dv_1),val(\dv_2),\ldots,
val(\dv_m)$ of $\sdom$ elements.

We assume every statement in our program is uniquely labeled from a
set of labels called $\s{Labels}$. Basic program statements
(non-terminal $s$) can either be assignments, conditionals, while
loops, or $\ext$. Statements other than assignments are
self-explanatory. The syntax of assignments is designed to follow a
strict discipline. Real and integer variables can either be assigned
the value of real/integer expression or samples drawn using the
Laplace or discrete Laplace mechanism. $\sdom$ variables are either
assigned values of $\sdom$ expressions or values drawn either using an
exponential mechanism ($\pexp {a\epsilon, F(\tilde{\dv}), E}$) or a
user-defined distribution ($\chose(a\epsilon,\tilde{E})$). For the
exponential mechanism, we require that the scoring function $F$ be
computable and return a rational value. Both of these restrictions are
unlikely to be severe in practice. 
%
%
%
In the case of the user defined distribution, we demand that the
probability with which a value \blue{$d$} in $\sdom$ is chosen (as a function of
the privacy budget $\epsilon$), be definable in $\threal$, 
\blue{and that there is an algorithm that on input $a, \tilde{d},v$ returns the formula 
defining the probability of sampling $d\in \sdom$ from the distribution $\chose(a\epsilon,\tilde{d})$ 
where $\tilde{d}$ is a sequence of values from 
$\sdom$}. This
restriction is exploited in Section~\ref{sec:dipwhile} to get
decidability for a sub-fragment.

Finally, we consider assignments to Boolean variables. The interesting
cases are those where the Boolean variable stores the result of the
comparison of two expressions. The syntax does not allow for comparing
real and integer expressions. This restriction is exploited later in
Section~\ref{sec:dipwhile} when the decidable fragment is identified.
Finally, we will assume that in any execution, if a variable appears
on the right side of an assignment statement, then it should have been
assigned a value before. This assumption is not restrictive but is
technically convenient when defining the semantics for programs.

\subsection{Markov Kernel Semantics}
\label{sec:simple-sem}

We briefly sketch a ``natural'' semantics for
$\s{Simple}$ using Markov kernels. A key step in proving our
decidability result is to define a semantics using finite-state
(parametrized) DTMCs for the sub-fragment {\ourlang} defined in
Section~\ref{sec:dipwhile}. The DTMC semantics may not seem natural on
first reading. The point of the semantics in this section is,
therefore, to argue the correctness of our decision procedure on the basis
of the equivalence of these two semantics for {\ourlang}
(Sections~\ref{sec:dipwhile} and~\ref{sec:decidability}). Details for this section are given in Appendix~\ref{app:semantics}
due of space constraints and because understanding
this semantics is not critical to our decidability proof.

\blue{Given a fixed $\epsilon>0$, the states in the Markov kernel-based semantics for a
program $P_\epsilon$}  will be of the form $(\ell, h_{Bool},
h_{\sdom}, h_{\integers}, h_{\Reals})$, where $\ell$ is the label of
the statement of $P_\epsilon$ to be executed next, the 
functions $h_{Bool}$, $h_{\sdom}$, $h_{\integers}$ and $h_{\Reals}$
assign values to the Boolean, $\sdom$, real and integer variables of the 
program $P_\epsilon$
respectively. Given  an input state $\fb{in}$, the initial state will
correspond to one where $\sdom$-valued input variables get the values
given in $\fb{in}$, and all other variables either get $\false$ or $0$,
depending on their type. Observe
that for a program $P_\epsilon$ with $k$ program statements, $i$ Boolean variables, $j$ $\sdom$ variables, $s$ integer variables, $t$ real variables
a state  $(\ell, h_{Bool},
h_{\sdom}, h_{\integers}, h_{\Reals})$ can be uniquely identified with an element of the set $D_{P_\epsilon}= \set{1,\ldots,k}\times\sF_{Bool}^i\times \sdom^j \times \integers^s\times\Reals^t.$ The ``natural" Borel $\sigma$-algebra on $D_{P_\epsilon}$ induces a $\sigma$-algebra on the states of $P_\epsilon.$


The semantics of $\s{Simple}$ programs can be defined as a Markov
kernel over this $\sigma$-algebra on states. Intuitively, the Markov
kernel $K_\epsilon$ corresponding to a program $P_\epsilon$ is such
that for a state $s$ and a measurable set of states $C$,
$K_{\epsilon}(s,C)$ is the probability of transitioning to a state in
$C$ from $s$. The precise definition of this Markov kernel is in
Appendix~\ref{app:semantics}.

Executions are just sequences of states, and the $\sigma$-field on
executions is the product of the $\sigma$-field on
states. The Markov kernel defines a probability measure on this
$\sigma$-field. Given all these observations, we take
$\Prob_{natural}( P_\epsilon(\fb{in}) = \fb{out})$ to denote the
probability (as defined by the Markov kernel of $P_\epsilon$) of the
set of all executions that start in the initial state
corresponding to $\fb{in}$ and end in an exit state with $\fb{out}$ as
the valuation of output variables\brown{; the precise definition is  in
Appendix~\ref{app:semantics}}. 
For the rest
of the paper, we will assume that our programs terminate with
probability $1.$

\subsection{Undecidability}
\label{sec:undec}
The problem of checking differential privacy for $\s{Simple}$ programs
is undecidable.
\begin{theorem}
\label{thm:undecidability}
The {\problemone} problem  and the  {\problemtwo}  problem for programs $P_\epsilon$ in
$\s{Simple}$ is undecidable.
\end{theorem}
The proof of Theorem~\ref{thm:undecidability} reduces
the non-halting problem for deterministic 2-counter Minsky
machines 
to the {\problemone}
problem. More precisely, we show that given a 2-counter Minsky machine
$\cM$ (with no input), there is a program
$P^{\cM}_\epsilon \in \s{Simple}$ such that 
\begin{itemize}
\item $P^{\cM}_\epsilon$ has only one input $\dv_{\fb{in}}$ 
  and one output $\dv_{\fb{out}}$ taking values in $\sdom=\set{0,1}$;
\item $P^{\cM}_\epsilon$ terminates with probability $1$ for all $\epsilon 
  \in \Reals^{>0}$;
\item $P^\cM_{\epsilon}$ is $(\epsilon,0)$-differentially private 
  with respect to the adjacency relation
  $\Phi=\{(0,1),$ $(1,0)\}$ if and only
  if $\cM$ does not halt.
\end{itemize}
This construction shows that {\problemtwo} is
undecidable. Undecidability of {\problemone} is obtained by taking
$\epsilon$ to be any constant rational number, say $\frac 1 2.$ \brown{The
formal details of the reduction 
are in Appendix~\ref{app:undecid}.}

\section{{\ourlang}: A decidable class of programs}
\label{sec:dipwhile}
We now discuss a restricted class of programs, for which we can establish decidability of checking differential privacy.  The class of programs that we consider are exactly those programs in $\s{Simple}$ that satisfy the following restriction: 
\begin{description}
\item[Bounded Assignments] We do not allow assignments to real and
  integer variables within the scope of a while loop. This restriction ensures
  that assignments to such variables happen only a \emph{bounded}
  number of times during execution. Thus, without loss of
  generality, we assume that real and integer variables are
  assigned \emph{at most once} as a program with multiple assignments to a single real and  variables can always be rewritten to an equivalent program with each assignment to a variable  being an assignment to a fresh variable.
\end{description}
We refer to this restricted class as {\ourlang}.  The {\ourlang}
language is surprisingly expressive --- many known randomized
algorithms for differential privacy can be encoded. We give an example
of such encodings in {\ourlang}. We omit labels of program statements
unless they are needed.

\begin{example}
\label{ex:svtlang}
Algorithm~\ref{fig:SVTL} shows how SVT can be encoded in our
language with $T=0,\Delta=1, N=2, c=1. $ In the example we are
modeling $\bot$ by $0$ and $\top$ by $1.$ Though for-loops are not
part of our program syntax, they can modeled as while loops, or if
bounded (like here), they can be unrolled.
\end{example}

\begin{figure*}
\begin{minipage}[T]{0.4\textwidth}
\IncMargin{0.5em} 

 \RestyleAlgo{boxed} 
\removelatexerror
\begin{algorithm}[H]
  \DontPrintSemicolon

 \KwIn{$q_1,q_2$}
 \KwOut{$out_1,out_2$}
 \;

 \SetKwSty{textsfsf}
 \nl $T\gets 0$;\;
 \nl  $out_1\gets 0$;\;
 \nl  $out_2\gets 0$;\;
 \nl  $\rv_T \gets \Lap{ \frac {\epsilon} {2 } , T};$\;
 \nl   $\rv_1\gets \Lap{\frac \epsilon {4} , q_1}$;\; 
  \nl $\bv\gets \rv_1 \geq \rv_T$;\;
 
  \nl \uIf{$\bv$}{
   \nl    $out_1 \gets 1$\;
      }   
     \Else
     { 
\nl   $\rv_2 \gets \Lap{\frac \epsilon {4} , q_2}$;\;
 \nl  $\bv\gets \rv_2 \geq \rv_T$;\;
 \nl  \If{$\bv$}{
  \nl     $out_2 \gets 1$\;    
      }   
      }
  \nl    $\ext$\;
\caption{\footnotesize{SVT  for $1$-sensitive queries with $N=2$,$c=1$ and $T=0.$ The numbers at the beginning of a line indicate the label of the statement.}}
\label{fig:SVTL}
\end{algorithm}
\DecMargin{0.5em}
\end{minipage}
\qquad
\begin{minipage}[T]{0.6\textwidth}
\setlength{\tabcolsep}{1pt}
\begin{center}
\begin{tikzpicture}[
  every text node part/.style={align=center},
  state/.style={draw,rounded corners,minimum width=4.1cm},
  trans/.style={above,pos=0.65}]
\footnotesize
\node(d1) at (0,8.2) {$\vdots$};
\node[state](s5) at (0,7) {%
        \begin{tabular}{ll}
        {\bf 9}: & $q_1:u$, $q_2:v$, $T:0$, \\
        & $out_1:0$, $out_2:0$, $b: \bot$\\
        & $r_T: (\frac{1}{2},0)$ $r_1: (\frac{1}{4},u)$\\
        & $r_1 < r_T$
        \end{tabular}};
\node[state](s6) at (0,4) {%
        \begin{tabular}{ll}
        {\bf 10}: & $q_1:u$, $q_2:v$, $T:0$, \\
        & $out_1:0$, $out_2:0$, $b: \bot$\\
        & $r_T: (\frac{1}{2},0)$ $r_1: (\frac{1}{4},u)$ 
        $r_2: (\frac{1}{4},v)$\\
        & $r_1 < r_T$
        \end{tabular}};
\node[state](s71) at (-2.5,1) {%
        \begin{tabular}{ll}
        {\bf 11}: & $q_1:u$, $q_2:v$, $T:0$,\\
        & $out_1:0$, $out_2:0$, $b:\top$\\
        & $r_T: (\frac{1}{2},0)$ $r_1: (\frac{1}{4},u)$
        $r_2: (\frac{1}{4},v)$\\
        & $r_1 < r_T$, $r_2 \geq r_T$
        \end{tabular}};
\node[state](s72) at (2.5,1) {%
        \begin{tabular}{ll}
        {\bf 11}: & $q_1:u$, $q_2:v$, $T:0$, \\
        & $out_1:0$, $out_2:0$, $b:\bot$\\
        & $r_T: (\frac{1}{2},0)$ $r_1: (\frac{1}{4},u)$
        $r_2: (\frac{1}{4},v)$\\
        & $r_1 < r_T$, $r_2 < r_T$
        \end{tabular}};
\node(d21) at (-2.5,0) {$\vdots$};
\node(d22) at (2.5,0) {$\vdots$};
\draw[->] (s5) -- node[left]{1} (s6);
\draw[->] (s6) -- node[trans]{$p$} (s71);
\draw[->] (s6) -- node[trans]{$q$} (s72);
\end{tikzpicture}
\end{center}
\caption{\footnotesize{Partial DTMC semantics of Algorithm~\ref{fig:SVTL} showing 
the steps when lines 9 and 10 are executed. $q_1$ and $q_2$ are assumed
to have values $u$ and $v$, respectively. Only values of  assigned
program variables is shown. Third line in state shows parameters for
the real values that were sampled. Last line shows the accumulated set
of Boolean conditions that hold on the path.}}
\label{fig:SVT-dtmc}
\end{minipage}
\end{figure*}

Appendix~\ref{app:example_dipwhile} shows how sampling from the standard exponential distribution can be encoded in {\ourlang}.
Other examples that can be encoded in our language (and for which the
decision procedure applies) include randomized response, 
the private smart sum
algorithm~\cite{CSS10} with finite discretization of the output space (See~\ref{sec:findis}), and private vertex cover~\cite{GLMRT10}.

The decidability of checking differential privacy for {\ourlang} shall rely on two observations.
 First, the semantics of
{\ourlang} programs can also be defined as finite-state discrete-time
Markov chains (DTMC), albeit with transition probabilities parameterized by $\epsilon$. This observation is surprising because
{\ourlang} programs have real and integer values variables, and so the
natural semantics has uncountably many
states (See Section~\ref{sec:simple-sem}). The key insight in establishing this observation is that an
equivalent semantics of {\ourlang} programs can be defined without
explicitly tracking the values of real and integer-valued
variables. Second, all the transition probabilities arising in our
semantics are definable in $\threal$. These two observations allow us
to 
 to establish decidability of
checking differential privacy of {\ourlang} programs. The rest of the section is devoted to establishing these observations. We start by formally defining  \emph{ parametrized DTMCs}. 
\subsection{Parameterized DTMCs}
\begin{definition}
\label{def:pDTMC}
A \emph{parametrized DTMC} is a pair $\blue{\cD} =
(\states,\ptransf)$, where $\states$ is a (countable) set of states,
and $\ptransf: \states \times \states \to (\Reals^{>0} \to [0,1])$ is
the \emph{probabilistic transition function}. For any pair of states
$\state,\state'$, $\ptransf$ returns a function from $\Reals^{>0}$ to
$[0,1]$, such that for every $\epsilon > 0$, $\sum_{\state' \in
  \states} \ptransf(\state,\state')(\epsilon) = 1$. We shall call
$\ptransf(\state,\state')$ as the probability of transitioning from
$\state$ to $\state'$.
\end{definition}

A \emph{definable} parametrized DTMC is a parametrized DTMC
$\blue{\cD} = (\states,\ptransf)$ such that for every pair of states
$\state,\state' \in \states$, the function $\ptransf(\state,\state')$
is definable in $\threal$.

A parametrized DTMC associates with each (finite) sequence of states
$\exec = \state_0, \state_1, \ldots \state_m$, a function
$\Prob(\exec) : \Reals^{>0} \to [0,1]$ that given an $\epsilon > 0$,
returns the probability of the sequence $\exec$ when the parameter's
value is fixed to $\epsilon$, i.e.,
$
\Prob(\rho)(\epsilon) = \prod_{i=0}^{m-1}
\ptransf(\state_i,\state_{i+1})(\epsilon).
$
For a state $\state_0$ and a set of states $\states' \subseteq
\states$, once again we have a function that given a value $\epsilon$
for the parameter, returns the probability of reaching $\states'$ from
$\state_0$. This can be formally defined as
$
\reachprob{\state_0}{\states'}(\epsilon) = \sum_{\exec \in
  \state_0(\states\setminus\states')^*\states'} \Prob(\rho)(\epsilon).
  $
In other words, $\reachprob{\state_0}{\states'}(\epsilon)$ is the sum
of the probability of all sequences starting in $\state_0$, ending in
$\states'$, such that no state except the last is in $\states'$. 

\subsection{Parametrized DTMC semantics of {\ourlang}}

The parametrized DTMC semantics of a {\ourlang} program $P_\epsilon$
shall be denoted as $\sem{P_\epsilon}. $ We describe
$\sem{P_\epsilon}$ informally here and defer the formal definition to
Appendix~\ref{app:dtmcsemantics}.  As mentioned above, the key insight
in defining the semantics of a {\ourlang} program as a finite-state,
parametrized DTMC, is that the actual values of real and integer
variables need not be tracked.  A state of $\sem{P_\epsilon}$ is going
to be a tuple of the form $(\ell, f_{Bool}, f_{\sdom}, f_{\s{int}},
f_{\s{real}}, C)$ where $\ell$ is the label of the statement of
$P_\epsilon$ to be executed next. 
$\sem{P_\epsilon}$ is an abstraction of the set of all concrete states
that are compatible with it.  The partial functions $f_{Bool}$ and
$f_{\sdom}$ assign values to the $Bool$ and $\sdom$ variables,
respectively; this is just like in the natural semantics.

Let us now look at the partial function $f_{\s{real}}$. Intuitively,
$f_{\s{real}}$ is supposed to be the ``valuation'' for the real
variables. But instead of mapping each variable to a \emph{concrete}
value in $\Reals$, we shall instead map it into a finite set. To
understand this mapping, let us recall that in {\ourlang}, a real
variable is assigned only once in a program. Further, such an
assignment either assigns the value of a linear expression over
program variables, or a value sampled using a Laplace mechanism. In
the former case, $f_{\s{real}}$ maps a variable to the linear
expression it is assigned; and in the latter case, the value of the
parameters of the Laplace mechanism used in sampling. In the latter
case, since the first parameter is always of the form $a\epsilon$, we
need to note only $a$ in the mapping.  Notice that the range of
$f_{\s{real}}$ is now a finite set as $P_\epsilon$ contains only a
finite number of linear expressions, and the parameters of sampled
Laplacian take values from the finite set $\sdom$.  Similarly, the
partial function $f_{\s{int}}$ maps each integer variable to either
the linear expression it is assigned or the parameters of the sampled
discrete Laplace mechanism. The last state component $C$ is the set of
Boolean conditions on real and integer variables that hold along the
path thus far; this shall become clearer when we describe the
transitions. Since the Boolean conditions must be Boolean expressions
in the program or their negation, $C$ is also a finite set. These
observations show that $\sem{P_\epsilon}$ has finitely many
states. Intuitively, a state of $\sem{P_\epsilon}$ is an abstraction
of the set of all concrete states that respect the Boolean conditions
in $C$ and the constraints imposed by assignments of real and integer
expressions to real and integer variables, respectively.

We now sketch how the state is updated in $\sem{P_\epsilon}$. Updates
to $\sdom$ variables shall be as expected --- it shall be a
probabilistic transition if the assignment samples using an
exponential mechanism or a user-defined distribution, and it shall be a
deterministic step updating $f_{\sdom}$ otherwise. Assignments to real
variables are \emph{always deterministic} steps that change the function
$f_{\s{real}}$. Thus, even if the step samples using the Laplace
mechanism, in the semantics, it shall be modeled as a deterministic step
where $f_{\s{real}}$ is updated by storing the parameters of the
distribution. Similarly, all integer assignments are deterministic
steps as well. 

The assignment of a Boolean expression to a Boolean variable is as
expected --- we update the valuation $f_{Bool}$ to reflect the
assignment. The unexpected case is $\bv \gets R_1 \sim R_2$ when a
boolean variable gets assigned the result of the comparison of two
real expressions; the case of comparing two integer expressions is
similar. In this case, if the probability of $C$ holding is $0$, then
our construction will ensure that this state is not reachable with
non-zero probability.  Otherwise, we transition to a state where
$R_1 \sim R_2$ is added to $C$ with probability equal to the
probability that $(R_1 \sim R_2)$ holds conditioned on the fact that
$C$ holds, and with the remaining probability, we shall transition to
the state where $\neg(R_1 \sim R_2)$ is added to $C$. Thus, Boolean
assignments which compare integer and real variables are modeled by
probabilistic transitions. Finally, branches and while loop conditions
are deterministic steps, with the value of the Boolean variable (of
the condition) in $f_{Bool}$ determining the choice of the next
statement.

Let $\Prob_{DTMC}( P_\epsilon(\fb{in}) = \fb{out})$ denote the
probability that $P_{\epsilon}$ outputs value $\fb{out}$ on the input
$\fb{in}$ under the DTMC semantics. This is just the probability of
reaching an exit state with $\fb{out}$ as valuation of output
variables from the initial state with $\fb{in}$ as the valuation of
input variables.  We can show that this probability is the same as the
probability $\Prob_{natural}( P_\epsilon(\fb{in}) = \fb{out})$
obtained by the natural semantics discussed above.
\brown{The informal ideas outlined above are fleshed out to give a precise
mathematical definition and presented in
Appendix~\ref{app:dtmcsemantics}.}

It is worth noting how key syntactic restrictions in {\ourlang}
programs play a role in defining its semantics. The first restriction
is that integer and real variables are not assigned in the scope of a
while loop. This restriction is critical to ensure that the DTMC
$\sem{P_\epsilon}$ is finite-state. Since we track distribution
parameters and linear expressions for such variables, this restriction
ensures that we only remember a bounded number of these. Second,
{\ourlang} disallows a comparison between real and integer expressions
in its syntax. Recall that such comparison steps result in a
probabilistic transition, where we compute the probability of the
comparison holding conditioned on the properties in $C$ holding. It is
unclear if a closed-form expression for such probabilities can be
computed when integer and real random variables are compared. Hence
such comparisons are disallowed.

Probabilistic transitions in our semantics arise due to two
reasons. First are assignments to $\sdom$ variables that sample
according to either the exponential or a user-defined
distribution. The resulting probabilities are easily seen to be
definable in $\threal.$ The second is due to comparisons between real
and integer expressions. We can prove that in this case also, the
resulting probabilities are definable in $\threal$; this proof is
non-trivial and deferred to Appendix~\ref{app:dtmc-def}. All these
observations together give us the following theorem.

\begin{theorem}
\label{thm:semantics}
For any {\ourlang} program $P_\epsilon$, $\sem{P_\epsilon}$ is a
finite, definable, parametrized DTMC that is computable.
\end{theorem}

\begin{example}
\label{ex:SVT-dtmc}
The parametrized DTMC semantics of Algorithm~\ref{fig:SVTL} is
partially shown in Figure~\ref{fig:SVT-dtmc}. We show only the
transitions corresponding to executing lines 9 and 10 of the
algorithm, when $q_1 = u$ and $q_2 =v$ initially; here $u,v \in
\{\bot,\top\}$. The multiple lines in a given state give the different
components of the state. The first two lines give the assignment to
$Bool$ and $\sdom$ variables, the third line gives values to the
integer/real variables, and the last line is the Boolean conditions
that hold along a path. Since 9 and 10 are in the else-branch, the
condition $r_1 < r_T$ holds. Notice that values to real variables are
not explicit values, but rather the parameters used when they were
sampled. Finally, observe that probabilistic branching takes place
when line 10 is executed, where the value of $b$ is taken to be the
result of comparing $r_2$ and $r_T$. The numbers $p$ and $q$
correspond to the probability that the conditions in a branch hold,
given the parameters used to sample the real variables and
\emph{conditioned} on the event that $r_1 < r_T$.
\end{example}

\section{Checking differential privacy for {\ourlang} programs}
\label{sec:decidability}

We shall now establish that the  problem of checking differential privacy for {\ourlang} programs is decidable. The proof 
relies on the characterization  of the semantics of a {\ourlang} program as a finite, definable, parameterized DTMC (See Theorem~\ref{thm:semantics}). An  
important observation about a finite, definable, parametrized DTMC  is that the probability of reaching a given set of states $\states'$ from 
a given state $\state_0$ is both definable and
computable.
\begin{lemma}
\label{lem:dtmc}
For any finite-state, definable, parametrized DTMC $\blue{\cD} =
(\states,\ptransf)$, any
state $\state_0 \in \states $ and set of states $\states'\subseteq \states$, the function
$\reachprob{\state_0}{\states'}$ is definable in $\threal$. Moreover,
there is an algorithm that computes the formula defining
$\reachprob{\state_0}{\states'}$.
\end{lemma}

The proof of Lemma~\ref{lem:dtmc} exploits the connection between
reachability probabilities in DTMCs and linear
programming~\cite{marta-book,baier-katoen}; 
\brown{details are  in
Appendix~\ref{app:dtmc}}.
The main result of the paper now follows from Theorem~\ref{thm:semantics} and Lemma~\ref{lem:dtmc}.
\label{thm:decidability-finite}

\begin{theorem}
\label{thm:decidability}
The {\problemone} and {\problemtwo} problems are decidable for
{\ourlang} programs $P_\epsilon$, rational numbers $t\in \Rats^{>0}$ and definable functions $\delta(\epsilon)$. Furthermore, if $P_\epsilon$  is not $(t\epsilon,\delta)$ differentially private for some rational number $t$ and admissible value of $\epsilon$
then we can compute a counter-example.
\end{theorem}
\begin{proof}
Let $\fb{in}$ and $\fb{out}$ be arbitrary valuations to input and output
variables, respectively. Observe that the function $\epsilon \mapsto
\Prob(P_\epsilon(\fb{in}) = \fb{out})$ is nothing but
$\reachprob{\state_0}{\states'}$ in $\sem{P_\epsilon}$, where
$\state_0$ is the initial state corresponding to valuation $\fb{in}$,
and $\states'$ is the set of all terminating states that have
valuation $\fb{out}$ for output variables. Since $\sem{P_\epsilon}$
(Theorem~\ref{thm:semantics}) and $\reachprob{\state_0}{\states'}$
(Lemma~\ref{lem:dtmc}) are computable, we can construct a formula
$\varphi_{\fb{in},\fb{out}}(\epsilon, x_{\fb{in},\fb{out}})$ of $\lngreal$
that defines the function $\epsilon \mapsto \Prob(P_\epsilon(\fb{in}) =
\fb{out})$.

Let $\varphi_\delta(\epsilon,x_\delta)$ be the formula defining the
function $\delta$. Let $t=\frac p q$ where $p,q$ are natural
numbers. Consider the sentence
\[
\begin{array}{rl}
\psi = & \hspace*{-0.1in}\forall \epsilon. \forall z.
       [\forall x_{\fb{in},\fb{out}}]_{\fb{in} \in \cU,\fb{out} \in \cV}.
       \forall x_\delta.\\
& ((\epsilon>0) \wedge (\eulerv{p\epsilon}=z^q) \wedge (z>0) \wedge \varphi_\delta(\epsilon,x_\delta) \\
 & \hspace*{0.1in}\; \;\bigwedge_{\fb{in}\in \cU, 
                     \fb{out} \in \cV}\:
           \varphi_{\fb{in},\fb{out}}(\epsilon,x_{\fb{in},\fb{out}}) ) \\
 & \hspace*{0.2in}\to (\bigwedge_{(\fb{in}_1,\fb{in}_2) \in \Phi,
   O\subseteq \cV}\:  \\
  &  \hspace*{0.5in} \sum_{\fb{out}\in O} x_{\fb{in}_1,\fb{out}} < 
      z \sum_{\fb{out}\in O} x_{\fb{in}_2,\fb{out}}\:  + x_\delta))
\end{array}
\]
It is easy to see $P_\epsilon$ is $(t\epsilon,\delta(\epsilon))$ differentially private for all
$\epsilon$ iff $\psi$ is true over the reals. In the syntax of
$\lngreal$, we cannot take $q$th roots of $\euler$; therefore, we
introduce the variable $z$, which enables us to write the constraints
using only $\eulerv{a\epsilon}$, where $a \in \Nats$. Notice that
$\psi$ belongs to $\lngreal$ if we convert it to prenex
form. Decidability, therefore, follows from the decidability of
$\threal$.

If $P_\epsilon$ is not differentially private, then the sentence
$\psi$ does not hold. The decision procedure for $\threal$ will, in
this case, return an $\epsilon_0$ that witnesses the privacy violation of
$P_\epsilon$.  Using $\epsilon_0$, the counter-example $(\fb{in},\fb{in'},O,\epsilon_0)$ can be easily constructed by enumerating $\fb{in}$, $\fb{in'}$ and $O$. 
\end{proof}

An easy consequence of Theorem~\ref{thm:decidability} is that differential privacy is decidable for the subclass of program in $\s{Simple}$  that do not have integer and real-valued  variables. Let {\finlang} denote this set of programs (\brown{See Appendix~\ref{app:finlang} for the formal syntax of {\finlang})}. Observe that due to the presence of $\Whiles$, {\finlang}  programs may still have unbounded length executions (including infinite executions). 

\begin{corollary}
\label{cor:finite}
The {\problemone} and {\problemtwo} problems are decidable for
{\finlang} programs $P_\epsilon$, rational numbers $t\in \Rats^{>0}$ and definable functions $\delta(\epsilon)$. 
\end{corollary}

We observe that our methods can be employed to analyze larger classes of programs (than just those in {\ourlang}). For example, a sufficient condition to ensure the decidability is to consider programs with the property that, for each input, the probability distribution on the outputs is
definable in $\threal$ \brown{(See Appendix~\ref{app:generalsemantic})}. We conclude the section by showing how our procedure is useful when reasoning about integer and real-valued outputs.

\begin{remark}
\blue{We sketch here how the proofs of Theorem~\ref{thm:decidability} changes when the set of admissible $\epsilon$ is taken to be an interval $I$ with rational end-points. Let $P_\epsilon, t$ and $\delta(\epsilon)$ be as in the proof of Theorem~\ref{thm:decidability}.
When $\epsilon$ is restricted to an interval $I$, we will require the user-definable distributions to be definable in $\threal$ only on the interval $I$. As in the proof of Theorem~\ref{thm:decidability}, we can construct a formula $\varphi_{\fb{in},\fb{out}}(\epsilon, x_{\fb{in},\fb{out}})$ of $\lngreal$
that defines the function $\epsilon \mapsto \Prob(P_\epsilon(\fb{in}) =\fb{out}).$  For simplicity, consider the case when $I$ be the interval $[r,s]$. Consider the sentence $\psi_{I}$ that is obtained from $\psi$ in the proof of Theorem~\ref{thm:decidability} by replacing the subformula $(\epsilon>0)$ by $(a\leq \epsilon) \wedge (\epsilon \leq b).$ Then  $P_\epsilon$ is $(t\epsilon,\delta(\epsilon))$ will be differentially private for all
$\epsilon \in I$ iff $\psi_{I}$ is true over the reals.  }
\end{remark}

\subsection{Finite discretization of infinite output spaces}
\label{sec:findis}

Our decision procedure assumes that the output space is finite. In
several examples, the program outputs are reals or unbounded integers
(and combinations thereof). Nevertheless, we argue that our decision
procedure is useful for the verification of differential privacy in
this case also. In particular, our method provides an
under-approximation technique for checking the differential privacy of
programs with infinite outputs. Our approach in such cases is to
discretize the output space into finitely many intervals.

We illustrate this for the special case when a program $P$ outputs the
value of one real random variable, say $\rv$. Now, suppose that we
modify $P$ to output a finite discretized version of $\rv$ as
follows. Let $\seq=a_0< a_1< \ldots a_n$ be a sequence of rationals
and let 
$\Disc_{\seq}(x)$ be equal to $a_0$ if $x \leq a_0$, equal to $a_i$
($0 < i < n$) if $a_{i-1} < x \leq a_{i}$, and equal to \blue{$a_n$ if $x >
a_{n-1}$}.
%

Consider the program $P_{\Disc,\seq}$ that instead of outputting
$\rv$, outputs $\Disc_\seq(\rv).$ It is easy to see that if $P$ is
differentially private then so must be $P_{\Disc,\seq}.$ Therefore, if
$P_{\Disc,\seq}$ is not differentially private then we can conclude
that $P$ is not differentially private. 
Thus, if our procedure finds
a counter-example for $P_{\Disc,\seq}$, then it also has proved that the program $P$ is not differentially private. Our method is, therefore, 
an under-approximation technique for checking the differential privacy
of $P$. In fact, it is a \emph{complete} under-approximation method in
the sense that $P$ is differentially private iff for each possible
$\seq$, $P_{\Disc,\seq}$ is differentially private.

\section{Experimental evaluation} 
\label{sec:experiments}

We implemented a simplified version of the algorithm, presented
earlier, for proving/disproving differential privacy of {\ourlang}
programs. Our tool {\tool}~\cite{github} handles loop-free programs,
i.e., acyclic programs. Programs with bounded loops (with constant
bounds) can be handled by unrolling loops. The tool takes in an input
program $P_\epsilon$ parametrized by $\epsilon$ and an adjacency
relation, and either proves $P_\epsilon$ to be differentially private
for all $\epsilon$ or returns a counter-example. The tool can also be
used to check differential privacy for a given, fixed $\epsilon$, or
to check for $k\epsilon$-differential privacy for some constant
$k$. {\tool} is implemented in C++ and uses Wolfram
Mathematica\textregistered. It works in two phases --- in the first
phase, a Mathematica\textregistered script is produced with commands
for all the output probability computations and the subsequent
inequality checks and in the second phase, the generated script is run
on Mathematica. \brown{Details about the tool and its design can be found in
Appendix~\ref{app:experiments}.}

We used various examples to measure the effectiveness of our
tool. These include SVT~\cite{lyu2016understanding,DR14}, Noisy
Maximum~\cite{DingWWZK18}, Noisy Histogram~\cite{DingWWZK18} and
Randomized Response~\cite{DNRRV09} and their variants. Detailed descriptions of these
algorithms and their variants can be found in Appendix~\ref{sec:ex}.

We ran all the experiments on an octa-core Intel\textregistered Core
i7-8550U @ 1.8gHz CPU with 8GB memory. The running times reported are
the average of 3 runs of the tool. In the tables, T1 refers to the
time needed by the C++ phase to generate the Mathematica scripts, and
T2 refers to the time used by Mathematica to check the scripts. Due to
space constraints, we report only a small fraction of our experiments\brown{;
full details of all our experiments can be found in
Appendix~\ref{app:experiments}}.

\begin{table}
 \scalebox{0.8}{ \begin{tabular}{|c|N|N|}
    \hline
    Algorithm & Runtime (\parta/\partb) & $\epsilon$-Diff. Private\\
    \hline
    SVT & 0s/825s & \cmark\\
    SVT2 & 0s/768s & \cmark\\
   SVT5 & 0s/2s & \xmark\\
    NMax4 & 1s/58s & \xmark\\
     Rand2 & 0s/0s & \xmark\\
    \hline
  \end{tabular}}
  \caption{{\footnotesize Runtime for 3 queries for each algorithm searching over adjacency pairs and all $\epsilon$>0, with parameters being [c=1,~$\Delta$=1,~$\sdom$=\{-1,0,1\},~$\seq=(-1<0<1)$]. For SVT, we also have $T$=0.}}
  \label{tab:q3runtime-small}
\end{table}
\begin{table}
\scalebox{0.8}{
  \begin{tabular}{|c|c|c|c|c|c|N|}
    \hline
    Algo & |Q| & Output & Input 1 & Input 2 & $\epsilon$ 
    & Runtime (\parta/\partb) \\
    \hline
   SVT5 & 2 & [$\bot$ $\top$] & [-1 0] & [-1 -1] & 27 & 0s/2s\\
    NMax3 & 3 & -1, $\seq=(-1<0<1)$ & [-1 -1 -1] & [0 0 0] & 27 & 0s/310s\\
    NMax4 & 1 & 0, $\seq =(-1<0<1)$ & [-1] & [0] & 27 & 0s/2s\\
     {Rand2} & 1 & [$\bot$] & [$\bot$] & [$\top$] & 9/34 & 0s/0s\\
    \hline
  \end{tabular}
  }
  \caption{{\footnotesize Smallest Counter-example found for each non-differentially private algorithm, searching over all adj. pairs and $\epsilon>0$, with parameters being [c=1,~$\Delta$=1,~$\sdom$=\{-1,0,1\}]}}
  \label{tab:counter-small}
\end{table}

Salient observations about our experiments are follows.
\begin{enumerate}
\item {\tool} successfully proves algorithms to be differentially private 
  and finds counter-examples to demonstrate a violation of privacy in
  reasonable time. Table~\ref{tab:q3runtime-small} shows the running
  time of {\tool} on some examples for 3 queries. We chose to use 3
  queries because for algorithms that are not private, counter-examples
  can be found with 3 queries.
\item The time to generate Mathematica scripts is significantly smaller 
  than the time taken by Mathematica to check the scripts (i.e., T1
  $\ll$ T2). Further, most of the time spent by Mathematica is for
  computing output probabilities; the time to perform comparison
  checks for adjacent inputs was relatively small. Thus, programs that
  do not use real variables (Rand2 in Table~\ref{tab:q3runtime-small},
  for example) can be analyzed more quickly.
\item For algorithms that are not differentially private, {\tool} can 
  automatically identify the pair of inputs, output, and $\epsilon$
  for which privacy is violated. Table~\ref{tab:counter-small}, shows
  the results for the smallest counter-example found by {\tool} for
  some examples. Further, counter-examples found by {\tool} are much
  smaller, in terms of queries, than those found in~\cite{DingWWZK18};
  the number of queries needed in the counter-examples
  in~\cite{DingWWZK18} for NMax3, NMax4,
  and SVT5 were 5, 5, and 10, respectively, as opposed
  to 3, 1, and 2 found by {\tool}.
\item {\tool} is the first automated tool that can check
  $(\epsilon,\delta)$-differential privacy. To evaluate this feature,
  we tested {\tool} on a version of SVT, Sparse~\cite{DR14}, which is
  manually proven to be
  $(\frac{\epsilon}{2},\delta_\s{svt})$-differentially private for any
  number of queries in~\cite{DR14} by using advanced composition
  theorems. Here $\delta_\s{svt}$ is a second parameter in the
  algorithm. In our experiments, we tested $(\frac \epsilon
  2, \delta_\s{svt})$-differential privacy of Sparse with fixed values
  of $\delta_\s{svt}$ for $c=1,2$ and $3$ queries, validating the
  result in~\cite{DR14}. As we were dealing with only $3$ queries, we
  also managed to obtain better bounds on the error parameter.
\end{enumerate}

\section{Related work}
\label{sec:related-work}

The main thread of related work has focused on formal systems for proving that an algorithm is differentially private. Such systems are helpful because they rule out the possibility of mistakes in privacy analyses. Starting from Reed and Pierce~\cite{RP10}, several authors~\cite{GHHNP13, AGHK18} have proposed linear (dependent) type systems for proving differential privacy. However, it is not possible to verify some of the most advanced examples, such as a sparse vector or vertex cover, using these type systems. Moreover, type-checking and type-inference for linear (dependent) types are challenging. For example, the type checking problem for DFuzz, a language for differential privacy, is undecidable~\cite{AmorimAGH15}. Barthe et al~\cite{BKOZ13,BGGHS16,BFGGHS16} develop several program logics based on probabilistic couplings for reasoning about differential privacy. These logics have been used successfully to analyze many classic examples from the literature, including the sparse vector technique. However, these logics are limited: they cannot disprove privacy; extensions may be required for specific examples; building proofs is challenging. The last issue has been addressed by a series of works that provide automated methods for proving differential privacy automatically. Zhang and Kifer~\cite{ZK17} introduce randomness alignments as an alternative to couplings and build a dependent type system that tracks randomness alignments. Automation is then achieved by type inference. Albarghouthi and Hsu~\cite{AH18} propose coupling strategies, which rely on a fine-grained notion of variable approximate coupling, which draws inspiration both from approximate couplings and randomness alignment. They synthesize coupling strategies by considering an extension of Horn clauses with probabilistic coupling constraints and developing algorithms to solve such constraints. Recently Wang et al~\cite{WDWKZ19} develop an improved method based on the idea of shadow executions. Their approach is able to verify Sparse Vector and many other challenging examples efficiently. However, these methods are limited to vanilla $\epsilon$-differential privacy and do not accommodate bounds that are obtained by advanced composition (since $\delta\neq 0$).

In an independent line of work, Chatzikokolakis, Gebler and
Palamidessi~\cite{chatzGP14} consider the problem of differential privacy for Markov
chains. 
Later, Liu, Wang, and Zhang~\cite{LiuWZ18} develop a probabilistic model checking approach for verifying differential privacy properties. Their approach is based on modeling differential private programs as Markov chains. Their encoding is more direct than ours (i.e.\, it assumes that a finite-state Markov chain is given), and they do not provide a decision procedure with real and integer variables. Furthermore, the DTMCs are not parameterized by $\epsilon.$ 
Chistikov and Murawski and Purser~\cite{ChistikovMP18,ChistikovMP19} propose an elegant method based on skewed Kantorovich distance for checking approximate differential privacy of Markov chains. 

The dual problem is to find violations of differential privacy automatically. This is useful to help privacy practitioners discover potential problems early in the development cycle.
Two recent and concurrent works by Ding et al~\cite{DingWWZK18} and
Bischel et al~\cite{BichselGDTV18} develop automated methods for
finding privacy violations. Ding et al. propose an approach that
combines purely statistical methods based on hypothesis testing and
symbolic execution. Bischel et al. develop an approach based on a
combination of optimization methods and language-specific techniques
for computing differentiable approximations of privacy
estimations. Both methods are fully automated.  However, both methods
can only be used for concrete numerical values of the privacy budget
$\epsilon$.

Gaboardi et. al~\cite{GaboardiNP19} study the complexity of deciding differential privacy for randomized Boolean circuits. Their results are proved by reduction to majority problems and are incomparable with ours: the only probabilistic choices in~\cite{GaboardiNP19}  are fair coin tosses and $\euler^{\epsilon}$ is taken to be a fixed rational number.
%

\section{Conclusions}

We showed that the problem checking differential
privacy is in general undecidable, identified an expressive
sub-class of programs ({\ourlang}) for which the problem is decidable,
and presented the results of analyzing many known differential privacy
algorithms using our tool {\tool} which implements a decision
procedure for {\ourlang} programs. Advantages of {\tool} include the
ability to automatically, both prove algorithms to be private for all
$\epsilon > 0$, and find counter-examples to demonstrate privacy
violations. In addition {\tool} can check bounds that are based on
concentration inequalities, in particular bounds that use advanced
composition theorems. Such bounds are out of reach of most other tools
that prove privacy or search for counter-examples.

In the future, it would be interesting to extend this work to handle
programs with input/output variables that take values in infinite
domains, and parametrized privacy algorithms that work for an
unbounded number of input and output variables. Another important
problem is developing decision procedures that can prove tight
accuracy bounds, and detect violations of accuracy bounds.
We also plan to investigate extending the decision procedure to 
cover algorithms that are currently out of the scope of our decision procedure such as 
the multiplicative weights and iterative database construction~\cite{HardtR10,GuptaRU12}, and those involving Gaussian distributions.

\section{Acknowledgements}
\blue{We thank the anonymous reviewers for their useful comments. Their inputs have improved the paper, especially the presentation of the semantics. 
Rohit Chadha was partially supported by NSF CNS 1553548 and NSF CCF 1900924. A. Prasad Sistla was partially supported by NSF CCF 1901069 and NSF  CCF 1564296.
Mahesh Viswanathan was partially supported by NSF CCF 1901069.}

\bibliography{header,main}
\clearpage
\appendix


\section{Semantics of  {$\s{Simple}$} }
\label{app:semantics}
In this section, we give the semantics  of  our  {$\s{Simple}$} language.
This semantics will be given as a set of computations and a
probability space on the set of computations.
Recall that we have assumed that in each computation, a reference to a
variable is preceded (sometime earlier) by an assignment to the variable. 

For the rest of this section, let us fix a {$\s{Simple}$} program
$P_\epsilon$ and an $\epsilon>0$. We let
 $\s{L_{P_\epsilon}}$  denote the set of labels appearing in
$P_\epsilon$. The set of Boolean variables, $\sdom$ variables (including input/output variables), integer variables and reals variables occurring in 
$P_\epsilon$ shall be denoted by $\cB_{P_\epsilon}$, $\cX_{P_\epsilon}$,  $\cZ_{P_\epsilon}$ and $\cR_{P_\epsilon}$ respectively. 


In order to define the semantics of $P_\epsilon$, we will use an
auxiliary function $\nxt$ that given a label, identifies the label of
the statement to be executed next. Observe that for most program
statements, the next statement to be executed is unique. However, for
{\Ifs} and {\Whiles} statements, the next statement depends on the
value of a Boolean expression. We will define $\nxt(\ell)$ to be a set
of pairs of the form $(\ell',c)$, where $c$ is a Boolean condition on the
variables of $P_\epsilon$, with the understanding that $\ell'$
is the  label of the next statement to be executed if $c$ currently holds. Thus, for a label $\ell$, $\nxt(\ell)$
will either be $\set{(\ell',\true)}$ or $\set{(\ell_1,c),(\ell_2,\neg
  c)}$. We do not give a precise definition of $\nxt(\cdot)$, but we
will use it when defining the semantics.

\paragraph{States.} 
States of $P_\epsilon$ will be of the form 
\[
(\ell,h_{Bool},h_{\sdom}, h_{\integers}, h_{\Reals}).
\] 
Informally, $\ell \in \s{L_{P_\epsilon}}$ is the label of the statement to be
executed, $h_{Bool}$, $h_{\sdom}$, $h_{\integers}$, and $h_{\Reals}$
are  functions assigning ``values'' to program variables (of
appropriate type).  More specifically, we have $h_{Bool} : \cB_{P_\epsilon} \to
\set{\true,\false}$, $h_{\sdom} : \cX_{P_\epsilon} \to \sdom$, $h_{\integers} : \cZ_{P_\epsilon}
\to \integers$ and $h_{\Reals} : \cR_{P_\epsilon} \to \Reals$ . We let $S$
denote the set of all states. We define a discrete state $ds$ to be a
tuple $(\ell,h_{Bool},h_{\sdom}, h_{\integers})$ where $\ell,h_{Bool},
h_{\sdom}, h_{\integers}$
are as defined above. Note that a discrete state does not specify
values to variables in $\cR_{P_\epsilon}.$
For a state $s$ and an expression $e$ which
is a Boolean, real or an integer expression, we let $Val(s,e)$ denote the value
obtained by evaluating $e$ in the state $s.$ Note that  if $e$ is a
boolean expression, $Val(s,e)$ is either True or False. We also define
the value of a comparison between two expressions as follows. For a
comparison expression $e_1 \sim e_2$, $Val(s, e_1 \sim e_2)= True$
if $Val(s,e_1) \sim Val(s,e_2)$ holds, otherwise $Val(s, e_1 \sim
e_2)= False.$ 
The value of a $\sdom$ expression $e,$ its value in state $s=(\ell,h_{Bool},h_{\sdom}, h_{\integers}, h_{\Reals})$ will be denoted by 
$h_{\sdom}(e).$ For a sequence of $\sdom$ expressions $\sdom$ $\tilde e= e_1,\ldots, e_m$,
$h_{\sdom}(\tilde e)$ will denote the sequence $h_{\sdom}(e_1),\ldots,h_{\sdom}(e_m).$

\paragraph{Measurable sets of states.} Let $\cR_{P_\epsilon}=\set{r_1,...,r_t}.$ With
each vector $u=(u_1,...,u_t)\in \Reals^t$, we associate a unique
function $h_{\Reals}^{u}: \cR_{P_\epsilon} \to \Reals$ such that
$h_{\Reals}^{u}(r_i)=u_i$ for $1\leq i\leq t.$
Given a
discrete state $ds=(\ell,h_{Bool},h_{\sdom}, h_{\integers})$ and a
Borel set $D\subseteq \Reals^t$, we let $\sem{(ds,D)} = \set{(\ell,
 h_{Bool},h_{\sdom}, h_{\integers},h_{\Reals}^u)\st u\in D}.$ Now, we
define $\cE$, the set of measurable sets of states, to be the
$\sigma$-algebra generated by the sets of states of the form
$\sem{(ds,D)}$ where $ds$ is a discrete state and $D\subseteq
\Reals^t$ is a Borel set.

\paragraph{Markov Kernel $K_\epsilon$.} We give the single step
semantics of the program $P_\epsilon$ as a Markov kernel from the
measure space $(S, \cE)$ to itself. Formally, $K_\epsilon: S \times \cE \to
\Reals$, where $K_\epsilon(s,C)$ gives the probability that the next
state of $P_\epsilon$ is in   $C$ given that its current state is $s.$
We fix the state $s=(\ell,h_{Bool},h_{\sdom}, h_{\integers},h_{\Reals})$ and the set $C\in \cE$ of states, and define the
value of $K_\epsilon(s,C)$ based on the following cases.

\paragraph{$\sdom$ assignments.}
Let $\nxt(\ell) = \set{(\ell',\true)}$ and let $\dv$ be the variable
being assigned in $\ell$. There are two cases to consider. First,
consider the case where $\dv$ is assigned a value of a $\sdom$
expression $e$. In this case, $K_\epsilon(s,C) = 1$ if 
$(\ell',h_{Bool},$  $h_{\sdom}[\dv \mapsto h_{\sdom}(e)],
h_{\integers},h_{\Reals})\in C;$ otherwise $K_\epsilon(s,C)=0.$
The second case is when $\dv$ is assigned a random
value according to $\pexp{a\epsilon, F(\tilde{\dv}),e}$ or
$\chose(a\epsilon, \tilde{e})$. For $d \in \sdom$, let $\s{prob}(d)$
be the probability of $d$  based on the
distribution; note, that these probabilities will depend on the value
of  $h_{\sdom}(e)$ and $h_{\sdom}(\tilde{e})$. Then, $K_\epsilon(s,C)
=\sum_{d\in D}\s{prob}(d)$ where $D = \set{ d \st (\ell',h_{Bool},
 h_{\sdom}[\dv \mapsto d],
h_{\integers},h_{\Reals})\in C}.$ Note that the right hand sum is zero if $D=\emptyset.$

\paragraph{Integer assignments.}
Let $\nxt(\ell) = \set{(\ell',\true)}$ and let $\iv$ be the variable
being assigned in $\ell$. Again there are two cases to
consider. First, consider the case where $\iv$ is assigned a value of
an integer expression $e$.
In this case, $K_\epsilon(s,C) = 1$ if 
$(\ell',h_{Bool},h_{\sdom},
h_{\integers} [\iv \mapsto Val(s,e)],h_{\Reals})\in C;$ otherwise $K_\epsilon(s,C)=0.$ 
Next, consider the case when $\iv$ is assigned a random value
according to $\DLap{a\epsilon,e}$. For $j\in \integers,$ let
$\s{prob}(j)$ be the probability assigned to the integer $j$ by the
distribution given by  $\DLap{a\epsilon,h_{\sdom}(e)}.$ 
Then, $K_\epsilon(s,C)
=\sum_{j\in D}\s{prob}(j)$ where $D = \set{ j \st (\ell',h_{Bool},
 h_{\sdom},
h_{\integers} [\iv \mapsto j],h_{\Reals})\in C}.$ Note that the right hand sum is zero if $D=\emptyset.$

\paragraph{Real assignments.}
Let $\nxt(\ell) = \set{(\ell',\true)}$ and let $\rv$ be the variable
being assigned in $\ell$. 
Again there are two cases to
consider. First, consider the case where $\rv$ is assigned a value of
a real expression $e$.
In this case, $K_\epsilon(s,C) = 1$ if 
$(\ell',h_{Bool},h_{\sdom},
h_{\integers},h_{\Reals} [\rv \mapsto Val(s,e)])\in C;$ otherwise $K_\epsilon(s,C)$ $=0.$ 
In the second case,  $\rv$ is assigned a random value according to
$\Lap{a\epsilon,e}$. In this case, $K_\epsilon(s,C)= Prob(D)$
where $D= \set{r\in \Reals \st (\ell',h_{Bool},h_{\sdom},
h_{\integers},h_{\Reals} [\rv \mapsto r])\in C}$ and $Prob(D)$ is the
probability given to set $D$ by the distribution
$\Lap{a\epsilon,h_\sdom(e)}$. Observe that $D\subseteq \Reals$ is a Borel set.

\paragraph{Boolean assignments.}
Again let $\nxt(\ell) = \set{(\ell',\true)}$ and let $\bv$ be the
variable being assigned in $\ell$ and $e$ the expression being assigned.
Now, $K_\epsilon(s,C) = 1$ if 
$(\ell',h_{Bool} [\bv \mapsto Val(s,e)],h_{\sdom},
h_{\integers},h_{\Reals})\in C;$ otherwise $K_\epsilon(s,C)=0.$ 

\paragraph{$\Ifs$ statement.}
In this case, $\nxt(\ell) = \set{(\ell_1,c),(\ell_2, \neg c)}$ for
some Boolean condition $c$. If
either 
$Val(s,c)=\true$ and  $(\ell_1,h_{Bool},
h_{\sdom},h_{\integers},h_{\Reals})\in C$ or $Val(s,c)=\false$ and  $(\ell_2,h_{Bool},$ $
h_{\sdom},h_{\integers},h_{\Reals})\in C$ then $K_\epsilon(s,C) = 1$, otherwise $K_\epsilon(s,C)=0$.

\paragraph{$\Whiles$ statement.}
Again let $\nxt(\ell) = \set{(\ell_1,c),(\ell_2, \neg c)}$. This case
is identical to the case of $\Ifs$ statement, and so is skipped.

\paragraph{$\ext$ statement.}
In this case, $K_\epsilon(s,C) = 1$ if $s\in C$; otherwise $K_\epsilon(s,C) = 0.$

\paragraph{Probability Spaces on finite executions.} For each $i> 0$,
 we define a probability space $\Phi_i\:=(S^i,{\Sigma}_i, \phi_i)$ capturing the set
 of finite executions of length $i$ $S^i$, the class $\Sigma_i$ of measurable sets
 of executions of length $i$ and a probability measure $\phi_i$, as follows. Let $\vec{C}=\:(C_1,...,C_i)$ be a
 sequence of measurable sets where, for $1\leq j\leq i$,  $C_j\in \cE.$ 
For each such $\vec{C}$, let $Exec(\vec{C})= \set{ (s_1,s_2,...,s_i)\st
  s_j\in C_j, 1\leq j\leq i}.$ 
 The set $\Sigma_i$ of
 measurable sets of finite executions of length $i$, is the
 $\sigma$-algebra generated by the sets of executions $Exec(\vec{C})$
where $\vec{C}$ is a vector of measurable sets as specified
above. Essentially, $(S^i,{\Sigma}_i)$ is the measurable space obtained by taking the
product of $(S,\cE)$, $i$ times. The probability measure $\phi_i$ is
defined by first fixing an initial state and using the Markov kernel $K_\epsilon$ as follows.
 
\paragraph{Initial State and initial distribution.}
For an integrable function $g$ with respect to a measure space $(X,\Sigma,\mu)$, let  $\int_{X_1} g \mu(\mathrm{d}x)$  denote the integral of function $g$ with respect to measure $\mu$
over a measurable set $X_1\in \Sigma.$ 
Let $\ell_{\fb{in}}$ be the label of the first statement 
of $P_\epsilon$. Let 
$h_{Bool}^{\fb{in}}$, $h_{\integers}^{\fb{in}}$, and 
$h_{\Reals}^{\fb{in}}$ be  functions such that $h_{Bool}^{\fb{in}}$  assigns $\false$ to every 
variable in $\cB_{P_\epsilon}$, and $h_{\integers}^{\fb{in}}$,$h_{\Reals}^{\fb{in}}$
assign value zero to every variable in $\cZ_{P_\epsilon}, \cR_{P_\epsilon}$ respectively . An 
initial state of $P_{\epsilon}$ will be of the form 
$(\ell_{\fb{in}},h_{Bool}^{\fb{in}},h_{\sdom}^{\fb{in}},
h_{\integers}^{\fb{in}}, h_{\Reals}^{\fb{in}})$, where 
$h_{\sdom}^{\fb{in}}$ assigns the given values to input variables and 
assigns zero to all other variables in $\cX_{P_\epsilon}$ (recall that all input
variables are in $\cX_{P_\epsilon}$); the 
values given to the input variables by $h_{\sdom}^{\fb{in}}$ will be the 
``initial input value''. We fix a unique initial state 
$s_{\s{init}}$.
Let $\phi_{\s{init}}$ be a distribution on the measure space $(S,\cE)$ such that 
for any $C'\in \cE$, $\phi_{\s{init}}(C')=1$ if $s_{\s{init}}\in C'$; otherwise,
$\phi_{\s{init}}(C')=0.$
Now,  $\phi_i$ is the unique probability measure defined by the
Markov kernel $K_\epsilon$ with initial distribution $\phi_{\s{init}}$ such that for each sequence of measurable sets  $\vec{C}=\:(C_1,...,C_i)$, $\phi_i(Exec(\vec C))$ is 
$$\int_{C_1}\int_{C_2} \cdots \int_{C_i} \mathbf{1}\; K_\epsilon(x_{i-1},\mathrm{d}x_i) \cdots  K_\epsilon(x_1,\mathrm{d}x_2)\phi_{\s{init}}(\mathrm{d}x_1)$$
where $\mathbf{1}$ is the constant function that takes $1$ everywhere. Please see~\cite{Cinlar-Book} for additional details.

We let
 $\Prob_{natural}( P_\epsilon(\fb{in}) = \fb{out})$ denote the probability that
 $P_{\epsilon}$ outputs value $\fb{out}$ on the  input $\fb{in}$. We define
 this probability as follows. Let $\alpha=(s_1,...,s_i)\in S^i$ be an
 execution. We say that $\alpha$ is a {\it
   required}  execution if $\alpha$ is a terminating 
 execution 
 with  output  $\fb{out}$, i.e., it satisfies the following two conditions: 
(i)  $s_i=(\ell,h_{Bool},h_{\sdom}, h_{\integers},h_{\Reals})$ where $\ell$ is the label of $ext$
 statement and valuation of output variables is $\fb{out}$; (ii) if
 $j<i$ and $s_{j}=(\ell', f'_{Bool},
 f'_{\sdom}, f'_{\s{int}},f'_\Reals)$ then $\ell'$ is not the label of  $ext$
 statement. For each $i>0$, let $Req_i$ be the set of all required
 executions in $S^i.$ It is easy to see that, for each $i>0$,
 $Req_i\in \Sigma_i$ and no execution in $Req_i$ is a prefix of an
 execution in $Req_{i+1}.$
We define  $ \Prob_{natural}( P_\epsilon(\fb{in}) = \fb{out})=\sum_{i>0}
\phi_i(Req_i).$ 

       
\newcommand{\stat} {\mathsf{state}}
\newcommand{\cntr} {\mathsf{cntr}}

\newcommand{\nb}[1]{\mathsf{\bv_{#1}^\mathsf{next}}}
\newcommand{\nr}[1] {\mathsf{\rv_{#1}^\mathsf{next}}}
\newcommand{\steps} {\mathsf{number\_steps}}
\newcommand{\posn}{\mathsf{pos}}
\newcommand{\currsteps}{\mathsf{curr\_step}}
\newcommand{\cnt}{\mathsf{continue}}

\section{Undecidability of checking differential privacy of ${\s{Simple}}$ programs }
\label{app:undecid}

In this section, we will prove Theorem~\ref{thm:undecidability}. That
is, we will show that both {\problemone} and {\problemtwo} are
undecidable.

\begin{proof}
Recall that a 2-counter Minsky Machine is 
tuple $\cM= (Q, \qs, \qf, \Delta^1_{inc}, \Delta^2_{inc}, \Delta^1_{jzdec}, \Delta^2_{jzdec} )$ where
\begin{itemize}
\item $Q$ is a finite  set of control states.
\item $\qs\in Q$ is the initial state.
\item $\qf\in Q$ is the final state.
\item $\Delta^i_{inc} \subseteq Q\times Q$ is the increment of counter $i$ for  $i=1,2.$
\item $\Delta^i_{jzdec} \subseteq Q\times Q \times Q$ is the conditional jump of counter $i$ for  $i=1,2.$
\end{itemize}

$\cM$ is said to be deterministic if from each state $q$, there is at
most one transition out of $q$.  The semantics of $\cM$ is defined in
terms of a transition system $(Conf, (\qs,0,0), \rightarrow)$ where
$Conf=Q\times \Nats \times \Nats$ is the set of configurations, $
(\qs,0,0)$ is the initial configuration and $\rightarrow$ is defined
as follows:

\begin{tabular}{ll}
$(q,i,j)\rightarrow (q',i+1,j)$ & if $(q,q')\in \Delta^1_{inc},$\\
$(q,i,j)\rightarrow (q',i,j+1)$ & if $(q,q')\in \Delta^2_{inc},$\\
$(q,i,j)\rightarrow (q',i,j)$  & if $i=0$ and $(q,q',q'')\in \Delta^1_{jzdec},$\\
$(q,i,j)\rightarrow (q'',i-1,j)$ & if $i\ne 0$ and $(q,q',q'')\in \Delta^1_{jzdec},$\\
$(q,i,j)\rightarrow (q',i,j)$ & if $j=0$ and $(q,q',q'')\in \Delta^2_{jzdec},$ \\
$(q,i,j)\rightarrow (q'',i,j-1)$ &  if $j\ne 0$ and $(q,q',q'')\in \Delta^2_{jzdec}.$\\
\end{tabular}

A sequence of configurations $s_0,s_1,\ldots s_k$ is said to be a
computation of $\cM$ is $s_0=(\qs,0,0)$ and $s_i\rightarrow s_{i+1}$
for $i=0,1,\ldots k-1.$ A computation $s_0,s_1,\ldots s_k$ is said to
be a terminating computation of $\cM$ if $s_k=(\qf,i,j)$ for some $i,j\in
\Nats.$

We show that given a 2-counter Minsky Machine $\cM$, there is a
program $P^\cM_\epsilon \in \s{Simple}$ such that for each
$\epsilon>0,$
\begin{itemize}
\item[(a)] $P^{\cM}_\epsilon$ has only one input $\dv_{\fb{in}}$ and
  only one output $\dv_{\fb{out}}$ taking values in $\sdom=\set{0,1}.$
\item[(b)] $P^\cM_\epsilon$ terminates with probability $1.$
\item[(c)] $P^\cM_\epsilon$ is $(\epsilon,0)$-differentially private
  with respect to the adjacency relation
  $\Phi=\set{(0,1),(1,0)}$ if and only if $\cM$ does
  not halt.
\end{itemize}

\RestyleAlgo{boxed}

  \removelatexerror
  \begin{algorithm}
  \DontPrintSemicolon
\SetAlgoLined

 \KwIn{$\dv_{\fb{in}}$}
 \KwOut{$\dv_{\fb{out}}$}
 \;
 
 $\dv_{\fb{out}} \gets 0$\;
 $\rv_0 \gets \Lap{ {\epsilon}  , 0}$\;
 $\bv_\mathsf{test}\gets \rv_0>0$\;

  \If {$\bv_\mathsf{test}$}
  {
    $\rv_\steps \gets  \Lap{ {\epsilon}  , 0}$\;
    $\rv_\currsteps \gets \rv_0$\;
    $\bv_\cnt \gets \rv_\steps > \rv_\currsteps $\;
  
    $\bv_1 \gets \true$\;
      $\bv_2 \gets \false$\;
    $\cdots$\;
    $\bv_m \gets \false$\;
    $\rv_1\gets \rv_0$\;
    $\rv_2\gets \rv_0$\;
    \While{$\bv_\cnt$}
     {
        $s_1$\;
        $\vdots$\;
         $s_n$\;
         $\bv_1 \gets \nb 1$\;
         $\ldots$\;
         $\bv_m\gets \nb m$\;
         $\rv_1 \gets \nr 1$\;
         $\rv_2 \gets \nr 2$\;
         $\rv_\currsteps \gets \rv_0 + \rv_\currsteps$\;
         $\bv_\cnt \gets \rv_\steps > \rv_\currsteps $\;
     }
    
   \If {$(\bv_m\; \mathrm{and}\;  \s{EQ}(\dv_{\fb{in}},1))$}
      {
       $\dv_{\fb{out}} \gets 1$\;
      } 
     
  }
 $\ext$

\caption{Program $P^\cM_\epsilon$ simulating $2$-counter machine $\cM$}
\label{fig:undecid}
\end{algorithm}


Given a $2$-counter Machine $\cM,$ $P^\cM_\epsilon$ is constructed as
follows. Without loss of generality, let  $Q=\set{q_1,\ldots,q_m}$ and let $q_1$ be the initial state and $q_m$ be the final state. 
We will model a state in $Q$ using $m$ Boolean variables 
$\bv_1,\ldots,\bv_m.$ If the current state is $q_i$ then $\bv_i$ will be set to true and all other variables will be set to false. The counters will be modeled using  real variables as follows. Initially a real variable $\rv_0$ will be sampled from Laplacian distribution.
If $\rv_0\leq 0$, we will exit the program. Otherwise, we will initialize two real variables $\rv_1,\rv_2$ to be $\rv_0.$ $\rv_1,\rv_2$ will model the counters as follows. If the first (second respectively) counter is going to hold natural number $i$ then $\rv_1=(i+1)\rv_0$ ($\rv_2=(i+1)\rv_0$ respectively). Incrementing the first counter (second respectively) counter is achieved by adding $r_0$ to $\rv_1$ ($\rv_2$ respectively). Decrementing the first counter (second respectively) counter is achieved by sibtracting $r_0$ from $\rv_1$ ($\rv_2$ respectively). 
 For encoding the  transition relations $ \Delta^1_{inc},
\Delta^2_{inc}, \Delta^1_{jzdec}$ and $\Delta^2_{jzdec},$ we use  variables  $\nb 1,\ldots,\nb m, \nr 1,\nr 2$ to compute the next configuration as expected.
 For
example, the transition $(q_i,q_j,q_k)\in \Delta^1_{jzdec}$ can be
encoded using conditional statements as follows:
 $$ \begin{array} {l}
         \bv_{i,j,k} \gets  \rv_1 =\rv_0 \\
         \mathsf{if}\;  (\bv_{i,j,k}  \mbox{ and } \bv_i) \\
                   \hspace{0.2cm}        \mathsf{then}\;   \nb j\gets \true;\; \nb 1\gets \false;\; \ldots\;  \nb {j-1}\gets \false;\; \\
                   \hspace{1.2cm}  \nb {j+1}\gets \false;\;  \ldots;\; \nb {m}\gets \false\\ 
                   \hspace{0.2cm}     
                   \mathsf{else}\; \nr 1 \gets \rv_1 - \rv_0;\;  \nb k\gets \true;\; \nb 1\gets \false;\; \ldots\;\\
                  \hspace{1.2cm}  \nb {k-1}\gets \false;\;  \nb {k+1}\gets \false;\; \ldots;\; \nb {m}\gets \false\\                               
          \mathsf{end} \\               
        
 \end{array} $$
 Let $s_1,s_2,\ldots, s_n$ be the statements encoding the transition
relation.  Consider the program $P^\cM_\epsilon$ given in
Algorithm~\ref{fig:undecid}. The program $P^\cM_\epsilon$ initially
samples $r_0$ from a continuous Laplacian distribution. If the sampled value is $\leq 0$
then it outputs $0$. Otherwise, it starts simulating $\cM$. In order to make sure that the program terminates, we sample another real variable $\rv_\steps$ and simulate $k$ steps of the program where $k$ is the smallest number such that $k\rv_0 > \rv_\steps.$

At
the end of the simulation, if the halting state is reached and the
input is $1$ then it outputs $1$. Otherwise, it outputs
$0$.
 
Clearly, $P^\cM_\epsilon$ satisfies properties (a) and (b) above. That
the program $P^\cM_\epsilon$ has property (c) above follows from the
following observations:
\begin{enumerate}
\item If $\cM$ does not halt then $P^\cM_\epsilon$ outputs $0$
  with probability $1.$
\item If $\cM$ halts then $P^\cM_\epsilon$ outputs $1$ with
  non-zero probability on input $1$ and outputs $1$ with zero
  probability on input $0.$
\end{enumerate}
This shows that {\problemone} is undecidable.  Undecidability of
{\problemone} is obtained by taking $\epsilon_0$ to be any constant
rational number, say $\frac 1 2.$
\end{proof}

\section{{\ourlang} encoding of  exponential distribution}
\label{app:example_dipwhile}

\begin{example}
Given $\epsilon>0$ and  $\mathsf{offset},$ let $\mathsf{Lap}^{+}(\epsilon,\mathsf{offset})$ be the  continuous distribution whose probability density function (p.d.f.) is given by 
 $$ f_{\epsilon,\mu}(x) = \begin{cases}\epsilon \   \euler^{-\epsilon  (x -\mathsf{offset})}  & \mbox{ if } x\geq \mathsf{offset}\\
                                                             0 & \mbox{otherwise}  
                     
 \end{cases} .$$
 Observe that the one-sided Laplacian distribution $\mathsf{Lap}^{+}(\epsilon,0)$  is the standard exponential distribution. Our language is expressive enough to encode one-sided Laplacians as follows. Consider the sequence of statements:
 $$ \begin{array}{l}
       X\gets \Lap{\epsilon,0};\\
       b \gets X\leq 0;\\
       \ifstatement {b} {Y\gets X} {Y\gets (-1) X };\\
       Z \gets Y+\mathsf{offset}
     \end{array}
 $$
 The effect of the sequence of statements is that $Z$ has the one-sided Laplacian distribution $\mathsf{Lap}^{+}(\epsilon,\mathsf{offset}).$
\end{example}


\section{Formal DTMC Semantics of {\ourlang} programs}
\label{app:dtmcsemantics}

We define formally $\sem{P_\epsilon}$, the DTMC semantics of an {\ourlang} program $P_\epsilon$. 
Let us recall some key restrictions in {\ourlang} programs. The first
restriction is that real and integer-valued variables are never
assigned within the scope of a $\s{while}$ statement. Hence, they are
assigned only a bounded number of times, and therefore, without loss
of generality, we can assume that they are assigned a value exactly
\emph{once}. Second, real valued expressions are never compared
against integer valued expressions.

Let us fix some basic notation. Partial functions from $A$ to $B$ will
be denoted as $A \pto B$. The value of $f: A \pto B$ on $a \in A$,
will be denoted as $f(a)$. Two partial functions $f$ and $g$ will be
equal (denoted $f \simeq g$) if for every element $a$, either $f$ and
$g$ are both undefined, or $f(a) = f(b)$. If $f : A \pto B$, $a \in A$
and $b \in B$, then $f[a \mapsto b]$ denotes the partial function that
agrees with $f$ on all elements of $A$ except $a$; on $a$, $f[a
  \mapsto b](a) = b$.

In the rest of this section let us fix a {\ourlang} program
$P_\epsilon$. $\s{L}$ will denote the set of labels appearing in
$P_\epsilon$. 
A valuation $val$ for $\sdom$ variables is a function that assigns a
value in $\sdom$ to variables in $\cX$; we will denote set of all such
valuations by $\s{V}_{\sdom}$. Given a valuation $val \in
\s{V}_{\sdom}$ and a real expression $e$, $val(e)$ denotes the real
expression that results from substituting all the $\sdom$ variables
appearing in $e$ by their value in $val$. Similarly, for an integer
expression, $val(e)$ is the partial evaluation of $e$ with respect to
$val$. Finally, for a comparison $e_1 \sim e_2$ between two
expressions $e_1$ and $e_2$, again we will define $val(e_1 \sim e_2)$
to be $val(e_1) \sim val(e_2)$. Let us denote the set of integer
expressions, real expressions, and Boolean comparisons, appearing on
the right hand side of assignments in $P_\epsilon$ by $P_Z, P_R$, and
$P_B$, respectively. Three sets of expressions will be used in
defining the semantics, and they are as follows.
\[
\begin{array}{l}
\s{zExp} = \set{val(e)\: |\: val \in \s{V}_{\sdom},\ e \in P_Z}\\
\s{rExp} = \set{val(e)\: |\: val \in \s{V}_{\sdom},\ e \in P_R}\\
\s{bExp} = \set{val(e)\: |\: val \in \s{V}_{\sdom},\ e \in P_B}
\end{array}
\]
Thus, $\s{zExp}$, $\s{rExp}$, and $\s{bExp}$ are partially evaluated
expression appearing on the right hand side of assignments in
$P_\epsilon$. Notice that the sets $\s{L}$, $\s{zExp}$, $\s{rExp}$,
and $\s{bExp}$ are all finite. Finally, let $\s{Const}$ be the set of
rational constants appearing as coefficient of $\epsilon$ of Laplace
and discrete Laplace assignments in $P_\epsilon$; again $\s{Const}$ is
finite.

In order to define the semantics of $P_\epsilon$, we will use an
auxiliary function $\nxt$ that given a label, identifies the label of
the statement to be executed next. Observe that for most program
statements, the next statement to be executed is unique. However, for
{\Ifs} and {\Whiles} statements, the next statement depends on the
value of a Boolean expression. We will define $\nxt(\ell)$ to be a set
of pairs of the form $(\ell',c)$ with the understanding that $\ell'$
is the next label if $c$ holds. Thus, for a label $\ell$, $\nxt(\ell)$
will either be $\set{(\ell',\true)}$ or $\set{(\ell_1,c),(\ell_2,\neg
  c)}$. We do not give a precise definition of $\nxt(\cdot)$, but we
will use it when defining the semantics.

The semantics of $P_\epsilon$ will given as a finite-state,
parametrized DTMC $\sem{P_\epsilon}$. To define the parametrized DTMC
$\sem{P_\epsilon}$, we need to define the states and the transitions. 

\paragraph{States.} 
States of $\sem{P_\epsilon}$ will be of the form 
\[
(\ell, f_{Bool}, f_{\sdom}, f_{\s{int}}, f_{\s{real}}, C).
\] 
Informally, $\ell \in \s{L}$ is the label of the statement to be
executed, $f_{Bool}$, $f_{\sdom}$, $f_{\s{int}}$, and $f_{\s{real}}$
are partial functions assigning ``values'' to program variables (of
appropriate type), and $C$ is a collection of inequalities among
program variables that hold on the current computational path. Both
$f_{Bool}$ and $f_{\sdom}$ are valuations for the appropriate set of
variables, and so we have $f_{Bool} : \cB \pto \set{\true,\false}$ and
$f_{\sdom} : \cX \pto \sdom$. For real and integer variables, instead
of tracking exact values, we will track the expressions used in
assignments and parameters of (discrete) Laplace mechanisms used in
random assignments. Therefore, we have $f_{\s{int}}: \cZ \pto \s{zExp} \cup
(\s{Const} \times \sdom)$ and $f_{\s{real}}: \cR \pto \s{rExp} \cup
(\s{Const} \times \sdom)$. Finally, $C \subseteq \s{bExp} \cup
\set{\neg e\: |\: e \in \s{bExp}}$. It follows immediately that the
set of states of $\sem{P_\epsilon}$ is finite.

\paragraph{Well-Formed States.}
The functions $f_*$ (for $* \in \{Bool,\sdom,$ $\s{int},\s{real}\}$)
assign values to program variables that have been assigned during the
computation thus far. Since we assume variables in {\ourlang} program
are defined before they are used, if a variable $\iv'$ appears in
$f_{\s{int}}(\iv) \in \s{zExp}$, then $f_{\s{int}}(\iv')$ must be
defined. A similar condition holds for real variables. The comparisons
in $C$ are also relationships that must hold on the current path, and
so all variables participating in it must be defined. If a state
satisfies these consistency properties between $f_{\s{int}}$,
$f_{\s{real}}$, and $C$, we will say it is \emph{well-formed}. 
All
reachable states in $\sem{P_\epsilon}$ will be well-formed. So when we
define transitions we will assume that the states are well-formed.

\paragraph{Initial States.}
Let $\ell_{\fb{in}}$ be the label of the first statement
$P_\epsilon$. Let $C^{\fb{in}} = \emptyset$, and let
$f_{Bool}^{\fb{in}}$, $f_{\s{int}}^{\fb{in}}$, and
$f_{\s{real}}^{\fb{in}}$ be partial functions with an empty domain. An
initial state of $\sem{P_\epsilon}$ will be of the form
$(\ell_{\fb{in}}, f_{Bool}^{\fb{in}}, f_{\sdom}^{\fb{in}},
f_{\s{int}}^{\fb{in}}, f_{\s{real}}^{\fb{in}}, C^{\fb{in}})$, where
$f_{\sdom}^{\fb{in}}$ is defined only on the input variables; the
values given to these variables by $f_{\sdom}^{\fb{in}}$ will be the
``initial input value''.

We will now define the semantics of transitions in
$\sem{P_\epsilon}$. Let us fix a state $\state = (\ell, f_{Bool},
f_{\sdom}, f_{\s{int}}, f_{\s{real}}, C)$. Transitions out of $\state$
will be defined based on  the effect of executing the
statement labeled $\ell$, and so its definition will depend on this
statement. We handle each case below.

\paragraph{$\sdom$ assignments.}
Let $\nxt(\ell) = \set{(\ell',\true)}$ and let $\dv$ be the variable
being assigned in $\ell$. There are two cases to consider. First,
consider the case where $\dv$ is assigned a value for a $\sdom$
expression $e$. In this case, $\sem{P_\epsilon}$ will transition to
\[
(\ell',f_{Bool}, f_{\sdom}[\dv \mapsto f_{\sdom}(e)], f_{\s{int}},
f_{\s{real}}, C)
\]
with probability 1. The second case is when $\dv$ is assigned a random
value according to $\pexp{a\epsilon, F(\tilde{\dv}),e}$ or
$\chose(a\epsilon, \tilde{e})$. For $d \in \sdom$, let $\s{prob}(d)$
be the probability of $d$ (as a function of $\epsilon$) based on the
distribution; note, that these probabilities will depend on the value
of $f_{\sdom}(e)$ and $f_{\sdom}(\tilde{e})$. Then, $\sem{P_\epsilon}$
will transition to
\[
(\ell',f_{Bool},f_{\sdom}[\dv \mapsto d], f_{\s{int}}, f_{\s{real}},
C)
\]
with probability $\s{prob}(d)$.

\paragraph{Integer assignments.}
Let $\nxt(\ell) = \set{(\ell',\true)}$ and let $\iv$ be the variable
being assigned in $\ell$. Again there are two cases to
consider. First, consider the case where $\iv$ is assigned a value for
an integer expression $e$. In this case, $\sem{P_\epsilon}$ will
transition to 
\[
(\ell',f_{Bool}, f_{\sdom}, f_{\s{int}}[\iv \mapsto f_{\sdom}(e)],
f_{\s{real}}, C)
\]
with probability 1. Next, if $\iv$ is assigned a random value
according to $\DLap{a\epsilon,e}$, then $\sem{P_\epsilon}$ transitions
to
\[
(\ell',f_{Bool}, f_{\sdom}, f_{\s{int}}[\iv \mapsto (a,f_{\sdom}(e))],
f_{\s{real}}, C)
\]
with probability 1. Notice that we have a deterministic transition
even if the assignment samples from a discrete Laplace. The effect of
choosing randomly a value will get accounted for during Boolean
assignments.

\paragraph{Real assignments.}
Let $\nxt(\ell) = \set{(\ell',\true)}$ and let $\rv$ be the variable
being assigned in $\ell$. First, if $\iv$ is assigned a value for a
real expression $e$, $\sem{P_\epsilon}$ will transition to
\[
(\ell',f_{Bool}, f_{\sdom}, f_{\s{int}}, f_{\s{real}}[\rv \mapsto
  f_{\sdom}(e)], C)
\]
with probability 1. If $\iv$ is assigned a random value according to
$\Lap{a\epsilon,e}$, then $\sem{P_\epsilon}$ transitions to
\[
(\ell',f_{Bool}, f_{\sdom}, f_{\s{int}}, f_{\s{real}}[\rv \mapsto
  (a,f_{\sdom}(e))], C)
\] 
with probability 1. Again sampling according to Laplace is modeled
deterministically.

\paragraph{Boolean assignments.}
Again let $\nxt(\ell) = \set{(\ell',\true)}$ and let $\bv$ be the
variable being assigned in $\ell$. When $\bv$ is assigned the value of
Boolean expression $e$, $\sem{P_\epsilon}$ transitions to 
\[
(\ell', f_{Bool}[\bv \mapsto f_{Bool}(e)], f_{\sdom}, f_{\s{int}},
f_{\s{real}}, C)
\] 
with probability 1. The interesting case is when $\bv$ is assigned the
result of comparing expressions $e_1 \sim e_2$.
If the probability of all conditions in $C$ holding is $0$, then let $p_1$ be $0.$
Otherwise, let $p_1$
denote the
probability of $f_{\sdom}(e_1) \sim f_{\sdom}(e_2)$ holding given all
conditions in $C$ hold; notice that this probability depends on the
functions $f_{\s{int}}$ and $f_{\s{real}}$ that store the parameters
to various random sampling steps. Now $\sem{P_\epsilon}$ will
transition to 
\[
(\ell', f_{Bool}[\bv \mapsto \true], f_{\sdom}, f_{\s{int}},
f_{\s{real}}, C \cup \set{f_{\sdom}(e_1) \sim f_{\sdom}(e_2)})
\] 
with probability $p_1$, and it will transition to
$$\begin{array}{l}
(\ell', f_{Bool}[\bv \mapsto \false], f_{\sdom}, f_{\s{int}},
f_{\s{real}}, \\
\hspace*{0.5cm}C \cup \set{\neg (f_{\sdom}(e_1) \sim f_{\sdom}(e_2))})
\end{array} $$
with probability $1-p_1$. Thus, the effect of the probabilistic sampling
steps for integer and real variables gets accounted for when the
result of a comparison is assigned to a Boolean variable.

\paragraph{$\Ifs$ statement.}
In this case, $\nxt(\ell) = \set{(\ell_1,c),(\ell_2, \neg c)}$. If
$f_{Bool}(c) = \true$ then we transition to 
\[
(\ell_1, f_{Bool}, f_{\sdom}, f_{\s{int}}, f_{\s{real}}, C)
\]
with probability 1. On the other hand, if $f_{Bool}(c) = \false$ then
transition to 
\[
(\ell_2, f_{Bool}, f_{\sdom}, f_{\s{int}},f_{\s{real}}, C)
\]
with probability 1.

\paragraph{$\Whiles$ statement.}
Again let $\nxt(\ell) = \set{(\ell_1,c),(\ell_2, \neg c)}$. This case
is identical to the case of $\Ifs$ statement, and so is skipped.

\paragraph{$\ext$ statement.}
In this case we stay in state $\state$ with probability 1.

\paragraph{Equivalence of the two semantics.}
Let $\fb{in}$ be a valuation over input variables and $\fb{out}$ be a valuation over output variables.
 We let $\Prob_{DTMC}( P_\epsilon(\fb{in}) = \fb{out})$ denote the probability that
 $P_{\epsilon}$ outputs value $\fb{out}$, on the  input $\fb{in}$, under the
 DTMC semantics. This probability is defined to be the probability
 of reaching a state of the form $(\ell, f_{Bool}, f_{\sdom},
 f_{\s{int}},f_{\s{real}}, C)$ where $\ell$ is the label  of an $\ext$
 statement and $f_{\sdom}$ assigns the values given by $\fb{out}$ to
 output variables, from
 an initial state in which the values of the input variables is given
 by $\fb{in}$, in the DTMC  $\sem{P_\epsilon}.$ The following theorem states the
 equivalence of the natural semantics given in Appendix~\ref{app:semantics} to that
 of the DTMC semantics for {\ourlang} programs.

\begin{theorem}
\label{thm:equivalence}
For every $\epsilon>0$ and  {\ourlang} program $P_{\epsilon}$, and for
every pair of evaluations $\fb{in},\fb{out}$ to the input and output
variables respectively,  $\Prob_{DTMC}( P_\epsilon(\fb{in}) = \fb{out})=\Prob_{natural}( P_\epsilon(\fb{in}) = \fb{out}).$
\end{theorem}

\begin{proof}[Proof Sketch]
Let us fix an $\epsilon>0$ and a program $P_\epsilon.$ Then  $\sem{P_\epsilon}$ can be considered as a (non-paramaterized) DTMC.  
For any path $\rho=\state_1,\ldots, \state_i$ in the DTMC $\sem{P_\epsilon}$, let $\s{prob}(\rho)$
denote the product of the probabilities of all the transitions in $\rho.$
We call  $\rho$ an {\it initialized} path if it starts with an
initial state, and a {\it proper} path if $\s{prob}(\rho)>0.$ 
For any {\it initialized  path} $\rho$ of $\sem{P_\epsilon}$, let $\s{prob}_{\sdom}(\rho)$ be the product of all the transitions in $\rho$ that result from random assignments to $\sdom$ variables,
$\s{prob}_{\integers}(\rho)$ be the product of the probabilities that result from a comparison between integer variables and $\s{prob}_{\Reals}(\rho)$ be the product of the probabilities that result from a comparison between real variables. It is easy to see that $$\s{prob}(\rho)= \s{prob}_\sdom(\rho)\, \s{prob}_{\Reals}(\rho)\,  \s{prob}_{\integers}(\rho).$$

 We recall some of the notation as defined in Appendix~\ref{app:semantics}. Let $S$ be the set of states of $P_\epsilon$ in the natural semantics. A
state $s\in S$ is a tuple $(\ell, h_{Bool}, h_{\sdom}, h_\integers,
h_\Reals)$  denoting the label of a statement to be executed, and
the values of 
Boolean, $\sdom$, integer and real variables of $P_\epsilon.$ A discrete
state of $P_\epsilon$, $ds$, is a tuple $(\ell, h_{Bool}, h_{\sdom}, h_\integers)$ specifying
the label of the statement and the values of Boolean, $\sdom$ and integer variables of $P_\epsilon.$ 
For a state $s=(\ell, h_{Bool}, h_{\sdom}, h_\integers,
h_\Reals)$, let $disc(s)$ be the discrete state $(\ell, h_{Bool}, h_{\sdom}, h_\integers).$
A discrete state $ds$ is said to be {\it initial}  if $ds = disc(s_{\s{init}})$ where $s_{\s{init}}$ is the initial state of $S$.
 
A discrete execution
$\beta=ds_1,\ldots,ds_i$ of $P_\epsilon$ is a sequence of discrete states. The
discrete execution $\beta=ds_1,\ldots,ds_i$ is an {\it initialized}  if $ds_1$ is the initial
discrete state. For a discrete execution $\beta$ as given above, let
$\s{ext}(\beta)= \set{(s_1,\ldots,s_i)\st ds_j=disc(s_j), 1\leq j\leq i}$. It is
easy to see that, for any discrete execution $\beta$ of length $i$,
$\s{ext}(\beta)$ is in $\Sigma_i$ (see Appendix~\ref{app:semantics}) , i.e., is
measurable. For a discrete computation $\beta$, of length $i>0$, let 
$\s{pr}(\beta)= \phi_i(\s{ext}(\beta))$ where $\phi_{i}$ is the probability
function  defined on the measure space $(S^i,\Sigma_i)$ in Appendix~\ref{app:semantics}.
If $\s{pr}(\beta)>0$ then we call $\beta$ a {\it proper} discrete execution of $P_\epsilon.$  

Consider an initialized proper discrete execution $\beta$ of length $i$, as given above,    
where
$ds_j=(\ell_j,h^j_{\sdom},h^j_{Bool},$ $h^j_\integers)$ for $1\leq j\leq
i$. 
It can be shown that  
there exists a unique initialized
path $\rho_\beta=\state_0,\ldots,\state_i$
in the DTMC $\sem{P_\epsilon}$ corresponding to $\beta$ such that for each $j$, 
\begin{enumerate}
\item the state
$\state_j=(\ell_j,f^j_{\sdom},f^j_{Bool},f^j_{\s{int}},$ $f^j_{\s{real}},
C_j)$ 
for some appropriate $f^j_{\sdom}$, $f^j_{Bool},f^j_{\s{int}},$ $f^j_{\s{real}}$ and $C_j$, and 
\item  $f^j_{Bool}(\bv)=h^j_{Bool}(\bv)$ ($f^j_{\sdom}(\dv)=h^j_{\sdom}(\dv)$ respectively) whenever $f^j_{Bool}(\bv)$ ($f^j_{\sdom}(\dv)$ respectively)  is defined. 
\end{enumerate}
Let $H$ be the
function mapping initialized proper discrete executions of $P_\epsilon$ to
corresponding initialized paths in $\sem{P_\epsilon}$, as specified above.

For an initialized proper discrete execution $\beta$ of length $i$ as above,
we define a number $p_j$  for each $j \leq i$ as follows. For $ds_j=(\ell_j,h^j_{\sdom},h^j_{Bool},h^j_\integers)$, let $\widehat {ds_j} = (\ell_j,h^j_{\sdom},h^j_{Bool},h^j_\integers,h^{\fb{in}}_\Reals)$ 
where $h^{\fb{in}}_\Reals$ is the function that maps each real variable of $P_\epsilon$ to $0.$
Let $K_\epsilon $ be the Markov kernel as defined in Appendix~\ref{app:semantics}.
If $j>1$ and $\ell_{j-1}$ is the label of an assignment to an integer variable that samples from a discrete Laplacian variable 
then $p_j = K_\epsilon (\widehat {ds_{j-1}}, \set {\widehat {ds_{j}}})$, otherwise $p_j=1.$ Let $\s{pr}_\integers(\beta)=p_1\ldots p_{i}.$
It can be shown using the  definition of measure $\phi_i$ on $\Sigma_i$ (See Section~\ref{app:semantics})  that $$\s{pr}(\beta)= \s{prob}_\sdom(H(\beta))\, \s{prob}_{\Reals}(H(\beta)) \s{pr}_\integers(\beta).$$
Furthermore, $prob(H(\beta))>0$ if $pr(\beta)>0.$

Now consider any initialized proper path $\rho=\state_1,\ldots,\state_i$ in $\sem{P_\epsilon}.$ From the above observations,
it can be shown that $\s{prob}(\rho) = \sum_{u\in
  H^{-1}(\rho)}\phi_{i}(\s{ext}(\beta)).$ 
Now, the  theorem follows from this observation and the definitions of
$\Prob_{DTMC}( P_\epsilon(\fb{in})$ $ = \fb{out})$ and $\Prob_{natural}( P_\epsilon(\fb{in}) = \fb{out}).$
\end{proof}


\paragraph{Complexity.} Now, we bound the size of the state space of DTMC $\sem {P_\epsilon}$ as follows.  Let $m\:=\abs{\sdom}\:=
2N_\mx+1$ and $m'$ be the length of $P_\epsilon.$ Let $n_1,n_2,n_3,n_4$,
respectively, be the number $\sdom$ variables, 
  boolean variables, 
  integer
variables, 
and real variables occurring in $P_\epsilon.$
In a state $s$ of $M_P$,
the number of possible values for $f_{dom}$ is $\leq m^{n_{1}}$, the
number of possible values for $f_{bool}$ is $\leq 2^{n_{2}}.$ The
number of possible values for $f_{int}$ can be bounded as follows. An
integer variable can be assigned a Laplacian distribution whose
parameters are pairs of the form $(a\epsilon, e)$ where $e$ is an
expression over variables in $U\cap \cX$; the number of such pairs is
$\leq m_1m^{n_{1}}$ where $m_1$ is the number of values of $a$ in $P$
and $m^{n_1}$ is the bound on the number of values of $e.$ An integer
variable can also be assigned a linear expression over integer
variables with coefficients that are integer constants or expressions
over $\sdom$ variables; the number of such linear combinations is
$\leq m_2m^{n_{1}}$ where $m_2$ is the number of such expressions
appearing in $P$. Since, $m_1+m_2\leq m'$, we see that number of
values that an integer variable can be mapped to is $\leq
m'm^{n_{1}}.$ Hence the number of possible values for $f_{int}$ is
$\leq (m'm^{n_{1}})^{n_3}.$ By a similar reasoning we observe that the
number of possible values for $f_{real}$ is $\leq
(m'm^{n_{1}})^{n_4}.$ Now we bound the number of values for $C$ as
follows. The only places where comparisons appear are on the right
hand sides of assignments to boolean variables. In each such
assignment we have comparisons over linear expressions of integer and
real variables ; such comparisons also have integer constants and
$\sdom$ variables appearing in them. Since the number of integer
constants is $\leq m'$ and the number of valuations to $\sdom$
variables $\leq m^{n_{1}}$, we get that the number of possible
comparions is $\leq m' m^{n_{1}}.$ Since $C$ is a subset of such
comparisons, the number of possible values for $C$ is $\leq 2^{(m'
  m^{n_{1}})}.$ Now, the number of states is bounded by the product of
possible values to each component of a state, which is seen to be
$O(2^n.n^{n_1+n_2+n_3+n_4})$ where $n= m'm^{n_1}.$


\section{ {\ourlang} programs are  finite, definable, parametrized DTMCs}
\label{app:dtmc-def}

We show the proof of Theorem~\ref{thm:semantics}, namely that for any
{\ourlang} program $P_\epsilon$, $\sem{P_\epsilon}$ is a finite,
definable, parametrized DTMC.

\begin{proof}
From our definition of the DTMC semantics (Appendix~\ref{app:dtmcsemantics}),
it follows that $\sem{P_\epsilon}$ is a finite parameterized DTMC. We
now show that it is definable also.  In order to show this, we have to
show that the transition probabilities of $\sem{P_\epsilon}$ are
definable. Observe that, by definition, the transition probabilities
of $\chose(a\epsilon,\tilde{E})$ construct are definable. The other
probabilistic transitions arise as a result of comparison between random
variables of the same sort or from using the exponential mechanism. These
transition probabilities turn out to be from a special class of definable functions.  We
define this form next.

\begin{definition}
Let $p(\epsilon)= \sum^{m}_{i=1} a_i \epsilon^{n_i} \eulerv{\epsilon q_i}$ where each $a_i$ is a rational number, $n_i$ is a natural number and $q_i$ is a non-negative rational number. 
We shall call all such expressions \emph{pseudo-polynomials} in $\epsilon.$ 
Given a real number $b>0$ and
 a pseudo-polynomial $p(\epsilon)$, $p(b)$ is the real number  obtained by substituting $b$ for  $\epsilon.$  
The ratio of two pseudo-polynomials in $\epsilon$, $\frac {p_1(\epsilon)} {p_2(\epsilon)},$ shall be called a  \emph{pseudo-rational function} in $\epsilon$ if $p_2(b)\ne 0$ for all real $b>0.$  Given a real number $b>0$ and a pseudo-rational function 
$rt(\epsilon)= \frac {p_1(\epsilon)} {p_2(\epsilon)},$ $rt(b)$ is defined to be $ \frac {p_1(b)} {p_2(b)}$.  

\end{definition}
Observe that a pseudo-rational function $rt$ defines a function $f_{rt}$ from the set of strictly positive reals to the set of reals. We will henceforth confuse $f_{rt}$ with $rt.$ 
Pseudo-rational functions are easily seen to be closed under addition and multiplication.  
\begin{proposition}
Each pseudo-rational function $rt$ is definable in  the theory $\threal.$  
\end{proposition}
\begin{proof}
Let $rt(\epsilon)= \frac{\sum^{m}_{i=1} a_i \epsilon^{n_i} \eulerv{\epsilon q_i}}
                              {\sum^{m'}_{i=1} a'_i \epsilon^{n'_i} \eulerv{\epsilon q'_i}}.$ Let $N$ be the least common multiple of all denominators of $q_i,q_i'.$     Let  $p_i = q_i N$ and $p_i'=q_i' N.$ Let $a$ be the least common multiple of all denominators of $a_i,a_i'.$ Let $b_i=a a_i $ and $b_i' = a a_i'.$
                              It is easy to see that $rt$ is definable by the formula $\phi(x):$
                              $$\phi(x) \equiv  \forall z. ((x \sum^{m'}_{i=1} b'_i \epsilon^{n'_i} z^{p'_i} =  \sum^{m}_{i=1} b_i \epsilon^{n_i} z^{p_i}) \wedge (z^N =\eulerv \epsilon) \wedge (z>0)).  $$
 Note that in the above formula,   $z$ is the $N$th root of $\epsilon.$                           
\end{proof}

Now, it follows from our restriction on our scoring functions, namely that they take values in rationals,  that the transition probabilities in exponential mechanism are pseudo-rational functions that can be computed.

Let us now consider the case of comparison between random variables.  Let $state=(\ell, f_{Bool}, f_{\sdom}, f_{\s{int}}, f_{\s{real}}, C)$ of $\sem{P_\epsilon}$ be a 
state of $\sem{P_\epsilon}.$ 
 Recall that  when we compare random variables in $state$, we add a new linear comparison $e$ to $C$. 
Further, in order to compute transition probabilities, we compute the conditional probability that the set of linear comparison $C\cup e$ is true given that $C$ is true. For this, it suffices to show that we can compute the probability that the set of linear comparisons $C$ is true and the probability $C\cup e$ is true. We make the following observations:
\begin{itemize}
\item Since every random variable must be defined before it is used, we can simplify  $C$ and $C\cup e$ to only refer to program variables that were used in random assignments.  

\item All our random assignments sample from independent random variables. 
Since we never compare integer and real random variables, it suffices to compute the probability that a system of linear comparisons over integers with integer coefficients hold and the probability that  a system of linear comparisons over reals with rational coefficients hold. We will now show that these probabilities can be computed and are pseudo-rational functions.  

\item
In order to compute the probability that  a system of linear comparisons over reals with rational coefficients hold, we only need to consider systems of linear  inequalities. Clearly any equality $u_1= u_2$ can be written as two inequalities, $ u_1 \leq u_2$ and $u_2 \leq u_1.$ If a comparison in $C$ is $u_1 \ne u_2$ then we can consider 
the systems   $C_1=(C \setminus \set{u_1 \ne u_2})  \cup \set {u_1<u_2} $ and  $C_2= (C \setminus \set{u_1 \ne u_2})  \cup \set {u_2<u_1} , $ compute  probabilities  of $C_1$ and $C_2$ separately and add them up to compute the probability that $C$ holds. Thus,  without loss of generality we can assume that $C$ consists of only linear inequalities. 
\end{itemize}


\paragraph{\textbf{Probability of system of linear inequalities over integers.}}
Let $\overline Z= (Z_1,\ldots, Z_n) $ be a discrete random variable taking values in $\integers^n.$ Consider a finite system of linear inequalities $C$ with integer coefficients and with $n$ unknowns $Z_1,\ldots,Z_n$. A solution of $C$ is a tuple $\overline b = (b_1,\ldots,b_n) \in \integers^n$ such that all inequalities in $C$ are satisfied when each $Z_j\in C$ is replaced by $b_j.$ Let $sol(C)\subseteq \integers^n$ denote the set of all solutions of $C.$
The probability that $\overline Z$ satisfies $C$ is said to be the probability of the event
 $E=\set{\overline Z = \overline b   \mbar  \overline b \mbox{ is a solution of }   C }.$ 
 We denote this probability by $\Prob(\overline Z \models C).$ We have the following:
 
 \begin{lemma}
 \label{lem:intrandom}
Let $C$ be a   finite system of linear inequalities with integer coefficients and with $n$ unknowns $Z_1,\ldots,Z_n$.
Let  $Z_j= \DLap{a_j\epsilon, \mu_1}$, $\ldots$, $Z_n= \DLap{a_n\epsilon, \mu_n}$ be mutually independent discrete Laplacians such that  for each $1 \leq j\leq n$, $a_j$ is a strictly positive rational number and $\mu_j$ is an integer. Let $\overline Z=(Z_1,\ldots, Z_n)$. 
There is  a pseudo-rational function $rt_{\overline Z,C}$ in $\epsilon$ such that   $\Prob(\overline Z \models C)=rt_{\overline Z, C}.$ 
The function $rt_{\overline Z,C}$ can be computed from $C, (a_1, \mu_1), \ldots, (a_n, \mu_n).$\end{lemma}
\begin{proof}
For, each $1\leq j\leq n,$ consider $Y_j=\DLap{a_j\epsilon,0}.$ It is easy to see that $Z_j$ has the same distribution as $Y_j+\mu_j.$ Now consider the system of inequalities $C'$ in which each $Z_j$ is replaced by $Y_j+\mu_j.$ Let  $\overline Y=(Y_1,\ldots, Y_n).$ It is easy to see that $\Prob(\overline Z \models C)= \Prob(\overline Y \models C').$ 
This observation implies that it suffices to prove the Lemma in the special case that each $\mu_j=0.$ Thus, for the rest of the proof we assume that each $\mu_j=0.$

Now, consider a set $\pos\subseteq \set{1,\ldots,n}.$ Let $C_{\pos}$ be the system of inequalities $C \cup \set{Z_j\geq 0 \mbar j\in \pos} \cup  \set{Z_j< 0 \mbar j\not\in \pos}.$ It is easy to see that the set of solutions of $C$ is the \emph{disjoint} union $\cup_{\set{\pos\subseteq {1,\ldots,n}}} C_{\pos}.$ Thus, it suffices to the prove that for each $\pos\subseteq \set{1,\ldots,n},$ $\Prob(\overline Z \models C_\pos)$ is a pseudo-rational function that can be computed. 

Consider the system of inequalities  $C'_\pos$ obtained from $C_\pos$ by replacing each $Z_j$    by $Y_j$ for $j\in\pos$ and by $-Y_j$ for $j\not\in \pos.$  Let  $Y=(Y_1,\ldots, Y_n).$ 
From the fact that Laplacians are symmetric distributions, it follows each $Y_j$ has the same distribution as $Z_j$.  
Thus,  $\Prob(\overline Z \models C_\pos)= \Prob(\overline Y \models C'_\pos).$  Observe that the set of solutions of $C'_\pos$ are a subset of $\Nats^n$. 
Without loss of generality, we can also assume that the terms in each inequality  of $C'_\pos$ are rearranged so that the constant terms in $C'_\pos$ and the coefficients of  the variables $Y_j$ are  natural numbers, ie, non-negative integers.  

Therefore,  $C'_\pos$ is a system of linear inequalities with natural number coefficients.  We are interested in solutions of $C'_\pos$ over  natural numbers. For such system of inequalities, the set of solutions can be written as a \emph{disjoint} union of \emph{simple linear sets}~\cite{SimpleLinear}; a set $S\subseteq \Nats^n$ is said to be \emph{linear}  if there are tuples 
$\overline {b}_0,\overline{p}_1,\ldots,\overline{p}_m\in \Nats^n$ such that $S=\set{\overline{b}_0+\sum^m_{i=1} k_i \overline{p}_i \mbar \mbox{ for each i, }  k_i\in \Nats }$ and \emph{simple} if each $\overline b\in S$ has a unique representation 
as a sum $\overline{b}_0+\sum^m_{i=1} k_i \overline{p}_i.$ $\overline{b}_0$ is said to be the offset of  $S$ and $\overline{p}_1,\ldots,\overline{p}_m$ the periods of $S.$ From the fact that the set of solutions of $C'_\pos$ can be written as a \emph{disjoint} union of \emph{simple linear sets}, it follows that it suffices to show that $\Prob(\overline Y \in S \mbar S \mbox{ is simple linear})$ is a pseudo-rational function in $\epsilon$. 
In order to show this we need a couple of additional notations. 

For two $n$-tuples $\overline x= (x_1,\ldots,x_n)$ and $\overline y= (y_1,\ldots,y_n),$  $\overline x \cdot \overline y$ will denote the sum $\sum^n_{j=1} x_jy_j.$
Secondly, we will denote the tuple $(a_1,\ldots,a_n)$ by $\overline a.$
 
Fix a simple semilinear set $S.$ Let $\overline{b}_0$  be its offset and $\overline{p}_1,\ldots,\overline{p}_m$ its periods. 
Let $\kappa = \prod^n_{i=1}  \frac {1-\euler^{-a_i\epsilon}} 
    {1+\euler^{-a_i\epsilon}}.$ 
From the fact that each $\overline b\in S$ has a unique representation as a sum $\overline{b}_0+\sum^m_{i=1} k_i \overline{p}_i,$ it follows that 
$$\begin{array}{lcl}
\Prob(\overline Y \in S) &=& \sum_{k_1\in \Nats}\cdots \sum_{k_m\in \Nats}  \Prob(\overline Y = \overline{b}_0+\sum^m_{i=1} k_i \overline{p}_i )\\
                             &=& \sum_{k_1\in \Nats}\cdots \sum_{k_m\in \Nats}   
  \kappa\;  \eulerv  {-\epsilon (\overline{b}_0\cdot \overline a+ k_1 \overline{p}_1\cdot \overline a +\cdots+ k_m \overline{p}_m\cdot \overline a)}\\
                             &=& \kappa\; (\eulerv{-\epsilon \overline{b}_0\cdot \overline a}) \\
                             && \hspace*{0.5cm}(\sum_{k_1\in \Nats}\eulerv{-\epsilon k_1 \overline{p}_1 \cdot \overline a})\cdots (\sum_{k_m\in \Nats}\eulerv{-\epsilon k_m \overline{p}_1 \cdot \overline a})  \\
                             &=&  \kappa\; (\eulerv{-\epsilon \overline{b}_0 \cdot \overline a})(\frac 1 {1-\eulerv {-\epsilon \overline{p}_1\cdot \overline a}}) \cdots (\frac 1 {1-\eulerv {-\epsilon \overline{p}_m\cdot \overline a}}) 
\end{array}$$ 
The latter is clearly a pseudo-rational function in $\epsilon.$ 
\end{proof}

\paragraph{\textbf{Probability of system of linear inequalities over reals.}}
Let $\overline R= (R_1,\ldots, R_n) $ be a continuous random variable taking values in $\Reals^n.$ Consider a finite system of linear inequalities $C$ with rational coefficients and with $n$ unknowns $R_1,\ldots,R_n$. As in the case of discrete random variables , we can define $sol(C)\subseteq \Reals^n$, the set of solutions, and  $\Prob(\overline R \models C),$ the probability that $\overline R$ satisfies C. We have the following result. 
\begin{lemma}
  \label{lem:realrandom}
Let $C$ be a   finite system of linear inequalities with rational coefficients and with $n$ unknowns $R_1,\ldots,R_n$.
Let  $R_1= \Lap{a_1\epsilon, \mu_1}$, $\ldots$, $R_n= \Lap{a_n\epsilon, \mu_n}$ be mutually independent Laplacian doistributions such that  for each $1 \leq j\leq n$, $a_j$ is a strictly positive rational number and $\mu_j$ is a rational number. Let $\overline R =(R_1,\ldots, R_n)$. 
There is  a pseudo-rational function $rt_{\overline R,C}$ in $\epsilon$ such that   $\Prob(\overline  R \models C)=rt_{\overline R, C}.$ 
The function $rt_{\overline R,C}$ can be computed from $C, (a_1, \mu_1), \ldots, (a_n, \mu_n).$\end{lemma}
\begin{proof}
As in the proof of Lemma~\ref{lem:intrandom}, it suffices to consider the case when each $\mu_i=0$ and to  
show that the probability measure of the set $Sol= sol(C) \cap \set{(b_1,\ldots,b_n) \mbar b_i\in \Reals^{>0} }$ is a computable pseudo-rational function. 

Since $\overline R$ is continuous, we can also assume that each inequality is of the form $\leq.$ This is because the measure of any set in $\Reals^n$ that satisfies a linear equation over $n$ unknowns   $R_1,\ldots,R_n$ is $0.$
 There are computable finite sets $S_1,\ldots,S_m $ such that (See~\cite{Polytope})
\begin{enumerate}
\item $Sol= S_1\cup\ldots S_m,$
\item The measure of the $S_i \cap S_j $ is $0$ for $i\ne j,$ and
\item  Each $S_i$ is a positive repetitive polyhedra. $S\subseteq (\Reals^{>0})^n$ is said to be a positive repetitive polyhedra if  there are constants $h^-_0, h^+_0$  and functions $h^-_1(x_1),\,$ $ h^+_1(x_1),\, h^-_2(x_1,x_2),\, $  $h^+_2(x_1,x_2),\, \ldots,\, h^-_{n-1} (x_1,x_2,\ldots x_{n-1}),$ 
 $h^+_{n-1} (x_1,x_2,\ldots x_{n-1})$ such that 
\begin{itemize}
\item $S_i =\begin{array}{l}\set{(x_1,\ldots, x_n)\;\mbar\; h^-_0 \leq x_1\leq  h^+_0,   \ldots, \\ \hspace*{0.6cm} h^-_{n-1} (x_1,\ldots x_{n-1}) \leq x_n \leq h^+_{n-1} (x_1,\ldots x_{n-1})}.\end{array}$
\item $h^-_0$ is a rational number $\geq 0.$ 
\item $h^+_0$ is either $\infty$ or a rational number.
\item For each $1\leq j\leq n,$ $h^-_j$ is a linear function  in its arguments.  In the latter case,  $h^-_j$ has  rational coefficients.
\item  For each $1\leq j\leq n,$ $h^+_j$ is either $\infty$ or a linear function  in its arguments.  $h^+_j$ has  rational coefficients in the latter case. 
\item For each $1\leq j\leq n,$ $h^-_j \ne h^+_j.$
\end{itemize}

\end{enumerate}
Thanks to conditions (1) and (2) above, it suffices to show that for any positive repetitive polyhedra $S$, the probability measure  of the event $\set{\overline R=\overline b \mbar \overline b \in S}$ is a pseudo-rational function. 

Fix $S$ and let $h^-_0, h^+_0, h^-_1,h^+_1,\ldots, h^-_{n-1}, h^+_{n-1}$ be as above, 
 The measure of the event $\set{\overline R=\overline b \mbar \overline b \in S}$ can be computed using  the nested integral  $$ F= \int^{h^+_0}_{h^-_0} f_{a_1} (x_1)\int^{h^+_1}_{h^-_1} f_{a_2} (x_2 ) \cdots \int^{h^+_{n-1}}_{h^-_{n-1}} f_{a_n} (x_n )\,  d{x_{n}} \cdots d{x_1} $$
 where $f_{a_i}(x_i)=\frac {a_i \epsilon} 2 \eulerv{- a_i\epsilon  x_i}$ is the pdf of $R_i$ (we always have that $x_i\geq 0$) and the arguments of $h^+_i,h^-_i$ are omitted for readability.  
 
For $1\leq j\leq n$, let  $I_{j}$ be the nested integral $$I_{j}=\int^{h^+_{j-1}}_{h^-_{j-1}} f_{a_j} (x_j ) \cdots \int^{h^+_{n-1}}_{h^-_{n-1}} f_{a_n} (x_n )\,  d{x_{n}} \cdots d{x_j}.$$
We claim by induction on  $k=n-j$ that $I_j$ is a finite sum of terms of the form $$ a  \epsilon ^m \eulerv {b \epsilon} (x^{m_1}_1 \eulerv{\epsilon b_1 x_1}) \ldots  (x^{m_{j-1}}_{j-1} \eulerv{\epsilon b_{j-1} x_{j-1}})$$
where $a,b, b_1,\ldots, b_{j-i} $ are rational numbers (including negative numbers), $m$ is an integer, and $m_1,\ldots, m_{j-1}$ are natural numbers.  We will assume that the sum is always presented in simplest form, namely, that all cancellations have already taken place in the sum. 

 Clearly the claim is true when $k=0.$ Suppose that the claim is true for $k=k_0.$  Let  $j_0=n-k_0.$
Suppose $$ w= a  \epsilon ^m \eulerv {b \epsilon} (x^{m_1}_1 \eulerv{\epsilon b_1 x_1}) \ldots  (x^{m_{j_0-1}}_{j_0-1} \eulerv{\epsilon b_{j_0-1} x_{j_0-1}})$$
  is a summand in $I_{j_0}.$ Let $k=k_0+1$ and $j=n-k=n-k_0-1= j_0-1.$

Consider the indefinite integral $$\begin{array}{lcl}
J&=&\int f_{a_{j_0-1}} w\, dx_{j_0-1}\\
              &=&   \int \frac {a_{j_0-1} \epsilon} 2 \, \eulerv{- a_{j_0-1}\epsilon x_{j_0-1}}\, w\, dx_{j_0-1}\\
              &=&   \frac{ a_{j_0-1}} 2  \epsilon ^{m+1} \eulerv {b \epsilon} (x^{m_1}_1 \eulerv{\epsilon b_1 x_1}) \ldots  (x^{m_{j_0-2}}_{j_0-2} \eulerv{\epsilon b_{j_0-2} x_{j_0-2}}) \\
               && \hspace*{0.8in}\int  x^{m_{j_0-1}}_{j_0-1} \eulerv{\epsilon (b_{j_0-1} - a_{j_0-1}) x_{j_0-1}} d{x_{j_0-1}} \end{array} $$
       Let $$J'=\int  x^{m_{j_0-1}}_{j_0-1} \eulerv{\epsilon (b_{j_0-1} - a_{j_0-1}) x_{j_0-1}} d{x_{j_0-1}} .$$        
             Now, if   $b_{j_0-1} - a_{j_0-1}=0$ then  $$J'=\frac {x^{m_{j_0-1}+1}_{j_0-1}} {m_{j_0-1+1}}.$$
  If $b_{j_0-1} - a_{j_0-1} \ne 0$ then by doing a change of variables $t=(b_{j_0-1} - a_{j_0-1}) \epsilon x_{j_0-1},$ it is not too hard to show that 
  $$J'= \sum^{m_{j_0-1}}_{k=0} {c_k} {\epsilon^{t_k}}  x_{j_0-1}^k   \eulerv{\epsilon (b_{j_0-1} - a_{j_0-1}) x_{j_0-1}}$$
where $c_k$ is a rational number  and $t_k$ an integer for each $k.$

Thus, the indefinite integeral $J$ is a sum, each of whose terms is of the form   $$a'\epsilon ^{m'} \eulerv {b' \epsilon} (x^{m'_1}_1 \eulerv{\epsilon b'_1 x_1}) \ldots  (x^{m'_{j_0-1}}_{j_0-1} \eulerv{\epsilon b'_{j_0-1} x_{j_0-1}}).$$ If $h^-_{j_0-2}$ and $h^+_{j_0-2}$ are linear functions, we get immediately that $I_{j}= \int^{h^+_{j_0-2}}_{h^-_{j_0-2}}  f_{a_{j_0-1}} w\, dx_{j_0-1} $ is of the right form. The induction step follows in this case.  

If $h^+_{j_0-2} =\infty,$ and each $b'_j$ in a summand of $J$ is strictly negative, then it is also easy to see that the induction step follows. Apriori, it seems that there might be a problem when $b'_j\geq 0$ as in this case, $I_{j}$ will evaluate to either $\infty$ or $-\infty.$  This, however, will contradict the fact that the nested integral $F$ defines probability of an event (and hence is bounded above by $1$).  Thus, if $h^+_{j_0-2} =\infty$ then $b_j$ must be strictly negative. 

The claim immediately implies that the measure of the set $Sol= sol(C) \cap \set{(b_1,\ldots,b_n) \mbar b_i\in \Reals^{>0} }$ is a pseudo-rational function.
\end{proof}
\end{proof}

\section{Reachability in Parametrized DTMCs}
\label{app:dtmc}

In this section we will prove Lemma~\ref{lem:dtmc}. Let us first
recall how reachability probabilities are computed in
(non-parametrized) finite-state DTMCs. Recall that a
(non-parametrized) DTMC is a pair $(Q,\delta)$ where $Q$ is a finite
set of states, and $\delta: Q \times Q \to [0,1]$ is such that for
every $q \in Q$, $\sum_{q' \in Q} \delta(q,q') = 1$. So in a DTMC the
transition probabilities are fixed, and are not functions of a
parameter. The probability of reaching a set of states $Q' \subseteq
Q$ from a state $q_0$ is computed by solving a more general problem,
namely, the problem of computing the probability of reaching $Q'$ from
each state $q \in Q$. Let the variable $x_q$ denote the probability of
reaching $Q'$ from state $q$. One simple observation is that if $q \in
Q'$ then $x_q = 1$. Second, if $Q_0$ denotes the set of all states
from which $Q'$ is not reachable in the underlying graph (i.e., one
where we ignore the probabilities and just have edges for all
transitions that are non-zero), then $x_q = 0$ if $q \in Q_0$. Now the
set $Q_0$ can be computed by performing a simple graph search on the
underlying graph. For states $q \not\in (Q' \cup Q_0)$, we could write
$x_q$ as $x_q = \sum_{q'\in Q} \delta(q,q')x_{q'}$. This gives us the
following system of linear equations.
\[
\begin{array}{ll}
x_q = 1 & \mbox{if } q \in Q'\\
x_q = 0 & \mbox{if } q \in Q_0\\
x_q = \sum_{q'\in Q} \delta(q,q')x_{q'} & \mbox{otherwise}
\end{array}
\]
The above system of linear equations can be shown to have a unique
solution, with the solution giving the probability of reaching $Q'$
from each state $q$.

Now let us consider a parametrized DTMC $\cD = (\states,
\ptransf)$. Let $\varphi_{\state\state'}$ be a $\lngreal$ formula that
defines the function $\ptransf(\state,\state')$. Recall that in the
algorithm outlined in the previous paragraph, one crucial step is to
compute the set of states that have probability $0$ of reaching the
target set. This requires knowing the underlying graph of the DTMC,
i.e., knowing which transitions have probability 0 and which ones have
probability $> 0$. In a parametrized DTMC this is challenging because
the probability of transitions depends on the value of $\epsilon$, and
our goal is to compute the reachability probability as a function of
$\epsilon$. We will overcome this challenge by ``guessing'' the
underlying graph.

Let $C \subseteq \states \times \states$. We will construct a formula
$\varphi_C$ that will capture the constraints that reachablity
probabilities need to satisfy under the assumption that the
probability of edges in $C$ is $0$, and those outside $C$ is $>
0$. Based on the assumption that $C$ is exactly the set of 0
probability edges, we can compute the set $\states_0^C$ of states that
cannot reach $\states'$. The formula $\varphi_C$ will have variables
that will have the following intuitive interpretations ---
$p_{\state\state'}$ the probability of transitioning from $\state$ to
$\state'$; $x_\state$ the probability of reaching $Z'$ from state
$\state$.
\[
\begin{array}{rl}
\varphi_C = 
 & \hspace*{-0.1in}
  \bigwedge_{(\state,\state') \in C} (p_{\state\state'} = 0)
  \wedge
  \bigwedge_{(\state,\state') \not\in C} (p_{\state\state'} > 0)
  \wedge
  \bigwedge_{\state \in \states'} (x_{\state} = 1)\\
 & \hspace*{-0.1in} 
  \wedge
  \bigwedge_{\state \in \states_0^C} (x_{\state} = 0)
  \wedge
  \bigwedge_{\state \not\in (\states' \cup \states_0^C)}
           (x_{\state} = \sum_{\state'} p_{\state\state'}x_{\state'}).
\end{array}
\]
Notice that $\varphi_C$ is a formula in $\lngreal$. $\varphi_C$ can be
used to construct the formula we want. To construct the formula
$\varphi_{\state_0\states'}$ that characterizes the probability of
reaching $\states'$ from $\state_0$, we need to account for two
things. First, we need to ensure that $p_{\state\state'}$ is indeed
the probability of transitioning from $\state$ to $\state'$. Second,
we need to account for the fact that we don't know the exact set of
edges with probability $0$. Based on these observations, we can define
$\varphi_{\state_0,\states'}$ as follows.
\[
\varphi_{\state_0\states'} = [\exists x_\state]_{\state \neq \state_0}
  [\exists p_{\state\state'}]_{\state,\state' \in \states}
  \bigwedge_{\state,\state' \in \states} 
         \varphi_{\state\state'}(\epsilon,p_{\state\state'}) \wedge
  \left(\bigvee_{C \subseteq \states \times \states} \varphi_C\right)
\]
In the above definition of $\varphi_{\state_0\states'}$ all variables
except $x_{\state_0}$ (and $\epsilon$) are existentially
quantified. Notice, that $\varphi_{\state_0\states'}$ is in $\lngreal$
provided we pull all the quantifiers to get it in prenex form. Given
that $\states_0^C$ can be effectively constructed for any set $C$, the
above formula can also be computed for any parametrized DTMC $\cD$.

\section{Syntax of {\finlang} programs}
\label{app:finlang}

The syntax of  {\finlang} programs is presented in Figure~\ref{fig:BNF-finite}.
\begin{figure}
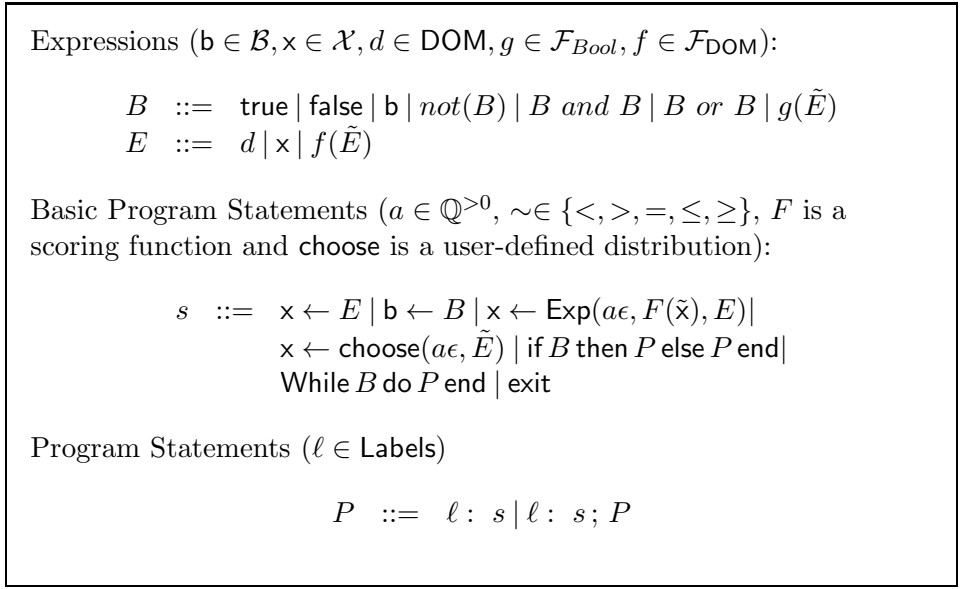

\begin{framed}
\raggedright

Expressions ($\bv\in \cB, \dv\in \cX, d\in \sdom, g\in \sF_{Bool}, f\in \sF_\sdom$):
\[
\begin{array}{lcl}
B&::=& \true \mbar \false \mbar \bv \mbar not(B) \mbar B\ and\ B \mbar B \ or \ B \mbar g(\tilde{E})\\
E&::=&   d \mbar \dv  \mbar f (\tilde{E}) \\
\end{array}$$

Basic Program Statements ($a\in \Rats^{>0}$, $\sim\in \set{<,>,=, \leq, \geq}$, $F$ is  a scoring function and $\chose$ is a user-defined distribution):
$$\begin{array}{lll}
 s&::= &      \dv\gets E \mbar  \bv\gets B \mbar      \dv\gets\pexp {a\epsilon, F(\tilde{\dv}), E} \mbar  \\
      &&   \dv\gets \chose(a\epsilon, \tilde{E}) \mbar 
               \ifstatement B P P \mbar \\ &&  \whilestatement B P \mbar     \ext\\
\end{array}
\]

Program Statements ($\ell \in \mathsf{Labels}$)
$$\begin{array}{lcl}
 P&::= &  \ell:\ s \mbar  \ell:\ s\, ;\,  P\\
\end{array} $$
%


\end{framed}
\caption{BNF grammar for {\finlang}. $\sdom$ is a finite discrete
  domain. $\sF_{Bool}$, ($\sF_{\sdom}$ resp) are set of functions that
  output Boolean values ($\sdom$ respectively).  $\cB,\cX$ are the
  sets of Boolean variables, and $\sdom$ variables,
  respectively. $\mathsf{Labels}$ is a set of program labels. For a
  syntactic class $S$, $\tilde{S}$ denotes a sequence of elements from
  $S$.}
\label{fig:BNF-finite}
\end{figure}

\subsection{A general semantic class of programs}
\label{app:generalsemantic}
Our methods imply decidability of checking differential privacy for a large
semantic class of programs (which include $\ourlang.$) 
A sufficient condition to ensure the decidability of checking
differential privacy is to consider programs with the property that
for each input, the probability distribution on the outputs is
definable in $\threal$: 
\begin{definition}
\label{def:definable dist}
A parametrized program $P_\epsilon$ with inputs $\cU$ and outputs
$\cV$ is said to identify a \emph{definable distribution} on $\cV$ if
for each $\fb{in} \in \cU$ and $\fb{out} \in \cV$ the function
$\epsilon \mapsto \Prob(P_\epsilon(\fb{in}) = \fb{out})$ is
definable in $\threal$.

A parametrized program $P_\epsilon$ with inputs $\cU$ and outputs
$\cV$ is said to \emph{effectively} identify a definable distribution on
$\cV$ if there is an algorithm $\cA$ such that for each
$\fb{in} \in \cU$ and $\fb{out} \in \cV$, $\cA$ outputs a
formula $\varphi_{\fb{in},\fb{out}}(\epsilon,x)$ in $\lngreal$
that defines the function
$\epsilon \mapsto \Prob(P_\epsilon(\fb{in}) = \fb{out})$.
\end{definition}

We can conclude by a proof similar to the proof Theorem~\ref{thm:decidability}.

\begin{theorem}
\label{thm:sem-decide}
The {\problemone} and {\problemtwo} problems are decidable for
programs $P_\epsilon$ that effectively identify a definable
distribution, rationals $t\in \Rats^>0$ and definable functions $\delta$ (in the case of the
{\problemtwo} problem). Furthermore, if $P_\epsilon$  is not $(t\epsilon,\delta)$ differentially private for some admissible value of $\epsilon$
then we can compute a counter-example. 
\end{theorem}


\section{Detailed Experimental Results}
\label{app:experiments}

We implemented a simplified version of the algorithm, presented
earlier, for proving/disproving differential privacy of {\ourlang}
programs. Our tool {\tool}~\cite{github} handles loop-free programs,
i.e., acyclic programs. Programs with bounded loops (with constant
bounds) can be handled by unrolling loops. The tool takes in an input
program $P_\epsilon$ parametrized by $\epsilon$, and either proves
$P_\epsilon$ to be differentially private for all $\epsilon$ or
returns a counter-example. The tool can also be used to check
differential privacy for a given, fixed $\epsilon$, or to check for
$k\epsilon$-differential privacy for some constant $k$. The design of
the tool will be discussed in detail in Section~\ref{sec:working}.

\subsection{Examples}
\label{sec:ex}

We used various examples to measure the effectiveness of our
tool. These include SVT~\cite{lyu2016understanding, DR14}, Noisy
Maximum~\cite{DingWWZK18}, Noisy Histogram~\cite{DingWWZK18} and
Randomized Response~\cite{DNRRV09}. Pseudocodes for all variants of
these examples that we tried are given in this section for
completeness. Though the pseudo-codes don't strictly adhere to the
syntax of {\ourlang} programs, they can easily be rewritten to fit the
syntax.

\paragraph{Sparse Vector Technique (SVT)}
We looked at six different variants of the Sparse Vector Technique
(SVT).  Algorithms addressed as SVT1-6, are Algorithms 1-6
in~\cite{lyu2016understanding}, respectively. In these programs, the
array $q$ represents the input queries. The array $out$ represents the
output array, $\bot$ represents False and $\top$ represents True.  In
all our experiments, we set the threshold $T=0.$ SVT1 was previously
introduced in this paper as Algorithm~\ref{fig:SVT} on
page~\pageref{fig:SVT}. The adjacency relation $\Phi$ we used is given
by $(q_1,q_2) \in \Phi$ if and only if $\abs{q_1[i]-q_2[i]}\leq 1$ for
all $i$. While \subref{algo-svt1} and \subref{algo-svt2} are
differentially private, the other four variants are not. We will
present counter-examples for all four of these variants in
Section~\ref{sec:exp-res}. The pseudocode for the six variants of SVT
are given in Figures~\ref{fig:svtA} and~\ref{fig:svtB}.

\begin{figure*}[!htb]
\RestyleAlgo{boxed}
\begin{subfigure}{0.45\textwidth}
  \renewcommand\thesubfigure{SVT1}
  \caption{First Instantiation of SVT}\label{algo-svt1}
  \removelatexerror
  \begin{algorithm}[H]
  \DontPrintSemicolon
  \KwIn{$q[1:N]$}
  \KwOut{$out[1:N]$}
  $\rv_T \gets \Lap{ \frac {\epsilon} {2 \Delta} , T}$\;
  $count \gets 0$\;
  \For{$i\gets 1$ \KwTo $N$}
  {
    $\rv\gets \Lap{\frac \epsilon {4c\Delta} , q[i]}$\;
    $\bv\gets \rv \geq \rv_T$\;
    \uIf{$b$}{
      $out[i] \gets \top$\;
      $count \gets count + 1$\;
      \If{$count\geq c$} {\exit} 
      }
    \Else{
      $out[i] \gets \bot$}
  }
  \end{algorithm}
\end{subfigure}
\begin{subfigure}{0.45\textwidth}
  \renewcommand\thesubfigure{SVT2}
  \caption{Second Instantiation of SVT}\label{algo-svt2}
  \removelatexerror
  \begin{algorithm}[H]
  \DontPrintSemicolon
  \KwIn{$q[1:N]$}
  \KwOut{$out[1:N]$}
  $\rv_T \gets \Lap{ \frac {\epsilon} {2c \Delta} , T}$\;
  $count \gets 0$\;
  \For{$i\gets 1$ \KwTo $N$}
  {
    $\rv\gets \Lap{\frac \epsilon {4c\Delta} , q[i]}$\;
    $\bv\gets \rv \geq \rv_T$\;
    \uIf{$b$}{
      $out[i] \gets \top$, $\rv_T\gets \Lap{\frac \epsilon {2c\Delta} , T}$\;
      $count \gets count + 1$\;
      \If{$count\geq c$} {\exit} 
      }
    \Else{
      $out[i] \gets \bot$}
  }
  \end{algorithm}
\end{subfigure}
\caption{Sparse Vector Technique Algorithms}
\label{fig:svtA}
\end{figure*}

\begin{figure*}[!htb]
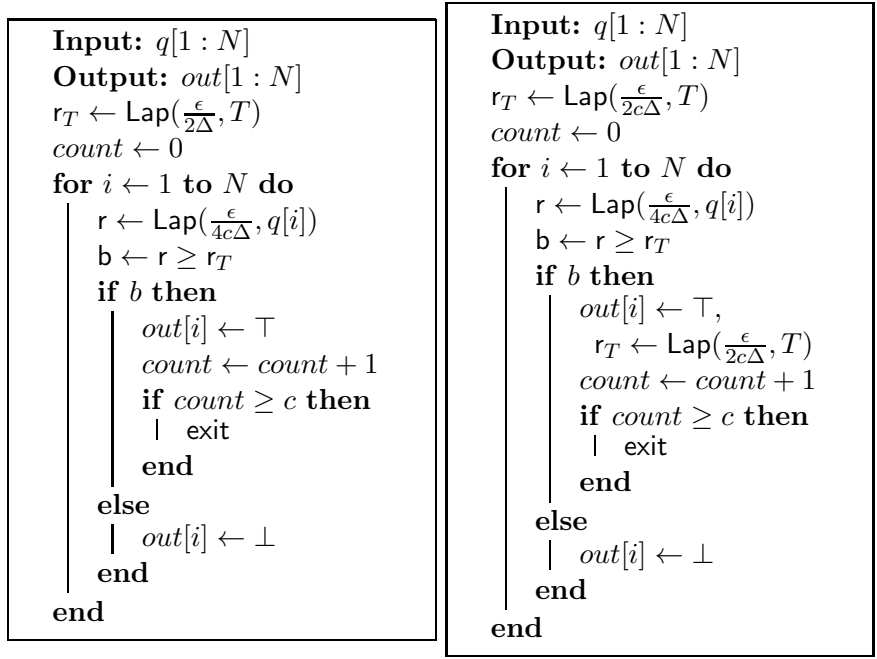

\RestyleAlgo{boxed}
\begin{subfigure}{0.45\textwidth}
  \renewcommand\thesubfigure{SVT3}
  \caption{Third Instantiation of SVT}\label{algo-svt3}
  \removelatexerror
  \begin{algorithm}[H]
  \DontPrintSemicolon
  \KwIn{$q[1:N]$}
  \KwOut{$out[1:N]$}
  $\rv_T \gets \Lap{ \frac {\epsilon} {2 \Delta} , T}$\;
  $count \gets 0$\;
  \For{$i\gets 1$ \KwTo $N$}
  {
    $\rv\gets \Lap{\frac \epsilon {2c\Delta} , q[i]}$\;
    $\bv\gets \rv \geq \rv_T$\;
    \uIf{$b$}{
      $out[i] \gets \Disc_{\seq}(\rv)$\;
      $count \gets count + 1$\;
      \If{$count\geq c$} {\exit} 
      }
    \Else{
      $out[i] \gets \bot$}
  }
  \end{algorithm}
\end{subfigure}
\begin{subfigure}{0.45\textwidth}
  \renewcommand\thesubfigure{SVT4}
  \caption{Fourth Instantiation of SVT}\label{algo-svt4}
  \removelatexerror
  \begin{algorithm}[H]
  \DontPrintSemicolon
  \KwIn{$q[1:N]$}
  \KwOut{$out[1:N]$}
  $\rv_T \gets \Lap{ \frac {\epsilon} {4 \Delta} , T}$\;
  $count \gets 0$\;
  \For{$i\gets 1$ \KwTo $N$}
  {
    $\rv\gets \Lap{\frac {3\epsilon} {4\Delta} , q[i]}$\;
    $\bv\gets \rv \geq \rv_T$\;
    \uIf{$b$}{
      $out[i] \gets \top$\;
      $count \gets count + 1$\;
      \If{$count\geq c$} {\exit} 
      }
    \Else{
      $out[i] \gets \bot$}
  }
  \end{algorithm}
\end{subfigure}
\begin{subfigure}{0.45\textwidth}
  \renewcommand\thesubfigure{SVT5}
  \caption{Fifth Instantiation of SVT}\label{algo-svt5}
  \removelatexerror
  \begin{algorithm}[H]
  \DontPrintSemicolon
  \KwIn{$q[1:N]$}
  \KwOut{$out[1:N]$}
  $\rv_T \gets \Lap{ \frac {\epsilon} {2 \Delta} , T}$\;
  \For{$i\gets 1$ \KwTo $N$}
  {
    $\rv\gets q[i]$\;
    $\bv\gets \rv \geq \rv_T$\;
    \uIf{$b$}{
      $out[i] \gets \top$\;
      }
    \Else{
      $out[i] \gets \bot$}
  }
  \end{algorithm}
\end{subfigure}
\begin{subfigure}{0.45\textwidth}
  \renewcommand\thesubfigure{SVT6}
  \caption{Sixth Instantiation of SVT}\label{algo-svt6}
  \removelatexerror
  \begin{algorithm}[H]
  \DontPrintSemicolon
  \KwIn{$q[1:N]$}
  \KwOut{$out[1:N]$}
  $\rv_T \gets \Lap{ \frac {\epsilon} {2 \Delta} , T}$\;
  \For{$i\gets 1$ \KwTo $N$}
  {
    $\rv\gets \Lap{\frac{\epsilon}{2\Delta},q[i]}$\;
    $\bv\gets \rv \geq \rv_T$\;
    \uIf{$b$}{
      $out[i] \gets \top$\;
      }
    \Else{
      $out[i] \gets \bot$}
  }
  \end{algorithm}
\end{subfigure}
\caption{Sparse Vector Technique Algorithms}
\label{fig:svtB}
\end{figure*}

\paragraph{Noisy Maximum}
Noisy maximum algorithms are a differentially private way to compute
different statistical measures for a given set of queries. Algorithms
addressed as NMax1-4 are Algorithms 5-8, respectively,
in~\cite{DingWWZK18}. Algorithms~\subref{algo-nmax5}
and~\subref{algo-nmax6} are mechanisms to compute the index of the
query with maximum value after adding a Laplacian (or exponential)
noise. Inputs $Q_1$ and $Q_2$ are considered adjacent iff
$\abs{{Q_1}[i]-{Q_2}[i]}\leq 1$ for all $i$. Under this relation,
Algorithms~\subref{algo-nmax5} and~\subref{algo-nmax6} are both
$\epsilon$-differentially private. Algorithms~\subref{algo-nmax7}
and~\subref{algo-nmax8} are variants to print the maximum value
instead of the index. These variants are shown to be not
differentially private in Section~\ref{sec:exp-res}.  The pseudocode
for these algorithms can be found in Figure~\ref{fig:nmax}.

\begin{figure*}[!htb]
\RestyleAlgo{boxed}
\begin{subfigure}{0.45\textwidth}
  \renewcommand\thesubfigure{NMax1}
  \caption{Correct Noisy Max with Laplacian Noise}\label{algo-nmax5}
  \removelatexerror
  \begin{algorithm}[H]
  \DontPrintSemicolon
  \KwIn{$q[1:N]$}
  \KwOut{$out$}
  \;
  NoisyVector $\gets []$\;
  \For{$i\gets 1$ \KwTo $N$}{
    NoisyVector[i] $\gets$ $\Lap{\frac{\epsilon}{2}, q[i]}$
  }
  out $\gets$ argmax(NoisyVector)\;
  \end{algorithm}
\end{subfigure}
\begin{subfigure}{0.45\textwidth}
  \renewcommand\thesubfigure{NMax2}
  \caption{Correct Noisy Max with Exponential Noise}\label{algo-nmax6}
  \removelatexerror
  \begin{algorithm}[H]
  \DontPrintSemicolon
  \KwIn{$q[1:N]$}
  \KwOut{$out$}
  \;
  NoisyVector $\gets []$\;
  \For{$i\gets 1$ \KwTo $N$}{
    NoisyVector[i] $\gets$ $\mathsf{Lap}^{+}(\frac{\epsilon}{2}, q[i])$
  }
  out $\gets$ argmax(NoisyVector)\;
  \end{algorithm}
\end{subfigure}
\begin{subfigure}{0.45\textwidth}
  \renewcommand\thesubfigure{NMax3}
  \caption{Incorrect Noisy Max with Laplacian Noise}\label{algo-nmax7}
  \removelatexerror
  \begin{algorithm}[H]
  \DontPrintSemicolon
  \KwIn{$q[1:N]$}
  \KwOut{$out$}
  \;
  NoisyVector $\gets []$\;
  \For{$i\gets 1$ \KwTo $N$}{
    NoisyVector[i] $\gets$ $\Lap{\frac{\epsilon}{2}, q[i]}$
  }
  out $\gets$ $\Disc_{\seq}$(max(NoisyVector))\;
  \end{algorithm}
\end{subfigure}
\begin{subfigure}{0.45\textwidth}
  \renewcommand\thesubfigure{NMax4}
  \caption{Incorrect Noisy Max with Laplacian Noise}\label{algo-nmax8}
  \removelatexerror
  \begin{algorithm}[H]
  \DontPrintSemicolon
  \KwIn{$q[1:N]$}
  \KwOut{$out$}
  \;
  NoisyVector $\gets []$\;
  \For{$i\gets 1$ \KwTo $N$}{
    NoisyVector[i] $\gets$ $\mathsf{Lap}^{+}(\frac{\epsilon}{2}, q[i])$
  }
  out $\gets$ $\Disc_{\seq}$(max(NoisyVector))\;
  \end{algorithm}
\end{subfigure}
\caption{Noisy Max Algorithms}
\label{fig:nmax}
\end{figure*}

\paragraph{Histogram Algorithms}
Histogram algorithms also target computing statistical measures on
queries in a differentially private manner. Algorithms referred to as
Hist1-2 here are Algorithms 9-10
in~\cite{DingWWZK18}. Algorithm~\subref{algo-nmax9}
and~\subref{algo-nmax10} are variants of noisy maximum, where we
return the histogram, instead of the maximum. Under the above
adjacency relation where $Q_1$ and $Q_2$ are adjacent if
$\abs{{Q_1}[i]-{Q_2}[i]}\leq 1$ for all $i$, both these variants are
not $\epsilon$-differentially private. However, if we consider an
alternative definition for the adjacency relation, where $Q_1$ and
$Q_2$ are adjacent iff $\sum_{i}\Big(\abs{{Q_1}[i]-{Q_2}[i]}\Big)\leq
1$, then~\subref{algo-nmax9} is $\epsilon$-differentially private
but~\subref{algo-nmax10} still is not. All experiments listed in
Section~\ref{sec:exp-res} for Algorithms~\subref{algo-nmax5}
and~\subref{algo-nmax6} were run using the second adjacency relation.
The pseudocode for these algorithms can be found in
Figure~\ref{fig:nhist}.

\begin{figure*}[!htb]
\RestyleAlgo{boxed}
\begin{subfigure}{0.45\textwidth}
  \renewcommand\thesubfigure{Hist1}
  \caption{Noisy Histogram}\label{algo-nmax9}
  \removelatexerror
  \begin{algorithm}[H]
  \DontPrintSemicolon
  \KwIn{$q[1:N]$}
  \KwOut{$out[1:N]$}
  \;
  NoisyVector $\gets []$\;
  \For{$i\gets 1$ \KwTo $N$}{
    NoisyVector[i] $\gets$ $\Lap{\epsilon, q[i]}$
  }
  out $\gets$ $\Disc_{\seq}$(NoisyVector)\;
  \end{algorithm}
\end{subfigure}
\begin{subfigure}{0.45\textwidth}
  \renewcommand\thesubfigure{Hist2}
  \caption{Noisy Histogram, Wrong Scale}\label{algo-nmax10}
  \removelatexerror
  \begin{algorithm}[H]
  \DontPrintSemicolon
  \KwIn{$q[1:N]$}
  \KwOut{$out[1:N]$}
  \;
  NoisyVector $\gets []$\;
  \For{$i\gets 1$ \KwTo $N$}{
    NoisyVector[i] $\gets$ $\Lap{\frac{1}{\epsilon}, q[i]}$
  }
  out $\gets$ $\Disc_{\seq}$(NoisyVector)\;
  \end{algorithm}
\end{subfigure}
\caption{Noisy Histogram Algorithms}
\label{fig:nhist}
\end{figure*}

\paragraph{Randomized Response}
All the previous algorithms use the Laplace mechanism. Randomized
Response~\cite{DNRRV09}, on the other hand, uses discrete
probabilities. In this algorithm (henceforth called  \subref{algo-rr1}), given a set of Boolean input
queries, we flip each input query with a probability of
$\frac{\euler^\epsilon-1}{2}$ and output the resulting outcome
. We
also consider a non-private version (called  \subref{algo-rr2})  where the input query is flipped
with probability $\frac{1-\epsilon}{2}$ . The pseudocodes can be found in Figure~\ref{fig:rr}.

\begin{figure*}[!htb]
 \RestyleAlgo{boxed}
\begin{subfigure}{0.48\textwidth}
  \renewcommand\thesubfigure{Rand1}
  \caption{Differentially Private Randomized Response}\label{algo-rr1}
  \removelatexerror
  \begin{algorithm}[H]
  \DontPrintSemicolon
  \KwIn{$q[1:N]$}
  \KwOut{$out[1:N]$}
  \;
 
  \For{$i\gets 1$ \KwTo $N$}{
    out[i] $\gets
    \begin{cases}
        q[i] & \mbox{with prob} = \frac{\euler^\epsilon}{1+\euler^\epsilon}\\
        \neg q[i] & \mbox{with prob} = \frac{1}{1+\euler^\epsilon}
    \end{cases}$\;
  }
  \end{algorithm}
\end{subfigure}
\begin{subfigure}{0.48\textwidth}
  \renewcommand\thesubfigure{Rand2}
  \caption{Non-Differentially Private Randomized Response}\label{algo-rr2}
  \removelatexerror
  \begin{algorithm}[H]
  \DontPrintSemicolon
  \KwIn{$q[1:N]$}
  \KwOut{$out[1:N]$}
  \;
 
  \For{$i\gets 1$ \KwTo $N$}{
    out[i] $\gets
    \begin{cases}
        q[i] & \mbox{with prob} = \frac{1+\epsilon}{2}\\
        \neg q[i] & \mbox{with prob} = \frac{1-\epsilon}{2}
    \end{cases}$\;
  }
  \end{algorithm}
\end{subfigure}
\caption{Randomized Response Algorithms}
\label{fig:rr}
\end{figure*}

\paragraph{Sparse}
Sparse is a variant of SVT that is discussed in~\cite{DR14}. Our
reason for considering this example is to demonstrate our tool's
ability to handle $(\epsilon,\delta)$-differential privacy (see
Section~\ref{sec:sparse}).  Pseudocode for this algorithm is provided
in Section~\ref{sec:sparse}.

\subsection{Tool Design}
\label{sec:working}

Given a program and an adjacency relation, {\tool} outputs $\true$ if
the program is differentially private and outputs a counter-example if
it is not. The tool works in two phases. In the first phase, the tool
parses the program, computes symbolic expressions that capture the
output distribution, and identify inequalities that must hold for
differential privacy. The symbolic expressions for the probability
computation, and the logical constraints that must hold, are written
in a Wolfram Mathematica\textregistered script. In the second phase,
Mathematica is run to perform the symbolic computations and check the
results.

The computation of the output distribution proceeds in a manner
consistent with the decision procedure outlined in the proof of
Theorem~\ref{thm:decidability}. Recall that the parametrized DTMC
semantics, the state tracks constraints that must hold between
different real variables. These constraints can be tracked by
maintaining a partial order between the variables. One of the
engineering challenges we experienced was in the computation of the
probability of the partial order holding, given the parameters used
during sampling. The ``Probability[]'' command in Mathematica was very
slow and inefficient. Instead we decided to convert the partial order
into a set of total orders, and compute the probability of each total
order through integration.

For example, to compute the probability of $x_1<x_2<x_3...<x_n$, where
variable $x_i$ has p.d.f $D_i$, we would first compute the probability
$P(x_n>x) = \int_x^{\infty} D_n(y)dy$. We then compute the
probabilities $P(x_n>x_{n-1}>x) = \int_x^{\infty}
P(x_n>y)D_{n-1}(y)dy$, $P(x_n>x_{n-1}>x_{n-2}>x) = \int_x^{\infty}
P(x_n>x_{n-1}>y)D_{n-2}(y)dy$ and so on. Once we have computed
$P(x_n>x_{n-1}>...>x_1>x)$, we can compute
$P(x_n>x_{n-1}>...>x_1)=Lim_{x\rightarrow
-\infty}P(x_n>x_{n-1}>...>x_1>x)$.
Additionally, we try to optimize the above process by splitting the
partial order into connected components and computed probability for
each component. We also deal with constant assignments to real
variables by slightly modifying the integration method.

\subsection{Experimental Results}
\label{sec:exp-res}

We ran all the experiments on an octa-core Intel\textregistered Core
i7-8550U @ 1.8gHz CPU with 8GB memory. The tool is implemented in C++
and uses Wolfram Mathematica\textregistered. As mentioned in
Section \ref{sec:working}, the tool works in two phases --- in the
first phase, a Mathematica script is produced with commands for all
the output probability computations and the subsequent inequality
checks and in the second phase, the generated script is run on
Mathematica. In all the following tables, we refer the times of the
Script Generation Phase (i.e. Phase 1) as T1 and that of the Script
Validation Phase (i.e. Phase 2) as T2.

Unless stated otherwise, all the experiments were run with the
parameters $c=1$, $\Delta=1$ and discretization parameter
$\seq=(-1<0<1)$ wherever applicable. The range of input query values
was $\sdom=\{-1, 0, 1\}$ in all the experiments. The running times in
all experiments were averaged over 3 runs of the tool.

\begin{table*}
\begin{minipage}[b]{0.47\linewidth}\centering
  \begin{tabular}{|c|N|N|}
    \hline
    Algorithm & Runtime (\parta/\partb) & $\epsilon$-Diff. Private\\
    \hline
    \subref{algo-svt1} & 0s/825s & \cmark\\
    \subref{algo-svt2} & 0s/768s & \cmark\\
    \subref{algo-svt3} & 0s/3816s & \cmark\\
    \subref{algo-svt4} & 0s/269s & \xmark\\
    \subref{algo-svt5} & 0s/2s & \xmark\\
    \subref{algo-svt6} & 0s/661s & \xmark\\
    \subref{algo-nmax5} & 0s/197s & \cmark\\
    \subref{algo-nmax6} & 0s/59s & \cmark\\
    \subref{algo-nmax7} & 0s/310s & \xmark\\
    \subref{algo-nmax8} & 1s/58s & \xmark\\
    \subref{algo-nmax9} & 0s/1450s & \cmark\\ 
    \subref{algo-nmax10} & 0s/55s & \xmark\\
      \subref{algo-rr1}  & 0s/0s & \xmark\\
     \subref{algo-rr2} & 0s/0s & \xmark\\
    \hline
  \end{tabular}
  \caption{Runtime for 3 queries for each algorithm searching over adjacency pairs and all $\epsilon$>0, with parameters being [c=1,~$\Delta$=1,~$\sdom$=\{-1,0,1\},~$\seq=(-1<0<1)$]. For SVT, we also have $T$=0.}
  \label{tab:q3runtime}
\end{minipage}\hfill
\begin{minipage}[b]{0.47\linewidth}\centering
  \begin{tabular}{|c|c|c|N|N|}
    \hline
    |Q| & c & $\epsilon$ & \multicolumn{2}{c|}{\begin{tabular}{c|c}\multicolumn{2}{c}{Runtime (\parta/\partb)}\\\hline Fixed $\epsilon$ & General\\\end{tabular}}\\
    \hline
    1 & 1 & 1.0 & 0s/7s & 0s/16s\\
    1 & 1 & 0.5 & 0s/8s & 0s/16s\\
    2 & 1 & 1.0 & 0s/43s & 0s/113s\\
    2 & 1 & 0.5 & 0s/46s & 0s/113s\\
    2 & 2 & 1.0 & 0s/95s & 0s/155s\\
    2 & 2 & 0.5 & 0s/113s & 0s/155s\\
    3 & 1 & 1.0 & 0s/307s & 0s/825s\\
    3 & 1 & 0.5 & 0s/265s & 0s/825s\\
    3 & 2 & 1.0 & 0s/541s & 0s/1202s\\
    3 & 2 & 0.5 & 0s/572s & 0s/1202s\\
    4 & 1 & 1.0 & 0s/1772s & 0s/4727s\\
    4 & 1 & 0.5 & 0s/1832s & 0s/4727s\\
    4 & 2 & 1.0 & 1s/2904s & 0s/6715s\\
    4 & 2 & 0.5 & 1s/3295s & 0s/6715s\\
    \hline
  \end{tabular}
  \caption{Runtimes of \subref{algo-svt1} over different query length and counts, searching over all adjacency pairs and fixed $\epsilon$, with parameters being [$\Delta$=1,~$T$=0,~$\sdom$=\{-1,0,1\}]. }
  \label{tab:eps}
\end{minipage}
\end{table*}

Table~\ref{tab:q3runtime} shows the runtime of our tool for all the
listed algorithms with 3 queries. We chose to use 3 queries because
counter-examples for most of the programs which were not differentially
private could be found with 3 queries; the only exception
being~\subref{algo-svt3}. Majority of the time is taken for running
the Mathematica code. We also observed that most of the time spent by
Mathematica was in computing the output probability; the time to
perform the inequality checks for adjacent inputs was relatively
smaller. Consequently, programs which do not use real variables are
much faster to run. Results in the table also show that the time taken
for disproving Differential Privacy is lower than the time for proving
Differential Privacy on average. This is because the tool terminates
on finding a counter-example. On the other hand, to prove differential
privacy the tool has to check all inequalities.

\begin{table*}[!t]
\scalebox{0.8}{
  \begin{tabular}{|c|c|c|c|c|c|N|}
    \hline
    Algo & |Q| & Output & Input 1 & Input 2 & $\epsilon$ 
    & Runtime (\parta/\partb) \\
    \hline
    \subref{algo-svt3} & 5 & [$\bot$ $\bot$ $\bot$ $\bot$ 0], $\seq=(0<1)$ & [-1 -1 -1 -1 -1] & [0 0 0 0 0] & 27 & 18s/5042s\\
    \subref{algo-svt4} & 2 & [$\bot$ $\top$] & [-1 0] & [0 -1] & 27/50 & 0s/81s\\
    \subref{algo-svt5} & 2 & [$\bot$ $\top$] & [-1 0] & [-1 -1] & 27 & 0s/2s\\
    \subref{algo-svt6} & 3 & [$\bot$ $\bot$ $\top$] & [-1 -1 0] & [0 0 -1] & 67/92 & 0s/661s\\
    \subref{algo-nmax7} & 3 & -1, $\seq=(-1<0<1)$ & [-1 -1 -1] & [0 0 0] & 27 & 0s/310s\\
    \subref{algo-nmax8} & 1 & 0, $\seq=(-1<0<1)$ & [-1] & [0] & 27 & 0s/2s\\
    \subref{algo-nmax10} & 1 & [-1], $\seq=(-1<0<1)$ & [-1] & [0] & 9/34 & 0s/3s\\
      \subref{algo-rr2} & 1 & [$\bot$] & [$\bot$] & [$\top$] & 9/34 & 0s/0s\\
    \hline
  \end{tabular}
  }
  \caption{Smallest Counter-example found for each non-differentially private algorithm, searching over all adj. pairs and $\epsilon>0$, with parameters being [c=1,~$\Delta$=1,~$\sdom$=\{-1,0,1\}]}
  \label{tab:counter}
\end{table*}

Table~\ref{tab:counter} lists the smallest counter-example found for
each non differentially private algorithm. Given a program and an
adjacency relation, the tool automatically finds an $\epsilon$, the
pair of adjacent inputs, and the output value that demonstrate the
violation of differential privacy. All four columns in the table were
output by the tool. Further, we observe that the counter-examples found
were much smaller, in number of queries, compared to those found
in~\cite{DingWWZK18}. For example, algorithms~\subref{algo-nmax7}
and \subref{algo-nmax8} counter-examples need just 3 and 1 queries
respectively, compared to the 5 queries required
in \cite{DingWWZK18}. Similarly, algorithm \subref{algo-svt5} has a
counter-example with just 2 queries, as compared to the 10 queries.

To study the performance of the tool as the number of queries
increases, we analyzed~\subref{algo-svt1} for various number of
queries. The running times along with the number of queries and the
value for $c$ is shown in Table~\ref{tab:query}. The table shows that
the tool can handle a reasonable number of queries.

In all the experiments so far, the value of $\epsilon$ was not
fixed. So {\tool} had to either prove privacy for all $\epsilon$ or
find an $\epsilon$ where privacy is violated. Many automated tools are
designed only to disprove differential privacy for a fixed
$\epsilon$. We tried the performance of the tool on SVT1 for a fixed
$\epsilon$. The results are reported in Table~\ref{tab:eps}. As can be
seen by comparing the numbers in Tables~\ref{tab:query}
and~\ref{tab:eps}, fixing $\epsilon$ makes the problem easier to
handle.

\begin{table*}[!t]
\begin{minipage}[b]{0.4\linewidth}\centering
  \begin{tabular}{|c|c|N|}
    \hline
    |Q| & c & Runtime (\parta/\partb) \\
    \hline
    1 & 1 & 0s/16s\\
    2 & 1 & 0s/113s\\
    2 & 2 & 0s/155s\\
    3 & 1 & 0s/825s\\
    3 & 2 & 0s/1202s\\
    4 & 1 & 0s/4727s\\
    4 & 2 & 0s/6715s\\
    \hline
  \end{tabular}
  \caption{Runtimes of \subref{algo-svt1} over different query length and counts, searching over all adjacency pairs and all $\epsilon$>0, with parameters being [$\Delta$=1,~$T$=0,~$\sdom$=\{-1,0,1\}]}
  \label{tab:query}
\end{minipage}\hfill
\begin{minipage}[b]{0.5\linewidth}\centering
  \begin{tabular}{|c|N|N|M|}
    \hline
    \#Queries & 1 Pair Runtime (\parta/\partb) & General Runtime (\parta/\partb) & $\epsilon$-Diff. Private\\
  \hline
    1 & 0s/15s & 0s/25s & \cmark\\
    2 & 0s/40s & 0s/192s & \cmark\\
    3 & 0s/100s & 0s/1562s & \cmark\\
    4 & 0s/199s & 1s/10515s & \cmark\\
    5 & 0s/141s & 18s/5042s & \xmark\\
    \hline
  \end{tabular}
  \caption{Runtimes of \subref{algo-svt3} over different query lengths, searching over a single adj. pair ([00...]$\sim$[11...]) and all $\epsilon>0$, with parameters being [$c$=1,~$T$=0,~$\Delta$=1,~$\sdom$=\{-1,0,1\},~$\seq$=(0<1)]}
  \label{tab:adj}
\end{minipage}
\end{table*}

Finally, we wanted to explore the scalability of our tool when we
checking differential privacy for a single pair of adjacent inputs.
In Table~\ref{tab:adj}, we have the results when a non differentially
private algorithm, namely~\subref{algo-svt3} was run with a single
adjacency pair ([00...]$\sim$[11...]), while varying number of
queries. We notice that the running times is significantly lower in
this case. Another interesting observation is that the time taken
for 5 queries is lower than the time for 4 queries. This is because
with 5 queries, the tool successfully finds a counter-example and
terminates before checking the remaining inequalities.

\subsection{($\epsilon, \delta$)-Differential Privacy}
\label{sec:sparse}

\tool~can also verify ($\epsilon, \delta$)-differential privacy.
Algorithm~\ref{fig:svtDelta} (taken from \cite{DR14}), referred to
henceforth as Sparse, was used to evaluate {\tool}'s performance in
this case. This algorithm has been manually proven to be
$(\frac{\epsilon}{2},\delta_\s{svt})$-differentially private for any
number of queries in~\cite{DR14} by using advanced composition
theorems.

\begin{figure}
 \RestyleAlgo{boxed}
\begin{minipage}{0.35\textwidth}
  \removelatexerror
  \begin{algorithm}[H]
  \DontPrintSemicolon
  \KwIn{$q[1:N]$}
  \KwOut{$out[1:N]$}
  $\sigma \gets \frac{\epsilon}{2\sqrt{32c\ln{\frac{1}{\delta_\s{svt}}}}}$\;
  \;
  $\rv_T \gets \Lap{ \sigma , T}$\;
  $count \gets 0$\;
  \For{$i\gets 1$ \KwTo $N$}
  {
    $\rv\gets \Lap{\frac \sigma {2} , q[i]}$\;
    $\bv\gets \rv \geq \rv_T$\;
    \uIf{$b$}{
      $out[i] \gets \top$, $\rv_T\gets \Lap{\sigma , T}$\;
      $count \gets count + 1$\;
      \If{$count\geq c$} {\exit} 
      }
    \Else{
      $out[i] \gets \bot$}
  }
  \caption{Sparse algorithm}
\label{fig:svtDelta}

  \end{algorithm}
\end{minipage}
\hspace*{0.40cm}
\begin{minipage}{0.55\textwidth}
\begin{table}[H]
  \begin{tabular}{|c|c|c|c|c|}
    \hline
    $c$ & $\delta_\s{svt}$ & $\delta$ & Runtime (\parta/\partb) & ($\frac \epsilon 2, \delta$)-Diff. Privacy \\
    \hline
    1 & $\euler^{-\frac 1 {32}}$ & 0 & 0s/48s & \xmark\\
    1 & $\euler^{-\frac 1 {32}}$ &$\euler^{-3}$ & 0s/142s & \xmark\\
    1 & $\euler^{-\frac 1 {32}}$  &$\euler^{-2.125}$ & 0s/146s & \xmark\\
    1 & $\euler^{-\frac 1 {32}}$  &$\euler^{-2}$ & 0s/161s & \cmark\\
    2 & $\euler^{-\frac 1 {64}}$  &0 & 0s/72s & \xmark\\
    2 & $\euler^{-\frac 1 {64}}$ &$\euler^{-3}$ & 0s/187s & \xmark\\
   2 & $\euler^{-\frac 1 {64}}$ &$\euler^{-2.5}$ & 0s/182s & \xmark\\
    2 & $\euler^{-\frac 1 {64}}$ &$\euler^{-2}$ & 0s/288s & \cmark\\

    \hline
  \end{tabular}
  \caption{\tool result for ($\frac \epsilon 2, \delta$)-Diff. Privacy of SPARSE (Algorithm \ref{fig:svtDelta}) with 3 queries, searching over all adj. pairs and $\epsilon>0$, with parameters being [$T$=0,~$\sdom=\{0,1\}$]
  }
  \label{tab:eps-del}
\end{table}
\end{minipage}
\end{figure}

When $c=1$ and $\delta_\s{svt} = \euler^{- \frac 1 {32}}$, this algorithm is identical to Algorithm
\subref{algo-svt1}, where parameters $c$ and $\Delta$ are replaced by
parameter $\sigma$. This algorithm is, therefore,
$\epsilon$-differentially private. Further, our tool proves that the
algorithm is not $\frac{\epsilon}{2}$-differentially private. 
Thanks to the advanced composition theorem, 
we can show that the resulting algorithm is $(\frac \epsilon 2, \euler^{-\frac 1 {32}})$-differentially private. The tool
also shows that for all $\epsilon>0$, the algorithm is $(\frac{\epsilon}{2},
\euler^{-2})$-differentially private for $c=1$ for queries of length 3 
with $\sdom=\set{0,1}$ and $T=0$ (observe that $\euler^{-\frac 1 {32}}>\euler^{-2}$).
Additionally, we get a counter-example for
$(\frac{\epsilon}{2}, \euler^{-2.125})$-differential privacy.

When $c=2$ and $\delta_\s{svt} = \euler^{- \frac 1 {64}}$, Sparse
differs from~\subref{algo-svt1} since in this case we also need to
choose $\rv_T$ again after outputting a $\top$. The resulting program
is $(\frac{\epsilon}{2},\euler^{-1/64})$-differentially private thanks
to the advanced composition theorem.  {\tool} confirms that for
queries of length 3, the resulting program is infact
$(\frac{\epsilon}{2},\euler^{-2})$-differentially private with
$\sdom=\set{0,1}$ and $T=0$.  Further, {\tool} also demonstrates that
the resulting program is not $(\frac{\epsilon}{2},\euler^{-2.5})$
differentially private.

Here we are able to check the correctness of Sparse automatically, for
values of $c=1,2$ and for the above given values of $\delta_{\s{svt}}$
and for all $\epsilon>0.$ To the best of our knowledge, our approach
is the first method to automatically check this. These results are
summarized in Table~\ref{tab:eps-del}.

\end{document}